\documentclass[english,11pt]{article}
\usepackage[T1]{fontenc}
\usepackage[utf8]{inputenc}
\usepackage[english]{babel}
\usepackage[pdftex]{graphicx}
\usepackage[width=\textwidth]{caption}
\usepackage{amsmath,amsopn,amssymb,bm, dsfont,mathrsfs,mathtools}
\usepackage[margin=1in]{geometry}
\usepackage{array,placeins,flafter, longtable}
\usepackage{verbatim}
\usepackage{color}
\usepackage{natbib}
\usepackage{lmodern}
\usepackage[hidelinks]{hyperref}
\usepackage{booktabs, threeparttable}
\usepackage{xr}
\usepackage{adjustbox}
\usepackage{booktabs}
\usepackage{float}
\usepackage{enumitem}
\usepackage{subcaption}

\usepackage{pdflscape} 

	\usepackage{natbib}
	\usepackage{longtable}

	\usepackage[framemethod=tikz]{mdframed}
	\usepackage{todonotes}
	\presetkeys{todonotes}{color=purple!10,inline}{} 

	\newcommand{\rf}[1]{\comment{Reference: \url{#1}}}
	\newcommand{\ind}[1]{\mathds{1}\left\{#1\right\}}

\allowdisplaybreaks

\newcommand{\convD}{\stackrel{d}{\longrightarrow}}
\newcommand{\convP}{\stackrel{P}{\longrightarrow}}

\newcommand{\abs}[1]{\left\vert#1\right\vert}

\newcommand{\Gn}{\mathbb{G}_{n_1}}

\newcommand{\Hn}{\mathbb{H}_{n_3}}

\newcommand{\Hninv}{\mathbb H_{n_3}^{-1}}

\newcommand{\Vn}{\mathbb{V}_{n_2}}

\newcommand{\I}{\text{I}}
\newcommand{\N}{\mathbb N}

\newcommand{\R}{\mathbb R}

\newcommand{\eps}{\varepsilon}

\newtheorem{assumption}{Assumption}
\newtheorem{lemma}{Lemma}

\newtheorem{theorem}{Theorem}
\newtheorem{proposition}{Proposition}

\newtheorem{assump}{Assumption}

\makeatletter
\renewcommand{\section}{\@startsection{section}{2}{0mm}{-1.5\baselineskip}{1\baselineskip}{\normalfont\large\bfseries}}
\renewcommand{\subsection}{\@startsection{subsection}{2}{0mm}{-1.2\baselineskip}{1\baselineskip}{\normalfont\normalsize\bfseries}}
\renewcommand{\subsubsection}{\@startsection{subsubsection}{3}{0mm}{-0.8\baselineskip}{0.4\baselineskip}{\normalfont\normalsize\itshape}}
\makeatother

\begin{document}

\title{Asymptotic Properties of Empirical Quantile-Based Estimators\thanks{Julien Chhor and Martin Mugnier gratefully acknowledge financial support from the Agence Nationale de la Recherche under the grant ANR-17-EURE-0010 (Investissements d'Avenir program). We would like to thank Abhimanyu Gupta, Mauricio Olivares, Johan Segers, Ingrid Van Keilegom, as well as seminar participants at KU Leuven, TSE, Universit\"at Bonn, University College London,  University of Essex, and Universit\"at Z\"urich for helpful comments. Timoth\'ee Bacchi and Virgile Lacombe have provided excellent research assistance.} \\  }

\author{Julien Chhor\thanks{Toulouse School of Economics, University of Toulouse Capitole, Toulouse, France, julien.chhor@tse-fr.eu.} \and Xavier D'Haultf\oe{}uille\thanks{CREST-ENSAE, xavier.dhaultfoeuille@ensae.fr.} \and J\'er\'emy L'Hour\thanks{CFM \& CREST-ENSAE, jeremy.l.hour@ensae.fr.} \and Martin Mugnier\thanks{Paris School of Economics, martin.mugnier@psemail.eu.}}

\maketitle

\begin{abstract}
\noindent 
We consider inference for parameters of the form $\theta_0 = E[F_Y^{-1}\circ F_Z(X)]$ for some variables $X$, $Y$ and $Z$. Such parameters appear, in particular, in the ``changes-in-changes'' model of \cite{AtheyImbens2006}. We first establish that $\widehat{\theta}$, a plug-in estimator of $\theta_0$, is root-$n$ consistent and asymptotically normal under weaker conditions than those previously available, allowing in particular for unbounded variables. Next, we propose a new estimator of the asymptotic variance of $\widehat{\theta}$ and show its consistency, also allowing for unbounded variables. Monte Carlo simulations suggest that the conditions for root-$n$ consistency and asymptotic normality are, in some sense, minimal. These simulations highlight that our variance estimator also leads to more accurate inference than some alternative approaches.
\end{abstract}

\noindent\textbf{JEL Classification:} C14, C21, C23.

\medskip

\noindent\textbf{Keywords:} changes-in-changes, asymptotic inference, panel data.

\newpage
\setcounter{page}{1}
\addtolength{\baselineskip}{0.5\baselineskip}

\section{Introduction}

Quantile-quantile transforms, namely objects of the kind $F^{-1}\circ G$ where $F$ and $G$ are two cumulative distribution functions, appear commonly in economics. In particular, they have been used to recover distributions of unobserved potential outcomes. 
Prominent examples include the ``changes-in-changes'' (CIC) causal inference model developed by \cite{AtheyImbens2006}, and nonparametric instrumental variable quantile regression, see in particular \cite{vuong2017counterfactual} and \cite{wuthrich2020comparison}. In this setup, average treatment effects involve estimands of the form $\theta_0 = E[F_Y^{-1}\circ F_Z(X)]$ for some variables $X$, $Y$ and $Z$. The aim of this paper is to study inference for such parameters.  

\medskip
\cite{AtheyImbens2006} show that under suitable conditions, a plug-in estimator $\widehat{\theta}$ of $\theta_0$ is asymptotically normal, and establish the consistency of an estimator of its asymptotic variance. However, their results rely on strong assumptions. 
Specifically, they assume that the three variables (a) have bounded support, (b) each admit a continuously differentiable density, and (c) that these densities are bounded from above and below on their support. 
Such assumptions are overly restrictive for many variables of interest including, for instance,  wages, prices or profits. 

\medskip
The goal of this paper is to obtain similar results under substantially weaker conditions. 
This is important for establishing that such methods remain applicable to key economic variables that may not satisfy Assumptions (a)--(c) above. 
To this end, we first establish asymptotic normality of $\widehat{\theta}$. 
The main difficulty is that standard tools are no longer applicable under these weaker conditions.  When variables are bounded and their densities are bounded from below, the functional $(F_Y,F_X,F_Z)\mapsto \int F_Y^{-1}\circ F_Z dF_X$ is Hadamard differentiable \citep[see the proof of Theorem 7 in][]{CdCXdH2017}. However, Hadamard differentiability fails otherwise, already because $F\mapsto \int gdF$ is not  continuous with respect to the supremum norm when $g$ is unbounded. Similarly, we cannot directly exploit results on L-statistics \citep[see, e.g.][Chapter 19]{ShorackWellner1986}, as these correspond to the simpler case for which both $F_X$ and $F_Z$ are known. Instead, we rely on several results on weighted and unweighted empirical and quantile processes, see in particular Chapter 2, Section 7 and Chapter 11 in \cite{ShorackWellner1986} and \cite{csorgo1986}. We also exploit auxiliary results, including (i) the fact that order statistics of uniforms and uniforms spacing follow beta distributions; (ii) known bounds for the mean absolute deviation of beta distributions.

\medskip
Another contribution of this paper is to establish that in a sense that we make precise below, some of the conditions we impose for root-$n$ consistency and asymptotic normality are necessary as well. 
This implies fundamental constraints on the scope of methods relying on quantile-quantile transforms, at least if inference is based upon asymptotic normality. In the changes-in-changes model, for instance, this means that for the average treatment effect to be root-$n$ consistent and asymptotically normal, the distributions of the pre-treatment period outcome of the control and treatment groups must exhibit sufficiently similar tail behavior. 
These conditions are stronger than what is needed for root-$n$ consistency and asymptotic normality of quantile treatment effects. Intuitively, this is because unlike (non-extremal) quantile treatment effects, the average treatment effect depends on the tails of potential outcomes, whose corresponding quantiles are less precisely estimated. 

\medskip
Our second main contribution is to propose a new estimator of $\sigma^2$, the asymptotic variance of $\widehat{\theta}$. The plug-in estimator of \cite{AtheyImbens2006} includes a density term in its denominator. As a result, the consistency of this estimator becomes unclear when the density takes arbitrarily small values. A possible solution would be to trim the estimator, but this would introduce additional tuning parameters. 

\medskip
Instead, we consider an alternative estimator $\widehat{\sigma}^2$ based on a new expression of the asymptotic variance, still involving a density $f_U$ (that of $U:=F_Z(X)$) but without any denominator term. We consider a kernel density estimator of $f_U(u)$ with a varying bandwidth proportional to $u(1-u)$. Thus, the bandwidth shrinks as $u\to 0$ or $u\to 1$, a feature that is key to handling a possible explosion of $f_U(u)$ as $u\to 0$ or $u\to 1$. We show consistency of $\widehat{\sigma}^2$ under a slight strengthening of the conditions we impose for asymptotic normality. Notably, our proof does not require uniform or even pointwise consistency of our kernel-density estimator. Our estimator may be of interest for estimating functionals of probability densities, beyond the particular functional we consider here. 

\medskip
Finally, we investigate the finite-sample behavior of $\widehat{\theta}$ and inference based on asymptotic normality and $\widehat{\sigma}$ through Monte Carlo simulations. Our results suggest in particular that when our conditions for asymptotic normality hold, our inference method is already accurate with sample sizes around 100. Our estimator $\widehat{\sigma}$ also seems to perform better than that originally proposed by \cite{AtheyImbens2006}. Finally, when our conditions for asymptotic normality are violated, the distribution of $\widehat{\theta}$ does not appear to be normal, and none of the inference methods we consider, including the bootstrap, performs well.

\paragraph{Related literature.} 
In their seminal work, \cite{AtheyImbens2006} derive the asymptotic normality of $\widehat{\theta}$ and propose a consistent variance estimator. 
As discussed above, our main contribution is to extend these results to allow for unbounded variables. We also delineate some restrictions that the variable distributions should satisfy for the asymptotic variance to exist. 

\medskip
\cite{sun2025debiased} establish asymptotic normality of a debiased and semiparametrically efficient changes-in-changes estimator that flexibly accommodates continuous covariates. While they accommodate covariates, their result holds under a high-level condition (see their Assumption~4(a)).  Our paper shows that establishing weak and low-level conditions for asymptotic normality is nontrivial, even in the absence of covariates. 

\medskip
In a concurrent and independent line of research, \cite{beare2026convergencedistributionppprocess} establish the convergence in distribution in $L^1([0,1])$ of a process of the form $\widehat{F} \circ \widehat{G}^{-1}$ under similar conditions to those assumed in the present paper. However, we consider here the asymptotic normality of a quantity of a different kind, namely $\int_0^1 [\widehat{F}_Y^{-1}\circ \widehat{F}_Z] d\widehat{F}_X$, which does not seem to be a direct consequence of~\cite{beare2026convergencedistributionppprocess}, as it involves 
an additional  source of randomness via $\widehat{F}_X$.

\medskip
Our estimator of the asymptotic variance relies on nonparametric kernel density estimates with a varying bandwidth. Such estimators have been studied in mathematical statistics \citep[e.g.,][]{Jones1990, TerrellScott1992,chhor2025local}, but we use this technique in a different context, where the focus is not the density itself but a functional of it. Our  contribution is to show that such estimators lead to consistent estimation of the functional of interest under mild smoothness conditions, allowing in particular for the density to diverge at the boundaries. This is achieved by letting the bandwidth shrink appropriately near the boundary of the support. 

\paragraph{Notation.}
For any increasing function $F$ on the real line, we denote by $F^{-1}$ its left-continuous generalized inverse, $F^{-1}(q) = \inf\{x\in \mathbb R:F(x)\geq q\}$ for $q \in (0,1]$. In particular, for any real-valued random variable $W$ with cumulative distribution function (cdf) $F_W$, $F_W^{-1}$ is the corresponding quantile function. We denote by $\widehat{F}_W$ and $\widehat F_W^{-1}$ the corresponding empirical cdf and quantile function, obtained from a sample $(W_i)_{i=1,\dotsc, n}$. For any  $(x,y) \in \R^2$, we denote by $x \land y$ and $x \lor y$ the minimum and maximum of $x$ and $y$, respectively. We let ${\rm B}(\cdot,\cdot)$  denote the beta function, i.e., for all $x,y>0$,  ${\rm B}(x,y) = \int_0^1t^{x-1}(1-t)^{y-1}dt$. We let Beta$(\alpha,\beta)$ denote a random variable with the beta distribution with parameters $(\alpha,\beta)\in(0,\infty)^2$. 

\paragraph{Organization.} Section~\ref{sec:inf_obs_ranks} provides the asymptotic normality result of $\widehat{\theta}$, the plug-in estimator of $\theta_0$, introduces our estimator $\widehat{\sigma}^2$ of the corresponding asymptotic variance and shows its consistency. Section \ref{sec:mc_sim} studies the finite-sample behavior of $\widehat{\theta}$ and compares our   estimator $\widehat{\sigma}^2$ with alternative ones.  All the proofs are in the appendix, while the supplementary appendix gathers additional lemmas.

\section{Theory}\label{sec:inf_obs_ranks}

\subsection{Asymptotic normality of the plug-in estimator}

As mentioned above, we seek to estimate $\theta_0=\int_0^1 F_Y^{-1}(u) dF_U(u)$, where $U$ is unobserved but satisfying $U=F_Z(X)$, whereas $X$ and $Z$ are observed; note that $\theta_0$ is well-defined under Assumption \ref{ass:smoothness1} below (see Lemma~\ref{lem:int_FY} in Appendix~\ref{sec:tech_lemmas}). We observe three samples, $(Y_i)_{i=1,\ldots,n_1}$,  $(X_i)_{i=1,\ldots,n_2}$ and $(Z_i)_{i=1,\ldots,n_3}$. We consider the following plug-in estimator of $\theta_0$:
\begin{align*}
    \widehat{\theta} := \frac{1}{n_2}\sum_{i=1}^{n_2} \widehat{F}_Y^{-1}(\widehat{F}_Z(X_i)), 
\end{align*}
where $\widehat F_Y^{-1}$ is extended to $[0,1]$ by defining $\widehat F_Y^{-1}(0)=Y_{(1)}$. 

\medskip
We prove below that $\widehat{\theta}$ is asymptotically normal under the following conditions. Hereafter, we let $N := \min(n_1, n_2, n_3)$.

\begin{assumption}[Sampling] 
\label{ass:sampling_est}
~
\begin{enumerate}[label=(\roman*)]
   \item \label{ass:sampling_est1} $(Y_i)_{i=1,\ldots,n_1}$, $(X_i)_{i=1,\ldots,n_2}$ and $(Z_i)_{i=1,\ldots,n_3}$ are three samples of i.i.d. variables with respective cdfs $F_Y$, $F_X$ and $F_Z$.  
   \item \label{ass:sampling_est2} $(Y_i)_{i=1,\ldots,n_1}$, $(X_i)_{i=1,\ldots,n_2}$, and $(Z_i)_{i=1,\ldots,n_3}$ are mutually independent.
   \item \label{ass:sampling_est3} 
   For each $k \in \{1,2,3\}$, there exists $\lambda_k\in[0,1]$ such that $N/n_k\to \lambda_k$ as $N\to \infty$. 
\end{enumerate}
\end{assumption}

\begin{assumption}[Smoothness]\label{ass:smoothness1}
~
\begin{enumerate}[label=(\roman*)]
    \item  \label{ass:smoothness11} $F_Z$ is absolutely continuous with respect to the Lebesgue measure with density $f_Z$ supported on $[\underline{z},\overline z]$ with $-\infty\le \underline z<\bar z\le \infty$.
    \item \label{ass:smoothness12} $F_Y$ is continuous 
    and there exist $d_1,d_2 >0$ and $C_Y >0$ such that for all $t\in(0,1)$:
\begin{equation}
    |F_Y^{-1}(t)|  \le  C_Y t^{-d_1} (1-t)^{-d_2}.
\label{eq:cond_quant_Y}
\end{equation}
    \item \label{ass:smoothness13} $F_U$ is absolutely continuous with respect to the Lebesgue measure, with continuous density $f_U$ supported on a subset of $[0,1]$. There exist $b_1,b_2>0$ and $C_U >0$ such that for all $u \in (0,1)$:
\begin{equation}
    f_U(u) \le  C_U u^{-b_1} (1-u)^{-b_2}.
\label{eq:controle_fU}
\end{equation}
    \item \label{ass:smoothness14} $b_1+d_1<1/2$ and $b_2+d_2<1/2$.
\end{enumerate}
\end{assumption}

While Assumption~\ref{ass:sampling_est}\ref{ass:sampling_est2} may be too stringent to cover panel-data versions of \cite{AtheyImbens2006}'s model, it is plausible in the context of independent repeated cross-sections; we discuss the panel case in Section \ref{sub:samp_split} below. Assumption~\ref{ass:smoothness1} imposes restrictions on the distributions of $X$, $Y$ and $Z$. First, their cdf must be continuous. Second, $F_Y$ and $f_U$ must satisfy tail restrictions. In particular, \eqref{eq:cond_quant_Y} holds under the following moment condition on $Y$:
\begin{lemma}[Lower-Level Conditions on $Y$] 
\label{lemma:MomentY}
Assume $E[\vert Y \vert^p] < \infty$ for $p > 1$, then \eqref{eq:cond_quant_Y} holds with $d_1=d_2=1/p$.
\end{lemma}

Given that $U=F_Z(X)$, and assuming that $F_X$ is differentiable, we have $f_U(u)=f_X(F_Z^{-1}(u))/f_Z(F_Z^{-1}(u))$. Hence, \eqref{eq:controle_fU} (together with the constraints on $(b_1,b_2)$ implied by Assumption~\ref{ass:smoothness1}\ref{ass:smoothness14}) imposes that the tails of $X$ cannot be much heavier than those of $Z$. In the context of the changes-in-changes model, $X$ and $Z$ correspond to the pre-treatment period outcome of the control and treatment group, respectively. Thus, \eqref{eq:controle_fU} limits how  different the distributions of the outcome in the two groups can be. To illustrate this, assume that $X  \sim Z/c$ for some $c>0$ and $f_Z(z)\asymp K \exp(-L|z|^\alpha)$ for some $K, L,\alpha>0$ as $z\to-\infty$ (the same reasoning applies if we consider $z\to\infty$). Then, we show in Appendix \ref{app:proof_ineq_scale} that \eqref{eq:controle_fU} implies
\begin{equation}
c \ge (1-b_1)^{1/\alpha}.
    \label{eq:restriction_c}
\end{equation}
Similarly, if the densities of $X$ and $Z$ have power-law tails $|x|^{-c\alpha-1}$ and $|x|^{-\alpha-1}$, respectively, for some $c, \alpha>0$, one can show that \eqref{eq:controle_fU} implies $c>1-b_1$. 

\medskip
Finally, Assumption~\ref{ass:smoothness1}\ref{ass:smoothness14} implies a trade-off on the tails of $Y$ and $U$: the fatter the tails on $Y$, the lighter those on $U$ should be.

\begin{theorem} \label{thm:an}
If Assumptions \ref{ass:sampling_est} and \ref{ass:smoothness1}  hold and $\min(\lambda_1,\lambda_3)>0$, then, as $N\to\infty$,
\begin{equation*}
    \sqrt{N}(\widehat{\theta}-\theta_0) \convD \mathcal{N}(0,\sigma^2),
\end{equation*}
where $\sigma^2  = \left[\lambda_1 + \lambda_3 \right]E[\eta^2] + \lambda_2 E[\eps^2]$, 
with $\eta := - \int_0^{1}  [\mathds{1}\{F_Y(Y)\leq t\} - t]f_U(t)\,dF_Y^{-1}(t)$ and $\eps:= -\int_0^1 [\mathds{1}\{U\leq t\} - F_U(t)]\, dF_Y^{-1}(t)$.
\end{theorem}

The proof of Theorem \ref{thm:an} is long and technical.  The main difficulty lies in showing that various remainder terms are, indeed, negligible. To this end, we exploit several empirical process results, such as the convergence of the supremum of the weighted empirical quantile process \citep[see in particular Corollary 4.3.1 in][]{csorgo1986}. We also establish several results on quantile-quantile transforms that may be of independent interest, see in particular Lemma~\ref{lem:HnGn}. Our proof also relies on the fact that order statistics of uniform distributions and uniform spacings follow beta distributions, allowing us to leverage  properties of such distributions. Finally, we handle the L-statistic term by relying on the characterization in \cite{hecker1976}, as standard results on L-statistics, such as those in Chapter 19 of \cite{ShorackWellner1986}, do not apply here.

\medskip
The terms $\lambda_1 E[\eta^2]$, $\lambda_2E[\eps^2]$ and $\lambda_3 E[\eta^2]$ in the asymptotic variance correspond to the contributions of the three samples. Specifically, $\lambda_1 E[\eta^2]$ corresponds to the contribution of the estimation of the cdf of $Y$. The term $\lambda_2 E[\eps^2]$ is due to the fact that even if $F_Z$ were known, we would still estimate the cdf of $U$ using the sample $(F_Z(X_i))_{i=1,\ldots,n_2}$. The third term arises due to the estimation of $F_Z$. Perhaps surprisingly, it turns out that if $n_1=n_3$, so that $\lambda_1=\lambda_3$, this contribution is equal to that of the estimation of the cdf of $Y$.\footnote{Though this is not apparent in the expressions of \cite{AtheyImbens2006}, some algebra show that their first and second variance terms $V^p$ and $V^q$ are in fact equal.} Even if the case $\lambda_3=0$ is not covered by the theorem, the proof of Theorem \ref{thm:an} shows that if $U$ is observed (namely, if $F_Z$ is known), the estimator is still asymptotically normal with the same variance as above but with $\lambda_3$ set to 0.

\medskip
To what extent is Assumption \ref{ass:smoothness1} necessary for the result? We argue that, in some sense, Assumption \ref{ass:smoothness1}\ref{ass:smoothness14} is sharp. To see this, assume that $F_Y^{-1}$ is differentiable and 
\begin{align}
f_U(u)F_Y^{-1}{}'(u) \ge \underline{C} u^{-b_1-d_1-1}(1-u)^{-b_2-d_2-1}, \label{eq:ineg_f_U_FY}
\end{align}
for some $\underline{C}>0$. This inequality implies that \eqref{eq:cond_quant_Y} and \eqref{eq:controle_fU} are essentially sharp. The following proposition establishes that if this is the case, then Assumption \ref{ass:smoothness1}\ref{ass:smoothness14} is necessary for $E[\eps^2+\eta^2]<\infty$ to hold. Hence, the restrictions on $X$, $Y$ and $Z$ mentioned above (and in particular that the distributions of $X$ and $Z$ must be sufficiently similar) are, to some extent, required to ensure $E[\eps^2+\eta^2]<\infty$, and are not due to limitations in the proof of Theorem \ref{thm:an}.

\begin{proposition}
    Suppose that $F_Y^{-1}$ is differentiable, \eqref{eq:ineg_f_U_FY} holds, $\eps$ and $\eta$ are well-defined and $E[\eps^2+\eta^2]<\infty$. Then $b_k+d_k<1/2$ for $k=1,2$.
\label{prop:necessary_cond}
\end{proposition}

In a simpler setup than ours, \cite{mason1992necessary} show the stronger result that under mild regularity conditions, L-statistics are root-$n$ consistent and asymptotically normal if and only if an integral similar to $E[\eta^2]$ is finite: see the condition $\sigma^2(0)<\infty$ in their Theorem 1.1. 
We could thus expect that here as well, $\widehat{\theta}$ is root-$n$ consistent and asymptotically normal if and only if $E[\eps^2+\eta^2]<\infty$; our simulations below provide further support for this conjecture. 

\subsection{Consistent estimation of the asymptotic variance}\label{sec:var_est_obs_ranks}

Recall that $\sigma^2=(\lambda_1+\lambda_3)E[\eta^2] + \lambda_2E[\eps^2]$, with $\eta := - \int_0^{1}  [\mathds{1}\{F_Y(Y)\leq t\} - t]f_U(t)\,dF_Y^{-1}(t)$ and $\eps:= -\int_0^1 [\mathds{1}\{U\leq t\} - F_U(t)]\, dF_Y^{-1}(t)$. Note that 
\begin{align*}
\eps &=  -\left(F_Y^{-1}(U) -\int F_Y^{-1}dF_U\right) = \theta_0 - F_Y^{-1}(U).
\end{align*}
Then, let $\widehat U_i:=\widehat F_Z(X_i)$ and let $\widehat \eps_i=\widehat \theta - \widehat F_Y^{-1}(\widehat U_i)$. We can simply estimate $E[\eps^2]$ by  the sample average of $(\widehat \eps_i^2)_{i=1,...,n_2}$. 

\medskip
The estimation of $E[\eta^2]$ is more challenging. A natural idea would be to consider a plug-in estimator based on the definition of $\eta$. Let us assume, as we do in Assumption \ref{ass:smoothness2} below, that $F_Y^{-1}$ is differentiable. Let also $P(y):=E[(U-\ind{F_Y(y)\le U})/f_Y(F_Y^{-1}(U)]$, so that $\eta=P(Y)$. Then, following \cite{AtheyImbens2006}, we could estimate $E[\eta^2]$ by the sample average of $(\widehat{P}^2(Y_i))_{i=1,...,n_1}$, with 
\begin{equation}
\widehat{P}(y):=\frac{1}{n_2}\sum_{i=1}^{n_2} \frac{\widehat U_i-\ind{\widehat F_Y(y)\le \widehat U_i}}{\widehat{f}_Y(\widehat F_Y^{-1}(\widehat{U}_i))}.
    \label{eq:def_P_AI}
\end{equation}
However, the inverse-density weighting appearing in \eqref{eq:def_P_AI} makes it difficult to establish the consistency of this estimator, at least under the weak conditions we impose on the distributions of $Y$ and $U$. We circumvent this difficulty by employing another estimator, based on the following lemma. 

\begin{lemma}\label{lem:change_var_sig1}
Suppose that $F_Y$ is continuous and $\int_0^1 [t(1-t)]^{1/2} f_U(t)\,dF_Y^{-1}(t)<\infty$. Then, 
$\eta := - \int_0^{1}  [\mathds{1}\{F_Y(Y)\leq t\} - t]f_U(t)\,dF_Y^{-1}(t)$ is well-defined almost surely and satisfies $E[\eta^2]<\infty$. Moreover, it holds that
    \begin{equation*}
    E[\eta^2]=\int_{\R^2} f_U(F_Y(y)) \; f_U(F_Y(y')) \;[F_Y(y)\wedge F_Y(y')] \;[\bar F_Y(y)\land \bar F_Y(y')]dydy'.
\end{equation*}
where $\bar F_Y=1-F_Y$.
\end{lemma}

Lemma~\ref{lem:change_var_sig1} shows that $E[\eta^2]$ can be expressed as an integral that does not involve any inverse-density weighting. We develop a plug-in estimator for $E[\eta^2]$ based on this integral.  We consider sample-splitting, as it allows us to bound the variance of the  estimator in our consistency proof, though the simulations below suggest that this is in fact unnecessary. To simplify notation, assume that $n_1$, $n_2$ and $n_3$ are multiples of $2$. Let $\widehat F_Z^{(1)}$, $\widehat F_Z^{(2)}$ denote two sample-splitting estimators of $F_Z$:
\begin{align*}
    \widehat F_Z^{(1)}(z)&=\frac{2}{n_3}\sum_{i=1}^{n_3/2}\ind{Z_i\leq z},  
    \\
    \widehat F_Z^{(2)}(z)&=\frac{2}{n_3}\sum_{i=n_3/2+1}^{n_3}\ind{Z_i\leq z}.
\end{align*}
Let $\widehat F_Y^{(1)},\widehat F_Y^{(2)}$ denote two sample-splitting estimators of $F_Y$  defined analogously using the sample $(Y_i)_{i=1,\ldots,n_1}$. Let also $\widehat f_U^{(1)}$ and $\widehat f_U^{(2)}$ denote two sample-splitting kernel density estimators of $f_U$, namely for all $u\in(0,1)$,
\begin{align}
    \widehat f_U^{(1)}(u)&=\frac{1}{n_2h_{n_2,u}}\sum_{i=1}^{n_2/2}\ind{\abs{\widehat U_i^{(1)}-u}\leq h_{n_2,u}}, \label{eq_def_hat_f_U_est} \\
    \widehat f_U^{(2)}(u)& =\frac{1}{n_2h_{n_2,u}}\sum_{i=n_2/2+1}^{n_2}\ind{\abs{\widehat U_i^{(2)}-u}\leq h_{n_2,u}}, \notag
\end{align}
where $h_{n_2,u}:=\eps_{n_2}u(1-u)$ for some positive deterministic sequence  
$(\eps_{n})_{n\ge 1}$ satisfying the following conditions:

\begin{assumption}[Bandwidth conditions]
For all $n\ge 1$, $\eps_{n}\leq 1/2$ and as $n\to\infty$, $\eps_{n}\to0$,  $n\eps_{n}\to\infty$, and $(\log(n)/\sqrt{n})^{1/2-b_j-d_j}=o(\eps_{n})$ for $j\in\{1,2\}$. 
\label{ass:bandwidth}
\end{assumption}

We suggest choosing $\eps_{n_2}:=1/\log(n_2)$, which satisfies these restrictions for any $b_1,b_2,d_1,d_2$ that verify the assumptions of Theorems~\ref{thm:var_est} below. For all $s,t\in[0,1]$, define $w(s,t)=(s\wedge t)(1-s\lor t) = (s \land t) (\bar s \land \bar t)$ where $\bar x = 1-x$ for any $x \in \R$. We let
\[
\widehat E\left[\eta^2\right] := \int_{\R^2}\widehat f_U^{(1)}\!\left(\widehat F_Y^{(1)}(y)\right)\,\widehat f_U^{(2)}\!\left(\widehat F_Y^{(2)}(y')\right)\,w\!\left(\widehat F_Y^{(1)}(y), \widehat F_Y^{(2)}(y')\right)dydy'.
\]
Two remarks are in order. First, one could instead combine the subsamples of $U$ and $Y$ differently, replacing, for instance, $f_U^{(1)}\!\left(\widehat F_Y^{(1)}(y)\right)$ by $f_U^{(1)}\!\left(\widehat F_Y^{(2)}(y)\right)$, and then average the two estimators. Second,  when a uniform kernel is used in $\widehat f_U^{(1)}$ and $\widehat f_U^{(2)}$,  $E\left[\eta^2\right]$ is a double integral of a step function that vanishes outside the compact interval $[Y_{(1)},Y_{(n_1)}]$ and has jumps at the $Y_{(i)}$. Hence, its  computation is straightforward.

\medskip
Finally, given that $\sigma^2=(\lambda_1+\lambda_3)E[\eta^2] + \lambda_2E[\eps^2]$, we estimate $\sigma^2$ by
\begin{align*}
    \widehat \sigma^2&:=\frac{N(n_1+n_3)}{n_1n_3}\int_{\R^2}\widehat f_U^{(1)}\!\left(\widehat F_Y^{(1)}(y)\right)\,\widehat f_U^{(2)}\!\left(\widehat F_Y^{(2)}(y')\right)\,w\!\left(\widehat F_Y^{(1)}(y), \widehat F_Y^{(2)}(y')\right)dydy' + \frac{N}{n_2^2} \sum_{i=1}^{n_2}\widehat \eps_i^2.
\end{align*}

We show in Theorem \ref{thm:var_est} below that $\widehat \sigma^2$ is consistent under Assumptions \ref{ass:sampling_est}, \ref{ass:bandwidth} and the following strengthening of Assumption~\ref{ass:smoothness1}:

\begin{assumption}[Smoothness]\label{ass:smoothness2}
~
\begin{enumerate}[label=(\roman*)]
    \item  \label{ass:smoothness21} $F_Z$ is absolutely continuous with respect to the Lebesgue measure with density $f_Z$ supported on $[\underline{z},\overline z]$ with $-\infty\le \underline z<\bar z\le \infty$.

    \item \label{ass:smoothness22} The support of $Y$ is $[\underline{y},\overline y]$ for some $-\infty\le \underline{y}<\bar y\le \infty$. Moreover, $F_Y^{-1}$ is differentiable on ($\underline{y},\overline{y})$ and there exists $c_Y>0$ such that for all $t\in(0,1)$:
    \begin{equation}
    \big(F_Y^{-1}\big){}'(t) \leq c_Yt^{-(1+d_1)}(1-t)^{-(1+d_2)}.
    \label{eq:cond_quant_Y2}
    \end{equation}
    \item \label{ass:smoothness23} The mapping $g:(0,1)^2\to\R_+$ defined as
    \[
        g(s,t):=(s\wedge t)^{2b_1}(\bar s \land \bar t)^{2b_2}f_U(s)f_U(t), \quad \forall (s,t)\in(0,1)^2,
    \]
    is $\beta$-H\"older for some $\beta\in(0,1]$, i.e., there exists $c_U>0$ such that
    \[
        \abs{g(s,t)-g(s',t')}\leq c_U\left[\abs{s-s'}^\beta+\abs{t-t'}^\beta\right] \quad \forall s,s',t,t'\in(0,1).
    \]
    \item \label{ass:smoothness24} $(b_1\lor b_2) +( d_1\lor d_2)<1/2$.
\end{enumerate}
\end{assumption}

Condition~\ref{ass:smoothness21} is the same as in  Assumption~\ref{ass:smoothness1}, while Condition \ref{ass:smoothness24} is a slight strenghthening of Assumption~\ref{ass:smoothness1}\ref{ass:smoothness14}. Condition~\ref{ass:smoothness22} is similar to, but stronger than, Assumption~\ref{ass:smoothness1}\ref{ass:smoothness12}. Condition~\ref{ass:smoothness23} is also a strengthening of Assumption~\ref{ass:smoothness1}(iii). To see this, let $\varphi(s):=s^{b_1} (1-s)^{b_2} f_U(s)$ and note that under Condition~\ref{ass:smoothness23}, we have, 
    \begin{align*}
        2c_U |s-t|^\beta \geq \left|g(s,s)  - g(t,t)\right| = |\varphi^2(s) - \varphi^2(t)| = |\varphi(s) - \varphi(t)| |\varphi(s) + \varphi(t)| \geq |\varphi(s) - \varphi(t)|^2,
    \end{align*}
    where the last inequality follows since $\varphi$ is nonnegative on $[0,1]$. Hence, $|\varphi(s) - \varphi(t)| \le (2c_U)^{1/2} |s-s'|^{\beta/2}$, which implies that $\varphi$ is $\beta/2$-Hölder, and thus bounded, on $[0,1]$. Hence, Assumption~\ref{ass:smoothness1}\ref{ass:smoothness13} holds under Assumption \ref{ass:smoothness2}\ref{ass:smoothness23}. 
    We now state our second main theorem regarding the consistency of our asymptotic variance estimator.

\begin{theorem}\label{thm:var_est}
If Assumptions \ref{ass:sampling_est}, \ref{ass:bandwidth} and \ref{ass:smoothness2} hold and $\min(\lambda_1,\lambda_2,\lambda_3)>0$, then, as $N\to\infty$,
$$\widehat \sigma^2 \convP \sigma^2.$$
\end{theorem}

\subsection{Panel data applications}\label{sub:samp_split}

While Assumption \ref{ass:sampling_est} may be reasonable in the repeated cross sections setting of \cite{AtheyImbens2006}'s model, it does not cover panel data applications where $Y$ and $Z$ ($Y_{00}$ and $Y_{01}$ in \cite{AtheyImbens2006}) are observed on the same units and are thus possibly correlated. We adapt Assumption~\ref{ass:sampling_est} as follows. Hereafter, we let $N:=\min(n_1,n_2)$.

\begin{assumption}[Panel data]\label{ass:sampling_panel}
    ~
\begin{enumerate}[label=(\roman*)]
    \item \label{ass:sampling_panel1} $(Y_i,Z_i)_{i=1,\ldots,n_1}$ and $(X_i)_{i=1,\ldots,n_2}$ are two samples of i.i.d.~variables with respective cdfs $F_{Y,Z}$ (with marginals $F_Y$ and $F_Z$) and $F_X$.
    \item \label{ass:sampling_panel2} $(Y_i,Z_i)_{i=1,\ldots, n_1}$ and $(X_{i})_{i=1,\ldots, n_2}$ are mutually independent. 
    \item \label{ass:sampling_panel3} For each $k \in \{1,2\}$, there exists $\lambda_k\in[0,1]$, such that $N/n_k\to \lambda_k$ as $N\to \infty$. 
\end{enumerate}
\end{assumption}

To handle such cases, we introduce the following sample-splitting estimator of $\theta_0$:
\[
\widetilde\theta:=\frac{1}{2}\left(\widehat\theta^{(1)}+\widehat\theta^{(2)}\right),
\]
where, assuming that $n_1$ and $n_2$ are multiples of $2$ to simplify notation,
\begin{align*}
    \widehat\theta^{(1)}&:=\frac{1}{n_2/2}\sum_{i=1}^{n_2/2}{\widehat F_Y^{(1)}{}}^{-1}\left(\widehat F_Z^{(2)}(X_i)\right), \\
    \widehat\theta^{(1)}&:=\frac{1}{n_2/2}\sum_{i=n_2/2+1}^{n_2}{\widehat F_Y^{(2)}{}}^{-1}\left(\widehat F_Z^{(1)}(X_i)\right).
\end{align*}

\begin{theorem}\label{thm:an_panel}
If Assumptions~\ref{ass:smoothness1} and \ref{ass:sampling_panel} hold and $\min(\lambda_1,\lambda_2)>0$, then, as $N\to\infty$, we have
\begin{equation*}
    \sqrt{N}(\widetilde{\theta}-\theta_0) \convD \mathcal{N}(0,\widetilde \sigma^2),
\end{equation*}
where $\widetilde \sigma^2$ is the quantity $\sigma^2$ defined in Theorem~\ref{thm:an} with $\lambda_3=\lambda_1$.

\end{theorem}

The proof follows directly from Theorem~\ref{thm:an}: by sample splitting, $\widehat\theta^{(1)}$ and $\widehat\theta^{(2)}$ are independent and by Theorem~\ref{thm:an}, as $N\to\infty$, we have
\[
 \sqrt{N}(\widehat{\theta}^{(j)}-\theta_0) \convD \mathcal{N}(0,2\widetilde \sigma^2), \quad j=1,2.
\]

Similarly, a consistent estimator of the asymptotic variance $\widetilde\sigma^2$ can be obtained by considering a sample-splitting estimator of $\widetilde \sigma^2$ based on four splits of the sample $(Y_i,Z_i)_{i=1,\ldots,n_1}$ to ensure that $\widehat f_U^{(1)}$, $\widehat f_U^{(2)}$, $\widehat F_Y^{(1)}$ and $\widehat F_Y^{(2)}$, which appear in $\widehat E[\eta^2]$, are independent. To simplify notation, assume that $n_1$ is a multiple of $4$. Let
\begin{align*}
    \widehat {\widetilde{\sigma}}^2
    &:=\frac{2N}{n_1}\int_{\R^2}\check f_U^{(1)}\!\left(\check F_Y^{(1)}(y)\right)\,\check f_U^{(2)}\!\left(\check F_Y^{(2)}(y')\right)\,w\!\left(\check F_Y^{(1)}(y), \check F_Y^{(2)}(y')\right)dydy' + \frac{N}{n_2^2} \sum_{i=1}^{n_2} \frac{\left(\check \eps_i^{(1)}\right)^2+\left(\check \eps_i^{(2)}\right)^2}{2},
\end{align*}
where, for $j\in\{1,2\}$, $\check \eps_i^{(j)}=\widetilde \theta - {\widehat F_Y^{(j)}{}}^{-1}\left(\widehat F_Z^{(3-j)}(X_i)\right)$ and
\begin{align*}
    \check f_U^{(1)}(u)&=\frac{1}{n_2h_{n_2,u}}\sum_{i=1}^{n_2/2}\ind{\abs{\check F_Z^{(1)}(X_i)-u}\leq h_{n_2,u}},  \\
    \check f_U^{(2)}(u)& =\frac{1}{n_2h_{n_2,u}}\sum_{i=n_2/2+1}^{n_2}\ind{\abs{\check F_Z^{(2)}(X_i)-u}\leq h_{n_2,u}},  \\
    \check F_Z^{(1)}(z)&=\frac{4}{n_1}\sum_{i=1}^{n_1/4}\ind{Z_i\leq z},  
    \\
    \check F_Z^{(2)}(z)&=\frac{4}{n_1}\sum_{i=n_1/4+1}^{n_1/2}\ind{Z_i\leq z}, \\
    \check F_Y^{(1)}(y)&=\frac{4}{n_1}\sum_{i=n_1/2+1}^{3n_1/4}\ind{Y_i\leq y},  
    \\
    \check F_Y^{(2)}(y)&=\frac{4}{n_1}\sum_{i=3n_1/4+1}^{n_1}\ind{Y_i\leq y}.
\end{align*}

\begin{theorem}\label{thm:var_est_panel}
    If Assumptions~\ref{ass:smoothness2} and \ref{ass:sampling_panel} hold and $\min(\lambda_1,\lambda_2)>0$, then, as $N\to\infty$, 
    $$\widehat{\widetilde{\sigma}}^2 \convP \widetilde \sigma^2.$$
\end{theorem}
Theorem~\ref{thm:var_est_panel} follows directly by the independence induced by sample-splitting together with Theorem~\ref{thm:var_est}. 

\section{Monte Carlo simulations}\label{sec:mc_sim}

In this section, we investigate the finite sample properties of asymptotic confidence intervals based on Theorems~\ref{thm:an}--\ref{thm:var_est}. We consider a data generating process that provides a tight control on our assumptions. The random variables $Y_1,\ldots,Y_{n_1}$ are independently and identically distributed (i.i.d.) such that $Y_i=F_Y^{-1}(W_i)$ with $W_i\sim$Uniform$(0,1)$ and
\[
F_Y^{-1}(t)=-t^{-d_1}+(1-t)^{-d_2}, \quad \forall t\in(0,1).
\]
We also assume that $Z_1,\ldots,Z_{n_3}$ are i.i.d. with distribution $\mathcal N(0,1)$, whose cdf is denoted as $\Phi$, and $X_1,\ldots,X_{n_2}$ are i.i.d.~such that $X_i=\Phi^{-1}(V_i)$ with $V_i\sim{\rm Beta}(1-b_1,1-b_2)$.
All the random variables are mutually independent, $U_i\sim$Beta$(1-b_1,1-b_2)$, and 
    \[
    \theta_0 = \frac{{\rm B}(1-b_1,1-b_2-d_2)-{\rm B}(1-b_1-d_1,1-b_2)}{{\rm B}(1-b_1,1-b_2)}.
    \]
We consider $b_1=d_1=0$ and 
$$(b_2,d_2)\in\{(0.05,0.05),(0.20, 0.05), (0.05, 0.20), (0.20, 0.20), (0.30, 0.30), (0.40, 0.40)\}.$$
The sample size $N=n_1=n_2=n_3$ varies in $\{100, 500,1,000, 10,000\}$. The number of replications is $10,000$. 

\medskip
We first study the behavior of $\widehat{\theta}$ depending on $(b_2,d_2)$. Recall that by Theorem \ref{thm:an} and since $b_1=d_1=0$, $\widehat{\theta}$ is root-$N$ consistent if $b_2+d_2<0.5$. On the other hand, our results do not cover the cases $b_2=d_2=0.3$ and $b_2=d_2=0.4$. Figure \ref{fig:iqr_convergence} displays the log of the interquartile range  of $\widehat{\theta}$, denoted as IQR$(\widehat{\theta}) := \text{quantile}_{\hat \theta}(0.75) - \text{quantile}_{\hat \theta}(0.25)$, as a function of $\log(N)$ for the different values $(b_2,d_2)$. We also plot straight lines with slope $-1/2$ starting from the initial point corresponding to $\log(N)=\log(100)$. Deviations from these straight lines indicate discrepancies from root-$N$ convergence. It appears that such deviations are moderate for $b_2+d_2< 0.5$, but are large otherwise, with slopes smaller than $-1/2$. 

\begin{figure}[H]
    \begin{center}
    \includegraphics[width=0.75\textwidth]{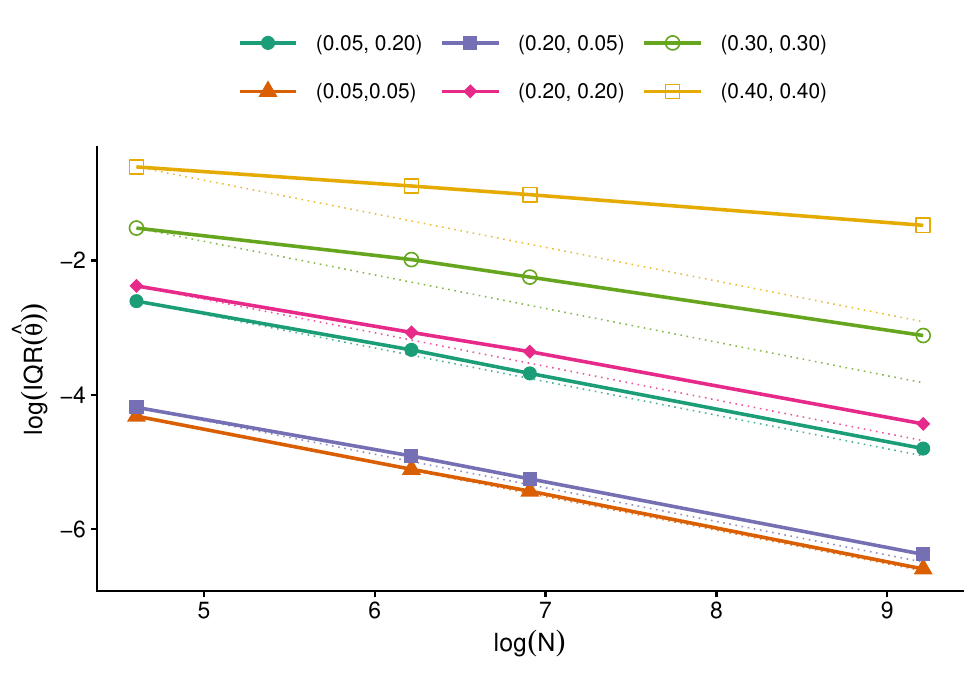}
    \end{center}
    {\small ~\\[-0.5cm] Notes: For each color, the dotted line is the straight line with slope $-1/2$ passing through the initial point on the corresponding solid line at $\log(N)=\log(100)$. Here, we take $b_1 = d_1 = 0$.}
    \caption{Log of IQR$(\widehat{\theta})$ as a function of $\log(N)$ for various $(b_2,d_2)$.}
\label{fig:iqr_convergence}
\end{figure}

Next, we investigate in Figure \ref{fig:true_dist_theta_hat} how close the distribution of  $\sqrt{N}(\widehat{\theta}-\theta_0)/\sigma$ is from a standard normal distribution. We consider both $(b_2,d_2)=(0.2,0.2)$ and $(b_2,d_2)=(0.3,0.3)$. Because $E[\eta^2+\eps^2]=\infty$ here when $b_2+d_2>0.5$, we redefine, with a slight abuse of notation, $\sigma$ as IQR$(\sqrt{N}(\widehat{\theta}-\theta_0))/1.349$, so that this is object is well-defined even with $(b_2,d_2)=(0.3,0.3)$. Note also that if $b_N(\widehat{\theta}-\theta_0)\convD \mathcal{N}(0,\widetilde{\sigma}^2)$ for some diverging sequence $(b_N)_N$ and some $\widetilde{\sigma}^2>0$, $\sqrt{N}(\widehat{\theta}-\theta_0)/\sigma$ would still tend to a standard normal distribution.  

\medskip
Again, we observe a close match between the distribution of  $\sqrt{N}(\widehat{\theta}-\theta_0)/\sigma$ and that of a standard normal distribution when $(b_2,d_2)=(0.2,0.2)$. When $(b_2,d_2)=(0.3,0.3)$, on the other hand, the discrepancy between the two distributions remains important even with $N=10,000$, the distribution of  $\sqrt{N}(\widehat{\theta}-\theta_0)/\sigma$ being substantially left-skewed.

\begin{figure}[H]
  \begin{subfigure}[b]{0.47\textwidth}
    \centering
    \includegraphics[width=\linewidth]{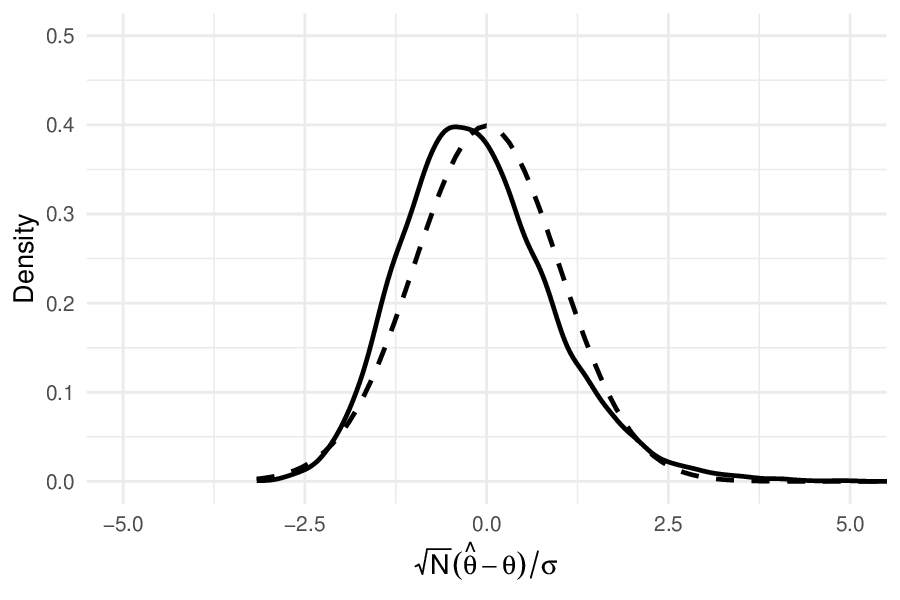}
    \caption{$(b_2,d_2)=(0.20,0.20), \, N=1,000$.}
  \end{subfigure}
  \hspace{0.3cm}
  \begin{subfigure}[b]{0.47\textwidth}
    \centering
    \includegraphics[width=\linewidth]{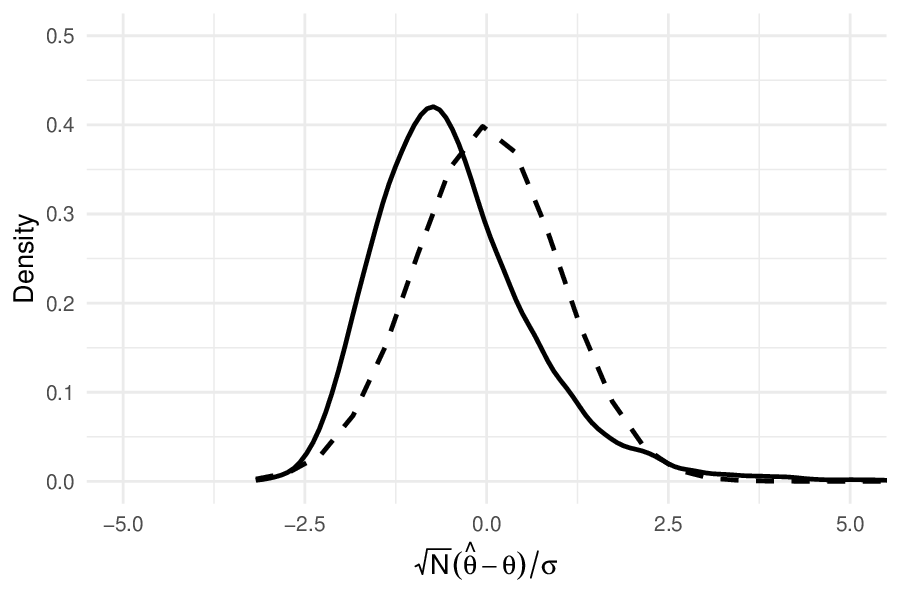}
    \caption{$(b_2,d_2)=(0.30,0.30), \, N=1,000$.}
  \end{subfigure}
\begin{subfigure}[b]{0.47\textwidth}
    \centering
    \includegraphics[width=\linewidth]{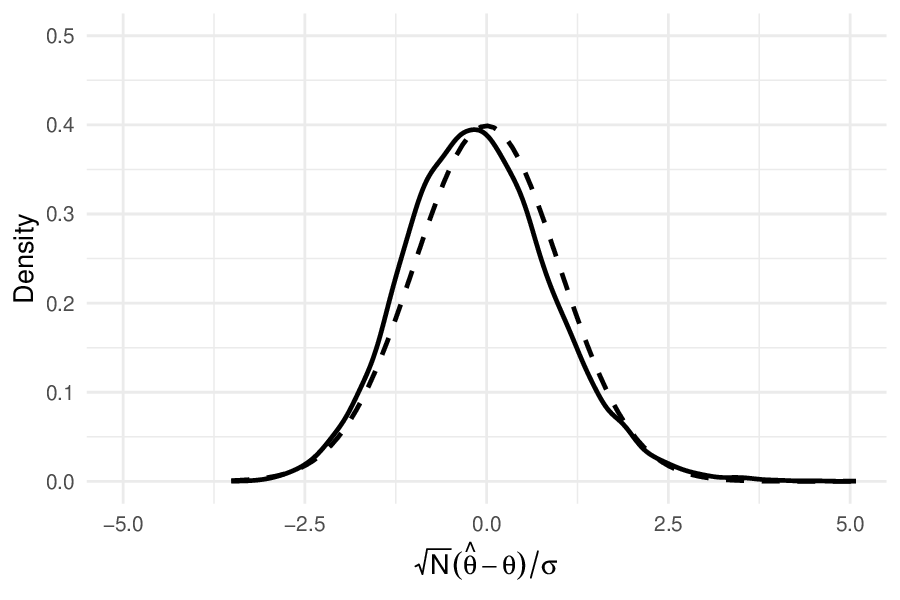}
    \caption{$(b_2,d_2)=(0.20,0.20), \, N=10,000$.}
  \end{subfigure}
  \hspace{0.3cm}
  \begin{subfigure}[b]{0.47\textwidth}
    \centering
    \includegraphics[width=\linewidth]{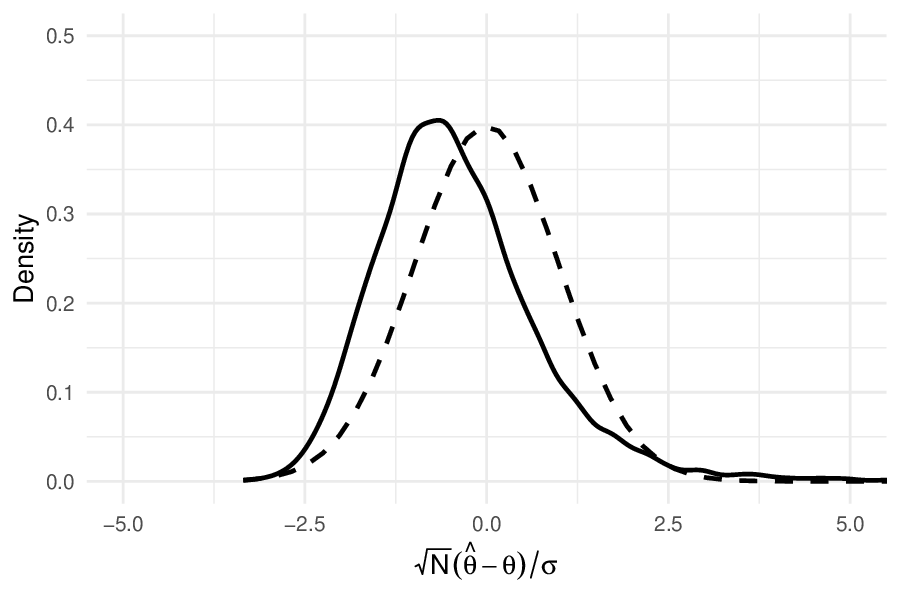}
    \caption{$(b_2,d_2)=(0.30,0.30), \, N=10,000$.}
  \end{subfigure}
  {\small ~\\[0.3cm] Notes: Here only, $\sigma=$IQR$(\sqrt{N}(\widehat{\theta}-\theta_0))/1.349$.}
  \caption{Empirical distribution of $\sqrt{N}(\widehat{\theta}-\theta_0)/\sigma$ (solid) v.s.~standard normal density (dotted).}
  \label{fig:true_dist_theta_hat}
\end{figure}

We now turn to inference. We consider six confidence intervals. The first five are based on asymptotic normality and different variance estimators, whereas the last relies on the bootstrap distribution. The first variance estimator (Split column) is $\widehat\sigma^2$, which is consistent for $\sigma^2$ under the assumptions of Theorem~\ref{thm:var_est}. As suggested in Section~\ref{sec:inf_obs_ranks}, we let $\eps_{n_2}=1/\log(n_2)$. The second variance estimator (No Split column) is based on a variant of $\widehat \sigma^2$ without sample splitting. The third estimator (Unif column) is based on a variant of $\widehat \sigma^2$ without sample splitting and where $h_{n_2,u}:=\eps_{n_2}$.  The fourth variance estimator (AI column) is that of \cite{AtheyImbens2006}. It estimates $E[\eta^2]$ by the sample average of $(\widehat{P}^2(Y_i))_{i=1,...,n_1}$, with $\widehat{P}$ defined in \eqref{eq:def_P_AI}. As in the simulations of \cite{AtheyImbens2006}, the estimator of $f_Y$ appearing in $\widehat{P}$ is a kernel (Epanechnikov)  estimator, with bandwidth equal $h_{n_1}=1.06n_1^{-1/5}/\widehat{\text{sd}}_Y$, where $\widehat{\text{sd}}_Y$ denotes the empirical standard deviation of $Y$. The fifth variance estimator (BSE) is based on the bootstrap (with $1,000$ random draws). The last column (BPC) reports the [0.025,0.975] percentile bootstrap confidence interval, based on 1,000 bootstrap samples.

\medskip
Table~\ref{tab:CI_len_est} reports the coverage rates and average lengths of the six confidence intervals. It shows that when $b_2+d_2<1/2$, all confidence intervals have coverage rates close to their nominal level as the sample size increases. The two confidence intervals whose coverage rate is closest to 95\% are those based on the percentile bootstrap, and ours without sample-splitting. For $N\in\{500,1,000,10,000\}$, their coverage rates is always between 0.93 and 0.96. Even with $N=100$, their coverage is always greater than or equal to 0.89. This suggests that sample-splitting is not needed for consistency of the variance estimator, and that in fact it may slightly worsen its finite sample properties. The confidence interval based on asymptotic normality and bootstrap standard errors also performs well. Table~\ref{tab:CI_len_est} also suggests that the AI estimator may be consistent, though the coverage of the corresponding confidence interval is systematically slightly below that of the other confidence intervals, except that based on our asymptotic variance estimator but using a constant bandwidth. The coverage of this latter confidence interval does not improve much with $N$ when $b_2=d_2=0.2$ and remains around 0.85. The results on this confidence interval underline the importance of allowing for a varying bandwidth when estimating the density of $U$.

\medskip
The cases $b_2=d_2=0.3$ and $b_2=d_2=0.4$ are in line with Figures \ref{fig:iqr_convergence} and \ref{fig:true_dist_theta_hat}. In such cases, the coverage rates are well below 0.95, and it is unclear whether the coverage of one of the six confidence intervals converges to this level. Specifically, the last four confidence intervals, including those based on the bootstrap, do not display any improvement when $b_2=d_2=0.40$. The coverage of our confidence interval does improve, but still only reaches 0.63 for $N=10,000$. Coverage is better for $b_2=d_2=0.30$, but even in this case distortion remains significant for $N=10,000$. Again, this suggests that our results are sharp at least in terms of the conditions on $(b_2,d_2)$.

\begin{table}[H]
\centering
\begin{threeparttable}
\resizebox{\textwidth}{!}{%
\begin{tabular}{lcccccc|cccccc}
\toprule
$(b_2,d_2)$ &
Split & No split & Unif & AI & BSE & BPC &
Split & No split & Unif & AI & BSE & BPC \\
\midrule
& \multicolumn{6}{c|}{$N=100$} &
\multicolumn{6}{c}{$N=500$} \\
\cmidrule(lr){2-7}\cmidrule(lr){8-13}
$(0.05, 0.05)$
& 0.92 & 0.93 & 0.90 & 0.92 & 0.93 & 0.94
& 0.94 & 0.95 & 0.92 & 0.94 & 0.95 & 0.95 \\
& [0.04] & [0.04] & [0.03] & [0.04] & [0.04] & [0.04]
& [0.02] & [0.02] & [0.02] & [0.02] & [0.02] & [0.02] \\
\addlinespace
$(0.20, 0.05)$
& 0.90 & 0.92 & 0.87 & 0.89 & 0.91 & 0.93
& 0.94 & 0.95 & 0.89 & 0.92 & 0.94 & 0.94 \\
& [0.04] & [0.05] & [0.04] & [0.04] & [0.05] & [0.05]
& [0.02] & [0.02] & [0.02] & [0.02] & [0.02] & [0.02] \\
\addlinespace
$(0.05, 0.20)$
& 0.90 & 0.92 & 0.88 & 0.90 & 0.92 & 0.94
& 0.94 & 0.95 & 0.90 & 0.92 & 0.94 & 0.95 \\
& [0.22] & [0.27] & [0.19] & [0.20] & [0.23] & [0.23]
& [0.11] & [0.13] & [0.09] & [0.10] & [0.11] & [0.11] \\
\addlinespace
$(0.20, 0.20)$
& 0.86 & 0.89 & 0.82 & 0.84 & 0.88 & 0.91
& 0.92 & 0.93 & 0.84 & 0.89 & 0.92 & 0.93 \\
& [0.27] & [0.36] & [0.22] & [0.23] & [0.28] & [0.28]
& [0.15] & [0.18] & [0.11] & [0.12] & [0.14] & [0.14] \\
\addlinespace
$(0.30, 0.30)$
& 0.74 & 0.79 & 0.66 & 0.69 & 0.77 & 0.82
& 0.82 & 0.85 & 0.65 & 0.73 & 0.82 & 0.85 \\
& [0.65] & [0.94] & [0.48] & [0.51] & [0.69] & [0.67]
& [0.45] & [0.64] & [0.26] & [0.30] & [0.41] & [0.40] \\
\addlinespace
$(0.40, 0.40)$
& 0.38 & 0.49 & 0.28 & 0.30 & 0.42 & 0.47
& 0.46 & 0.55 & 0.23 & 0.28 & 0.44 & 0.49 \\
& [1.51] & [2.55] & [1.05] & [1.11] & [1.73] & [1.63]
& [1.31] & [2.34] & [0.62] & [0.73] & [1.28] & [1.21] \\
\addlinespace
\midrule
& \multicolumn{6}{c|}{$N=1,000$} &
\multicolumn{6}{c}{$N=10,000$} \\
\cmidrule(lr){2-7}\cmidrule(lr){8-13}
$(0.05, 0.05)$
& 0.95 & 0.95 & 0.92 & 0.94 & 0.95 & 0.95
& 0.95 & 0.95 & 0.93 & 0.95 & 0.95 & 0.95 \\
& [0.01] & [0.01] & [0.01] & [0.01] & [0.01] & [0.01]
& [0.00] & [0.00] & [0.00] & [0.00] & [0.00] & [0.00] \\
\addlinespace
$(0.20, 0.05)$
& 0.95 & 0.96 & 0.90 & 0.93 & 0.95 & 0.95
& 0.96 & 0.96 & 0.90 & 0.94 & 0.95 & 0.95 \\
& [0.02] & [0.02] & [0.01] & [0.01] & [0.02] & [0.02]
& [0.01] & [0.01] & [0.00] & [0.00] & [0.00] & [0.00] \\
\addlinespace
$(0.05, 0.20)$
& 0.95 & 0.95 & 0.90 & 0.93 & 0.95 & 0.95
& 0.96 & 0.96 & 0.91 & 0.94 & 0.95 & 0.95 \\
& [0.08] & [0.09] & [0.06] & [0.07] & [0.08] & [0.08]
& [0.03] & [0.03] & [0.02] & [0.02] & [0.02] & [0.02] \\
\addlinespace
$(0.20, 0.20)$
& 0.93 & 0.94 & 0.84 & 0.89 & 0.93 & 0.94
& 0.96 & 0.96 & 0.85 & 0.92 & 0.94 & 0.94 \\
& [0.11] & [0.13] & [0.08] & [0.09] & [0.10] & [0.10]
& [0.04] & [0.05] & [0.03] & [0.03] & [0.03] & [0.03] \\
\addlinespace
$(0.30, 0.30)$
& 0.85 & 0.87 & 0.65 & 0.74 & 0.83 & 0.86
& 0.90 & 0.90 & 0.59 & 0.74 & 0.86 & 0.88 \\
& [0.37] & [0.52] & [0.19] & [0.23] & [0.32] & [0.31]
& [0.18] & [0.26] & [0.07] & [0.09] & [0.13] & [0.13] \\
\addlinespace
$(0.40, 0.40)$
& 0.49 & 0.58 & 0.21 & 0.26 & 0.44 & 0.49
& 0.57 & 0.63 & 0.14 & 0.19 & 0.45 & 0.49 \\
& [1.22] & [2.23] & [0.49] & [0.60] & [1.11] & [1.05]
& [0.97] & [1.76] & [0.21] & [0.30] & [0.70] & [0.66] \\
\bottomrule
\multicolumn{13}{p{520pt}}{{\footnotesize Notes: Average lengths into brackets. Results based on $10,000$ simulations. Split: as.~var.~estimator proposed in Section~\ref{sec:var_est_obs_ranks}. No split: same as before, but without sample splitting. Unif: same as no split but with a constant bandwidth. AI: Athey and Imbens estimator. BSE: bootstrap estimator. BPC: percentile bootstrap confidence interval, based on 1,000 bootstrap samples.}}
\end{tabular}
}
\caption{Coverage and average length of nominal $95\%$ confidence intervals.}
\label{tab:CI_len_est} 
\end{threeparttable}
\end{table}

\newpage
\bibliography{biblio}

\appendix

\section{Proofs of the main results}\label{sec:apdxB}

For any real number $x$ of function $f$, we let $\bar x:=1-x$ and $\bar f:=1-f$. We use ``$\lesssim$'' to indicate an inequality up to universal constant. In most cases below, this means a constant independent of $x$ and $n$. The floor, ceiling and identity functions are denoted by $\lfloor.\rfloor$, $\lceil.\rceil$ and $\I(.)$, respectively. 

\subsection{Proof of Lemma~\ref{lemma:MomentY}}

For any $y \in \R$, let $S_Y(y) := P(Y>y)$.
Observe that $E[|Y|]<\infty$ implies $tS_{Y}(t)\to 0$ and $tF_{Y}(-t)\to 0$ as $t\to \infty$. Thus $E[|Y|^p]<\infty$ implies $t^pS_{Y}(t)\to 0$ and $t^pF_{Y}(-t)\to 0$ as $t\to \infty$. The convergence to 0 of $t^pS_{Y}(t)$ implies that there exist $C>0$ and $t_1$ such that for all $t\geq t_1$, 
$$|t|^p (1-F_Y(t)) \leq C.$$
This implies that for all $u\geq F_Y(t_1)$, $|F_Y^{-1}(u)|^p (1-u) \leq C$ or, equivalently, 
$$|F_Y^{-1}(u)|\leq C^{1/p}(1-u)^{-1/p}.$$
Hence, there exists $C_1>0$ such that for all $u\geq F_Y(t_1)$,
$$|F_Y^{-1}(u)|\leq C_1[u(1-u)]^{-1/p}.$$
Using $t^pF_{Y}(-t)\to 0$ and a similar reasoning, there exist $C_2>0$ and $t_2\leq t_1$ such that for all $u\leq F_Y(t_2)$, $|F_Y^{-1}(u)|\leq C_2[u(1-u)]^{-1/p}$. The result follows since $|F_Y^{-1}(u)[u(1-u)]^{1/p}|$ is bounded on $[F_Y(t_2),F_Y(t_1)]$. \hfill$\square$

\subsection{Proof of Inequality \eqref{eq:restriction_c}}
\label{app:proof_ineq_scale}

Since $X  \sim Z/c$,  we have $f_X(z)=cf_Z(cz)$ and thus
\begin{equation}
f_U(u)=\frac{cf_Z(cF_Z^{-1}(u))}{f_Z(F_Z^{-1}(u))}.
    \label{eq:expr_fU_scale_model}
\end{equation}
Below, $C$ denotes a constant whose value changes from one line to another. By a standard Laplace tail estimate $\int_{-\infty}^x e^{-t} t^\gamma dt \asymp e^{-x} x^\gamma$ when $x \to -\infty$, we have
\begin{align}
    F_Z(z) &\asymp \int_{-\infty}^z e^{-L|z|^\alpha} dz = \frac{1}{\alpha} \int_{-\infty}^{z^\alpha} e^{-Lt} \, t^{1/\alpha-1} dt \asymp e^{-L|z^\alpha|} z^{1-\alpha} \nonumber\\
    &\asymp f_Z(z)/|z|^{\alpha-1}.
    \label{eq:equiv_fZ}
\end{align}
Hence, as $u\to 0$,
\begin{equation}
  f_Z \circ F_Z^{-1}(u)\asymp u |F_Z^{-1}(u)|^{\alpha-1}.  
    \label{eq:equiv_fZ2}
\end{equation}
Moreover, using again $F_Z(z)\asymp f_Z(z)/[L\alpha|z|^{\alpha-1}]$, we have
$$F_Z[cz] \asymp \frac{f_Z(cz)}{|z|^{\alpha-1}} \asymp   \frac{f_Z(z)^{c^{\alpha}}}{|z|^{\alpha-1}}.$$
Hence, 
\begin{equation}
F_Z[cF_Z^{-1}(u)] \asymp  \frac{f_Z\circ F_Z^{-1}(u)^{c^{\alpha}}}{|F_Z^{-1}(u)|^{\alpha-1}}.
    \label{eq:equiv_composee}
\end{equation}
Since, by \eqref{eq:equiv_fZ}, $f_Z[cF_Z^{-1}(u)] \asymp F_Z[cF_Z^{-1}(u)] |F_Z^{-1}(u)|^{\alpha-1}$, we obtain, using \eqref{eq:expr_fU_scale_model}, \eqref{eq:equiv_fZ2} and \eqref{eq:equiv_composee},
\begin{equation}
f_U(u)\asymp f_Z\circ F_Z^{-1}(u)^{c^{\alpha}-1} \asymp u^{c^{\alpha}-1} |F_Z^{-1}(u)|^{(\alpha-1)(c^{\alpha}-1)}.
    \label{eq:equiv_fU}
\end{equation}
Since $F_Z(z) \lesssim \exp(-K|z|^\alpha)$, $|F_Z^{-1}(u)| \lesssim |\ln(u)|^{1/\alpha}$. Then, \eqref{eq:controle_fU} and \eqref{eq:equiv_fU} imply $c^{\alpha}-1 \ge -b_1$, the inequality being strict if $(\alpha-1)(c^{\alpha}-1)>0$. Inequality \eqref{eq:restriction_c} follows.

\subsection{Proof of Theorem \ref{thm:an}}
Let $Y_{(1)}<\ldots<Y_{(n_1)}$ and $Z_{(1)}<\ldots<Z_{(n_3)}$ denote the order statistics associated with the samples $(Y_i)_{i=1,\ldots,n_1}$ and $(Z_i)_{i=1,\ldots,n_3}$. Let $\xi_{(i)}:=F_Y(Y_{(i)})\sim$Uniform(0,1) and $\zeta_{(i)}:=F_Z(Z_{(i)})\sim$Uniform(0,1). Let 
\begin{align*}
    \Gn(t)&:=\frac{1}{n_1}\sum_{i=1}^{n_1}\ind{\xi_{(i)}\leq t}, \\
    \Hn(t)&:=\frac{1}{n_3}\sum_{i=1}^{n_3}\ind{\zeta_{(i)}\leq t}.
\end{align*}
Notice that 
\[
    \widehat\theta = \int_0^1 F_Y^{-1} \circ \Gn^{-1} \circ \Hn \, d\widehat{F}_U,
\]
where all integrals are defined in the Lebesgue-Stieltjes sense. We decompose the difference $\widehat{\theta}-\theta_0$ into three parts that we study independently:
\begin{align*}
    \widehat{\theta} -  \theta_0 & = \underbrace{\int_0^1 F_Y^{-1} \circ \Gn^{-1} dF_U - \int_0^1 F_Y^{-1} dF_U}_{=:T_1} + \underbrace{\int_0^1 F_Y^{-1} \circ \Gn^{-1}  \, d\widehat{F}_U - \int_0^1 F_Y^{-1} \circ \Gn^{-1} dF_U}_{=:T_2} \\[5pt]
    & \quad + \underbrace{\int_0^1 F_Y^{-1} \circ \Gn^{-1} \circ \Hn \, d\widehat{F}_U - \int_0^1 F_Y^{-1} \circ \Gn^{-1}d \widehat{F}_U}_{=: T_3}.
\end{align*}
The proof proceeds in four steps. In the first step, we prove that $\sqrt{N}T_1$ is linear up to a negligible remainder term. In the second step, we prove the same result for $\sqrt{N}T_2$. In the third step, we show that $\sqrt{N}T_3$ can be expressed as a L-statistics plus some remainder. The fourth step concludes.

\paragraph{First step: linearization of $\sqrt{N}T_1$.}

We have
\begin{align}
    \int_0^1 F_Y^{-1} \circ \Gn^{-1} dF_U &=\sum_{i=0}^{n_1-1}F_Y^{-1}\left(\xi_{(i+1)}\right)F_U\left(\left[\frac{i}{n_1},\frac{i+1}{n_1}\right)\right) \notag \\[5pt]
    &=\int_0^1  F_Y^{-1}dF_U\circ\Gn. \label{eq:change_of_measure}
\end{align}
Hence,
\begin{align*}
    \sqrt{N} T_1 &= \sqrt{N}\left(\int_0^1  F_Y^{-1}dF_U\circ\Gn - \int_0^1  F_Y^{-1}dF_U\right)\\[5pt]
    &= \sqrt{N}\left(\int_{\xi_{(1)}}^{\xi_{(n_1)}} F_Y^{-1}d\big[F_U\circ\Gn - F_U\big] - \int_0^{\xi_{(1)}} F_Y^{-1}dF_U - \int_{\xi_{(n_1)}}^1 F_Y^{-1}dF_U  \right),
\end{align*}
where the second equality follows from $\Gn$ being constant on $[0,\xi_{(1)}]$ and $[\xi_{(n_1)},1]$. Integrating the first integral by part, we obtain
\begin{align*}
   \sqrt{N} T_1 &=  -\sqrt{N}\int_{\xi_{(1)}}^{\xi_{(n_1)}}\left[F_U\circ\Gn - F_U\right]d F_Y^{-1} +R_1+R_2,
\end{align*}
where 
\begin{align*}
    R_1&=\sqrt{N}F_Y^{-1}\big(\xi_{(n_1)}\big)\left[F_U(1)-F_U\big(\xi_{(n_1)}\big)\right]-F_Y^{-1}\big(\xi_{(1)}\big)\left[F_U\big(1/n_1\big)-F_U\big(\xi_{(1)}\big)\right], \\
    R_2&= \int_0^{\xi_{(1)}} F_Y^{-1}dF_U + \int_{\xi_{(n_1)}}^1 F_Y^{-1}dF_U.
\end{align*}
Lemma~\ref{lem:11} in Supplementary Appendix Section~\ref{sec:S1} establishes that $R_1$ and $R_2$ converge to zero in probability. This implies
\begin{equation}\label{eq:T1n1_rem1}
    \sqrt{N} T_1 =  -\sqrt{N}\int_{\xi_{(1)}}^{\xi_{(n_1)}}\left[F_U\circ\Gn - F_U\right]d F_Y^{-1}+o_P(1).
\end{equation}

\medskip
Next, we further decompose \eqref{eq:T1n1_rem1} as
$$\sqrt{N} T_1 = - \sqrt{N}\int_{\xi_{(1)}}^{\xi_{(n_1)}}\left[\Gn -\I\right] d\Lambda + R_3+o_P(1),$$
where $\Lambda$ is the measure defined by $d\Lambda/dF_Y^{-1}=f_U$ and
$$R_3 := \sqrt N \left(\int_{\xi_{(1)}}^{\xi_{(n_1)}} \left[\Gn - \I\right]f_Ud F_Y^{-1} - \int_{\xi_{(1)}}^{\xi_{(n_1)}}\left[F_U\circ\Gn - F_U\right]d F_Y^{-1}\right).$$
Lemma~\ref{lem:11} in Supplementary Appendix Section~\ref{sec:S1} establishes that $R_3$ converges to zero in probability. This implies
\begin{equation*}
\sqrt{N} T_1 = - \sqrt{N}\int_{\xi_{(1)}}^{\xi_{(n_1)}}\left[\Gn - \I\right] d\Lambda+o_P(1).
\end{equation*}

\medskip
Let 
$$R_4:=\sqrt{N}\left(\int_{\xi_{(1)}}^{\xi_{(n_1)}}\left[\Gn - \I\right] d\Lambda - \int_{0}^{1}\left[\Gn-\I\right] d\Lambda \right).$$
Lemma~\ref{lem:11} in Supplementary Appendix Section~\ref{sec:S1} establishes that $R_4$ converges to zero in probability. Given that $\frac{1}{\sqrt{n_1}} \sum_{i=1}^{n_1} \eta_i = O_P(1)$ (which is shown in the fourth step), this implies
\begin{equation*}
    \sqrt{N} T_1 = \frac{\lambda_1}{\sqrt{n_1}} \sum_{i=1}^{n_1} \eta_i + o_P(1), 
\end{equation*}
with $\eta_i := - \int_0^{1}  [\mathds{1}\{F_Y(Y_i)\leq t\} - t]d\Lambda(t)$. 

\paragraph{Second step: linearization of $\sqrt{N}T_2$.} 
Similarly, we have
\begin{align*}
    \int_0^1 F_Y^{-1} \circ \Gn^{-1} d\widehat F_U &=\sum_{i=0}^{n_1-1}F_Y^{-1}\left(\xi_{(i+1)}\right)\widehat F_U\circ\Gn\left(\left[\xi_{(i)},\xi_{(i+1)}\right)\right) \\
    &=\int_0^1  F_Y^{-1}d\widehat F_U\circ\Gn.
\end{align*}
Hence,
\begin{align*}
  \sqrt{N} T_2 &= \sqrt{N}\int_0^1 F_Y^{-1} d \left[\widehat{F}_{U} \circ \Gn - F_{U} \circ \Gn\right].
\end{align*}
An integration by part yields
\begin{align}
  \sqrt{N} T_2 &= \sqrt{N}\left[F_Y^{-1}(t) \left( \widehat{F}_U\big(\Gn(t)\big) - F_U\big(\Gn(t)\big) \right) \right]_0^1 - \sqrt{N}\int_0^1 \left[\widehat{F}_{U} \circ \Gn - F_{U} \circ \Gn\right] d F_Y^{-1} \notag \\
  &=  - \sqrt{N}\int_0^1 \left[\widehat{F}_{U} \circ \Gn - F_{U} \circ \Gn\right] d F_Y^{-1},\label{eq:rewriting_T2}
\end{align}
since for $t \in [0,\xi_{(1)})$, $\Gn(t)=0$ and $\widehat{F}_U(0) = F_U(0)=0$ because $(U_i)_{i=1,\ldots,n_2}$ is an i.i.d.~sample of random variables absolutely continuous with respect to the Lebesgue measure on $[0,1]$. Symmetrically, for $t \in (\xi_{(n_1)},1]$, $\Gn(t)=1$ and $\widehat{F}_U(1) = F_U(1)=1$. We now prove that 
 \begin{equation}
 - \sqrt{N}\int_0^1 \left[\widehat{F}_{U} \circ \Gn - F_{U} \circ \Gn\right] d F_Y^{-1} =  - \sqrt{N}\int_0^1 \left[\widehat{F}_{U} - F_{U}\right]dF_Y^{-1} + o_P(1). \label{eq:reste}
\end{equation}
Let $\Vn = \sqrt{n_2}(\widehat{F}_U \circ F_U^{-1}-\I)$ denote the empirical process associated with the uniform variables $(F_U(U_i))_{i=1,\ldots,n_2}$ and define
$$R_5 := \int_0^1 \left(\Vn \circ F_U \circ \Gn - \Vn\circ F_U\right) dF_Y^{-1}.$$
To show Equation \eqref{eq:reste}, it suffices to show that $R_5=o_P(1)$. Lemma~\ref{lem:12} in Supplementary Appendix Section~\ref{sec:S1} establishes the stronger result that $E[|R_5|]\to 0$. Hence, \eqref{eq:reste} holds. 

Combined with \eqref{eq:rewriting_T2}, and given that $\frac{1}{\sqrt{n_2}} \sum_{i=1}^{n_2} \eps_i = O_P(1)$ (which is shown in the fourth step), this implies
\begin{equation}
    \sqrt{N} T_2 = \frac{\lambda_2}{\sqrt{n_2}} \sum_{i=1}^{n_2} \eps_i +o_P(1), \label{eq:eps}
\end{equation}
with $\eps_i = -\int_0^1 [\mathds{1}\{U_i\leq t\} - F_U(t)] dF_Y^{-1}(t)$.

\paragraph{Third step: $\sqrt{N}T_3$ is a L-Statistics plus some remainder terms.} We prove the result in two sub-steps. We first show that 
\begin{equation}
   \sqrt{N}T_3  = - \sqrt{N} \int_{\xi_{(1)}}^{\xi_{(n_1)}}\left[\widehat{F}_U \circ \Hn^{-1} \circ \Gn - \widehat{F}_U \circ \Gn \right] dF_Y^{-1} +o_P(1). \label{eq:new_exp}
\end{equation}
Second, we show that 
\begin{equation}
\sqrt{N}T_3  = - \sqrt{N} \underbrace{\int_{\xi_{(1)}}^{\xi_{(n_1)}} \left[F_U \circ \Hn^{-1}\circ \Gn - F_U \circ  \Gn\right] dF_Y^{-1}}_{=:J_1} + o_P(1).
\label{eq:T3_J2}
\end{equation}
Let us then write $-\sqrt{N}J_1=\sqrt{N}J_2+R_{6}+R_{7}+R_{8}+R_{9}$, with:
\begin{align}
    J_2 & := - \int_{1/n_3}^{1-1/n_3}\left[\Hn^{-1}(x) - E[\Hn^{-1}(x)] \right]f_U(x)d F_Y^{-1}(x), \\[5pt]
    R_6 & := -\sqrt{N}\left(J_1 - \int_{\xi_{(1)}}^{\xi_{(n_1)}}\left[\Hn^{-1} \circ \Gn - \Gn \right]f_U \, dF_Y^{-1}\right), \label{eq:def_R1n} \\[5pt]
    R_7 & := -\sqrt N \left( \int_{\xi_{(1)}}^{\xi_{(n_1)}}\left[\Hn^{-1} \circ \Gn - \Gn \right]f_U dF_Y^{-1} - \int_{\xi_{(1)}}^{\xi_{(n_1)}}\left[\Hn^{-1} - \I \right]f_U d F_Y^{-1}\right),  \label{eq:def_R2n} \\[5pt]
    R_8 & :=  \sqrt{N}\int_{\xi_{(1)}}^{\xi_{(n_1)}}\left[x - E[\Hn^{-1}(x)]\right]f_U(x)d F_Y^{-1}(x),  \label{eq:def_R3n} \\[5pt]
    R_9 & :=  \sqrt{N} \left(\int_{1/n_3}^{1-1/n_3}\left[\Hn^{-1}(x) - E[\Hn^{-1}(x)] \right]f_U(x)d F_Y^{-1}(x) \right. \notag \\[5pt]
    & \left. \quad - \int_{\xi_{(1)}}^{\xi_{(n_1)}}\left[\Hn^{-1}(x) - E[\Hn^{-1}(x)]\right]f_U(x)d F_Y^{-1}(x) \right)  \label{eq:def_R4n}.
    \end{align}
Lemma~\ref{lem:13} in Supplementary Appendix Section~\ref{sec:S1} establishes that each of the four terms $R_{6}$--$R_{9}$ tends to 0 in probability.

{\bf First sub-step: Equation \eqref{eq:new_exp} holds.}
Let $\mathcal I_{0}:=[0,\Hn^{-1}(n_1^{-1})]$ and  $\mathcal I_{1}:=[\zeta_{(n_3)}, 1]$. Let $\mathbb G_{n_1,+}^{-1}$ denote the right-continuous generalized inverse of $\Gn$:
$$
\mathbb G_{n_1,+}^{-1}(t):=\sup\{x\in[0,1]:\Gn(x)\leq y\}, \quad \forall t\in(0,1).
$$
We recall that $$T_3 = \int_0^1 F_Y^{-1} \circ \Gn^{-1} \circ \Hn \, d\widehat{F}_U - \int_0^1 F_Y^{-1} \circ \Gn^{-1}d \widehat{F}_U.$$
By splitting the first integral in $\sqrt{N}T_3$ and applying Lemma~\ref{lem_Gn+_equiv_Gn}, we obtain
\begin{align}
   & \sqrt{N}T_3 \notag \\
   &\quad = \sqrt{N} \Big( \int_{(\mathcal I_0 \cup \mathcal I_1)^c} F_Y^{-1} \circ \left(\Hn^{-1}\circ \Gn\right)^{-1} d\widehat F_U  - \int_0^1  F_Y^{-1} \circ \Gn^{-1} d\widehat{F}_U  + \int_{\mathcal I_0 \cup \mathcal I_1} F_Y^{-1} \circ  \Gn^{-1}\circ\Hn d\widehat{F}_U  \notag \\[5pt]
   &\quad\quad  + \int_{(\mathcal I_0 \cup \mathcal I_1)^c}\left[ F_Y^{-1}\circ \Gn^{-1}\circ \Hn - F_Y^{-1} \circ  \mathbb G_{n_1,+}^{-1}  \circ \Hn\right]d\widehat{F}_U  \notag \\[5pt]
    & \quad\quad  + \int_{(\mathcal I_0 \cup \mathcal I_1)^c}\left[F_Y^{-1} \circ  \mathbb G_{n_1,+}^{-1}  \circ \Hn - F_Y^{-1} \circ  \left(\Hn^{-1}\circ \Gn\right)^{-1}\right]d\widehat{F}_U \Big) \notag \\[5pt]
 &\quad = \sqrt{N} \bigg( \int_{(\mathcal I_0 \cup \mathcal I_1)^c} F_Y^{-1} \circ \left(\Hn^{-1}\circ \Gn\right)^{-1} d\widehat F_U  - \int_0^1  F_Y^{-1} \circ \Gn^{-1} d\widehat{F}_U  \notag \\ 
 & \qquad + \int_{\mathcal I_0 \cup \mathcal I_1} F_Y^{-1} \circ  \left(\Hn^{-1}\circ \Gn\right)^{-1}d\widehat{F}_U \bigg)  +o_P(1). 
 \end{align}
By \eqref{eq:change_of_measure} and applying Lemma~\ref{lem:stieltjes_subs} with $a=\xi_{(1)}$, $b=\xi_{(n_1)}$, $f=F_Y^{-1}$ restricted to $[a,b]$, $N=\widehat F_U$, and $M=\Hn^{-1}\circ\Gn$, we obtain
 \begin{align}
     \sqrt{N}T_3&  = \sqrt{N}\Bigg(\int_{\xi_{(1)}}^{\xi_{(n_1)}} F_Y^{-1}  d\left[\widehat{F}_U \circ \Hn^{-1} \circ \Gn \right] - \int_{0}^{1} F_Y^{-1}  d \left[ \widehat{F}_U \circ \Gn \right]   \notag \\[5pt] 
      & \quad + \int_{\mathcal I_0 \cup \mathcal I_1} F_Y^{-1} \circ \Gn^{-1} \circ \Hn  d\widehat{F}_U \Bigg) + o_P(1) \notag \\[5pt]
      & = \sqrt{N}\int_{\xi_{(1)}}^{\xi_{(n_1)}} F_Y^{-1}  d\left[\widehat{F}_U \circ \Hn^{-1} \circ \Gn - \widehat{F}_U \circ \Gn \right] \notag \\[5pt] 
      &\qquad + \sqrt{N}\int_{\mathcal I_0 \cup \mathcal I_1} F_Y^{-1} \circ \Gn^{-1} \circ \Hn  d\widehat{F}_U +o_P(1), \label{eq:thm2_start}
 \end{align}
where we used that $\widehat{F}_U \circ \Gn$ is constant on the two segments $[0, \xi_{(1)}]$ and $[\xi_{(n_1)},1]$ to obtain the last equality. Note that
\begin{align*}
    & \sqrt{N}\left[F_Y^{-1}(t) \left(\widehat{F}_U \circ \Hn^{-1} \circ \Gn(t) - \widehat{F}_U \circ \Gn(t) \right) \right]_{t=\xi_{(1)}}^{t=\xi_{(n_1)}} \\[5pt]
    & \quad = \sqrt{N}\left[F_Y^{-1}(\xi_{(n_1)}) \left(\widehat{F}_U(\zeta_{(n_3)}) - 1 \right) -  F_Y^{-1}(\xi_{(1)})\left( \widehat F_U(\Hn^{-1}(n_1^{-1}))- \widehat{F}_U(n_1^{-1})\right) \right].
\end{align*}
Thus, an integration by part of the first term in \eqref{eq:thm2_start} yields 
\begin{align*}
   & \sqrt{N}T_{3N} \\
   & = - \sqrt{N} \int_{\xi_{(1)}}^{\xi_{(n_1)}}\left[\widehat{F}_U \circ \Hn^{-1} \circ \Gn - \widehat{F}_U \circ \Gn \right] dF_Y^{-1} \\[5pt]
   &\quad + \underbrace{\sqrt{N}F_Y^{-1}\big(\xi_{(n_1)}\big) \left(\widehat{F}_U\big(\zeta_{(n_3)}\big) - 1 \right)}_{=:A} - \underbrace{F_Y^{-1}(\xi_{(1)})\widehat F_U\big(\Hn^{-1}(n_1^{-1})\big)}_{=:B} - \underbrace{\sqrt{N}F_Y^{-1}\big(\xi_{(n_1)}\big) \widehat{F}_U\big(n_1^{-1}\big)}_{=:C}  \\[5pt]
   &\quad +\underbrace{\sqrt{N} \int_{\mathcal I_0 \cup \mathcal I_1} F_Y^{-1} \circ \Gn^{-1} \circ \Hn  d\widehat{F}_U}_{=:D}+o_P(1).
\end{align*}
To show \eqref{eq:new_exp}, it suffices to show that $A,B,C,D$ all converge to zero in probability. We show below the stronger result that each term converges in $L_1$. Under Assumption~\ref{ass:sampling_est} and Assumption~\ref{ass:smoothness1}\ref{ass:smoothness12}, we have
\begin{align*}
E\left[\abs{A}\right] &  = E\left[\abs{\sqrt{N}F_Y^{-1}\big(\xi_{(n_1)}\big)(1-\widehat F_U\big(\zeta_{(n_3)}\big)}\right] \\[5pt]
& = \sqrt{N}E\left[\abs{F_Y^{-1}\big(\xi_{(n_1)}\big)}\right]E\left[1-F_U\big(\zeta_{(n_3)}\big)\right]\\[5pt]
&\lesssim n_1^{1/2+d_2}E\left[\int_{\zeta_{(n_3)}}^{1}u^{-b_1}(1-u)^{-b_2}du\right] \\[5pt]
&\lesssim n_1^{1/2+d_2}E\left[(1-\zeta_{n_3})^{1-b_2}\right] \\[5pt]
&\leq n_1^{1/2+d_2}\left(E\left[1-\zeta_{n_3}\right]\right)^{1-b_2} \\[5pt]
&\lesssim n_1^{b_2+d_2-1/2} \\
&=o(1),
\end{align*}
where the first inequality follows from Lemma~\ref{lem:extrem_moments1} and Assumption~\ref{ass:smoothness1}\ref{ass:smoothness13}, and the last inequality follows from Jensen's inequality and $E\left[\zeta_{n_3}\right]=n_3/(n_3+1)$. Next, we obtain similarly
\begin{align*}
E\left[\abs{B}\right] &  = E\left[\abs{\sqrt{N}F_Y^{-1}\big(\xi_{(1)}\big)\widehat F_U\big(\Hn^{-1}(n_1^{-1})\big)}\right] \\[5pt]
& = \sqrt{N}E\left[\abs{F_Y^{-1}\big(\xi_{(1)}\big)}\right]E\left[F_U\big(\Hn^{-1}(n_1^{-1})\big)\right]\\[5pt]
&\lesssim n_1^{1/2+d_1}\left(E\left[\Hn^{-1}\left(\frac{1}{n_1}\right)\right]\right)^{1-b_1} \\[5pt]
&\lesssim n_1^{1/2+d_1}\left(\frac{\lceil n_3/n_1\rceil}{n_3+1}\right)^{1-b_1} \\[5pt]
&\lesssim n_1^{b_1+d_1-1/2} \\[5pt]
&=o(1),
\end{align*}
using Assumption~\ref{ass:sampling_est}\ref{ass:sampling_est3} and $\lambda_1>0$. Next,
\begin{align*}
E\left[\abs{C}\right] &  = E\left[\abs{\sqrt{N}F_Y^{-1}\big(\xi_{(n_1)}\big)\widehat{F}_U\big(n_1^{-1}\big)}\right] \\
&\quad  =\sqrt{N}E\left[\abs{F_Y^{-1}\big (\xi_{(1)}\big)}\right]F_U\big (n_1^{-1}\big )\\[5pt]
& \quad \lesssim n_1^{\frac12+d_1}\left(\frac{1}{n_1}\right)^{1-b_1} \\
&\quad =o(1),
\end{align*}
where the inequality follows from Lemma~\ref{lem:extrem_moments1} and Assumption~\ref{ass:smoothness1}\ref{ass:smoothness13}. Moreover, we have
\begin{align*}
   & \left\vert \sqrt{N} \int_{\mathcal I_0 \cup \mathcal I_1} F_Y^{-1} \circ \Gn^{-1} \circ \Hn  d\widehat{F}_U\right\vert \\[5pt]
   & \quad \leq \sqrt{N}\widehat{F}_U(\Hn^{-1}(n_1^{-1}))\abs{F_Y^{-1}(\xi_{(1)})} + \sqrt{N}\left[1 - \widehat{F}_U\big(\zeta_{(n_3)}\big) \right]\abs{F_Y^{-1}\big(\xi_{(n_1)}\big)} \\[5pt]
   & \quad = \abs{A}+\abs{B} \\[5pt]
   & \quad = o_P(1).
\end{align*}
Conclude that \eqref{eq:new_exp} holds.

{\bf Second sub-step: Equation \eqref{eq:T3_J2} holds.}

From~\eqref{eq:new_exp}, we have
\begin{align*}
\sqrt{N} T_3 \; = &  - \sqrt{N} \underbrace{\int_{\xi_{(1)}}^{\xi_{(n_1)}} \left[F_U \circ \Hn^{-1}\circ \Gn - F_U \circ  \Gn\right]  dF_Y^{-1}}_{=J_1} \\[5pt]
& - \sqrt{N} \int_{\xi_{(1)}}^{\xi_{(n_1)}} \left[\widehat{F}_U \circ \Hn^{-1}\circ \Gn - F_U \circ \Hn^{-1}\circ \Gn \right]  dF_Y^{-1} \\[5pt]
& - \sqrt{N} \int_{\xi_{(1)}}^{\xi_{(n_1)}} \left[F_U \circ  \Gn - \widehat{F}_U \circ \Gn \right]  dF_Y^{-1}.
\end{align*}
We show below that
\begin{equation}\label{eq:reste_j1n}
\sqrt{N} \int_{\xi_{(1)}}^{\xi_{(n_1)}} \left[\widehat{F}_U \circ \Hn^{-1}\circ \Gn - F_U \circ \Hn^{-1}\circ \Gn \right] dF_Y^{-1}= \sqrt{N}\int_{\xi_{(1)}}^{\xi_{(n_1)}} \left[\widehat F_U  - F_U \right]  dF_Y^{-1} + o_P(1).
\end{equation}
Once combined with 
\[
\sqrt{N}\int_{\xi_{(1)}}^{\xi_{(n_1)}} \left[\widehat F_U  - F_U \right]  dF_Y^{-1}=\sqrt{N}\int_0^1 \left[\widehat F_U  - F_U \right]  dF_Y^{-1}+o_P(1),
\]
which follows from \eqref{eq:tail_up}--\eqref{eq:tail_low}, an integration by parts, Cauchy--Schwarz inequality and Lemmas~\ref{lem:extrem_moments1}--\ref{lem:extrem_moments2}, with 
\[
- \sqrt{N} \int_{\xi_{(1)}}^{\xi_{(n_1)}} \left[F_U \circ  \Gn - \widehat{F}_U \circ \Gn \right]  dF_Y^{-1}=- \sqrt{N} \int_0^1 \left[F_U \circ  \Gn - \widehat{F}_U \circ \Gn \right]  dF_Y^{-1},
\]
and with 
\eqref{eq:reste}, this proves \eqref{eq:T3_J2}.

\medskip
To prove \eqref{eq:reste_j1n}, we closely follow the proof of \eqref{eq:reste}. Recall that $\Vn = \sqrt{n_2}(\widehat{F}_U \circ F_U^{-1}-\I)$, and let $\mathds 1_{\mathcal A_{n_1}(x)}:=\mathds 1\{\xi_{(1)}< x< \xi_{(n_1)}\}$,
$$R_{10}:= \int_0^1 \mathds 1_{\mathcal A_{n_1}}\left(\Vn \circ F_U \circ \Hn^{-1} \circ \Gn - \Vn\circ F_U\right)  dF_Y^{-1},$$
and $ I_N(x):=(x,\Hn^{-1} \circ\Gn(x)]$ if $\Hn^{-1} \circ\Gn(x)>x$, $I_N(x):=[\Hn^{-1} \circ\Gn(x),x)$ if $\Hn^{-1} \circ\Gn(x)<x$ and $\emptyset$ otherwise. We prove that $E[|R_{10}|]\to 0$. 
Proceeding as in the derivation of \eqref{eq:ineg2_Rn}, but conditioning on $(\xi_i,\zeta_i)_i$ rather than $(\xi_i)_i$ only, we get
$$E\left[|R_{10}|\right] \leq \int_0^1 E\left[\mathds 1_{\mathcal A_{n_1}(x)}|F_U(x) -F_U(\Hn^{-1}\circ \Gn(x))|\right]^{1/2}  dF_Y^{-1}(x).$$
Because $\Gn(x)\convP x$ and $\mathds 1_{\mathcal A_{n_1}(x)}\convP 1$, by uniform convergence of $\Hn^{-1}$ towards $\I$ and the continuous mapping theorem, $\mathds 1_{\mathcal A_{n_1}(x)}|F_U(\Hn^{-1}\circ \Gn(x))-F_U(x)|\convP 0$ for all $x\in[0,1]$. Moreover, $\mathds 1_{\mathcal A_{n_1}(x)}|F_U(\Hn^{-1}\circ \Gn(x))-F_U(x)|\leq 1$. Hence, for all $x\in[0,1]$,
$$E\left[\mathds 1_{\mathcal A_{n_1}(x)}|F_U(x) -F_U(\Hn^{-1}\circ \Gn(x))|\right] \to 0.$$
Next, we show $E[|R_{10}|]\to 0$ by proving (focusing on a neighborhood of $0$ without loss of generality) 
\begin{equation}
E\left[\mathds 1_{\mathcal A_{n_1}(x)}|F_U(x) -F_U(\Hn^{-1}\circ \Gn(x))|\right]\lesssim x^{1-b_1}   
    \label{eq:for_DCT_thm2}
\end{equation}
and applying the dominated convergence theorem. As in the previous steps, we apply Lemma \ref{lem:useful_lem} with $Q_n(x):=\Hn^{-1}\circ \Gn(x)$ and $B_n(x)=\mathds 1_{\mathcal A_{n_1}(x)}$. The two conditions of this lemma are checked in Lemma \ref{lem:HnGn}. Hence, \eqref{eq:for_DCT_thm2}, and thus \eqref{eq:T3_J2}, hold.

\paragraph{Fourth step: conclusion.}
By the previous steps, we have
$$\sqrt{N}\left(\widehat \theta -  \theta_0\right) = \sqrt{\frac{\lambda_1}{n_1}}\sum_{i=1}^{n_1} \eta_i + \sqrt{\frac{\lambda_2}{n_2}}\sum_{i=1}^{n_2} \eps_i + \sqrt{\lambda_3n_3}J_2 +o_P(1).$$
By definition of $\eta_i$ and $\eps_i$, we have $E[\eta_i]=E[\eps_i]=0$ and 
\begin{align*}
    E[\eta_i^2] &= \int_0^1\int_0^1 (s \wedge t - st) f_U(s) f_U(t) dF_Y^{-1}(s)dF_Y^{-1}(t), \\
    E[\eps_i^2]& = \int_0^1\int_0^1 (F_U(s \wedge t) - F_U(s)F_U(t)) dF_Y^{-1}(s) dF_Y^{-1}(t),
\end{align*}
which are both finite under Assumption~\ref{ass:smoothness1}\ref{ass:smoothness14} and by Lemma~1 in \cite{ShorackWellner1986}. Moreover, under Assumption~\ref{ass:sampling_est}\ref{ass:sampling_est2}, $(\eta_i)_{i=1,\ldots,n_1}$ and $(\eps_i)_{i=1,\ldots,n_2}$ are independent. Hence the first two terms on the right-hand side are asymptotically normal. Lemma~\ref{lem:14} in Supplementary Appendix Section~\ref{sec:S1} establishes that $\sqrt{N}J_2$ tends to a normal distribution. Moreover, by Assumption \ref{ass:sampling_est}\ref{ass:sampling_est2}, $J_2$ is independent of the $(\eta_i, \eps_i)_{i\geq 1}$. Therefore, the vector $\left(\sum_{i=1}^{n_1}(\eta_i)/\sqrt{n_1} + \sum_{i=1}^{n_2}(\eps_i)/\sqrt{n_2}, \sqrt{n_3}J_{2n_3}\right)$ converges jointly in distribution to two independent normal variables distributions. The result follows.  \hfill$\square$

\subsection{Proof of Proposition \ref{prop:necessary_cond}}

We focus on the condition $E[\eta^2]<\infty$, as it will be sufficient to conclude. Let $q(u):=f_U(u)F_Y^{-1}{}'(u)$ and $h(\xi):= \int_0^1(t-\ind{\xi\le t})q(t)dt$, so that $E[\eta^2]=\int_0^1 h^2(\xi)d\xi$. Remark that
$$h(\xi) = h(1/2) + \int_{1/2}^\xi q(t)dt.$$
Assume that $\xi < 1/2$. Then, using \eqref{eq:ineg_f_U_FY}, we obtain
$$-\int_{1/2}^\xi q(t)dt \ge C_1 \xi^{-(b_1+d_1)} + C_2$$
for some constants $C_1 >0$ and $C_2>0$. Hence,
$$E[\eta^2] = \int_0^1 h^2(\xi)d\xi \ge C'_1 \int_0^{1/2} \xi^{-2(b_1+d_1)}d\xi +C_2.$$
Then, $E[\eta^2] <\infty$ implies that $b_1+d_1<1/2$. A similar reasoning for $\xi>1/2$ yields $b_2+d_2<1/2$.

\subsection{Proof of Lemma~\ref{lem:change_var_sig1}}

Let $V:=F_Y(Y)$. Because $F_Y$ is continuous, $V$ is uniform. Now, remark that as soon as $V\in (0,1)$
$$|\ind{V\le t} - t|\le \frac{t(1-t)}{V\wedge (1-V)}.$$
This implies that as soon as $V\in (0,1)$, 
$$\int_0^1 |\ind{V\le t} - t|f_U(t)dF_Y^{-1}(t) \le \frac{\int_0^1 t(1-t)f_U(t)dF_Y^{-1}(t)}{V\wedge (1-V)}<\infty.$$
Hence, $\eta$ is well-defined almost surely. Next, by Fubini's theorem, 
$$\eta^2 = \int_{[0,1]^2} (\ind{V\le u} - u)(\ind{V\le u'} - u')f_U(u)f_U(u')dF^{-1}_Y(u)dF^{-1}_Y(u').$$   
Moreover, by Cauchy--Schwarz inequality,
$$E\left[|(\ind{V\le u} - u)(\ind{V\le u'} - u')|\right]\le [u(1-u)v(1-v)]^{1/2}.$$
As a result,
\begin{align*}
  & \int_{[0,1]^2} E\left[\left|(\ind{V\le u} - u)(\ind{V\le u'} -u')\right|\right]f_U(u)f_U(u') dF^{-1}_Y(u) dF^{-1}_Y(u') \\
  \le & \left(\int_{[0,1]} [u(1-u)]^{1/2}f_U(u)dF_Y^{-1}(u)\right)^2<\infty.  
\end{align*}
Then, by Fubini's theorem,
\begin{align*}
  & \int_{[0,1]^2}  E\left[(\ind{V\le u} - u)(\ind{V\le u'} - u')\right]f_U(u)f_U(u')dF^{-1}_Y(u)dF^{-1}_Y(u') \\
  = & E\left[\int_{[0,1]^2}  (\ind{V\le u} - u)(\ind{V\le u'} - u')f_U(u)f_U(u')dF^{-1}_Y(u) dF^{-1}_Y(u')\right]=E[\eta^2].
\end{align*}
Thus, $E[\eta^2]<\infty$ and since $E\left[(\ind{V\le u} - u)(\ind{V\le u'} - u')\right]=(u\wedge u')(1- u\vee u')$, 
$$E[\eta^2] = \int_{[0,1]^2}(u\wedge u')(1- u\vee u') f_U(u)f_U(u')dF^{-1}_Y(u)dF^{-1}_Y(u').$$
The result follows by the change of variable $y=F_Y^{-1}(u)$ and $y'=F_Y^{-1}(u')$. \hfill$\square$

\subsection{Proof of Theorem~\ref{thm:var_est}}

First, we use repeatedly the fact that by Assumption \ref{ass:bandwidth}, $n_2^{(b_1+ 2d_1)\lor (b_2+2d_2)-1}=o(\eps_{n_2})$ as $n_2\to\infty$. Now, let $\widetilde f_U^{(1)}$ and $\widetilde f_U^{(2)}$ denote two infeasible sample-split kernel density estimators of $f_U$: for all $u\in(0,1)$,
\begin{align}
    \widetilde f_U^{(1)}(u)&=\frac{1}{n_2h_{n_2,u}}\sum_{i=1}^{n_2/2}\ind{\abs{U_i-u}\leq h_{n_2,u}}, \label{eq_def_hat_f_U} \\
    \widetilde f_U^{(2)}(u)& =\frac{1}{n_2h_{n_2,u}}\sum_{i=n_2/2+1}^{n_2}\ind{\abs{U_i-u}\leq h_{n_2,u}}. \notag
\end{align}
We also define the following infeasible estimator of $\sigma^2$:
\begin{equation}
    \widetilde \sigma^2:=\frac{N}{n_1}\int_{\R^2}\widetilde f_U^{(1)}\!\left(\widehat F_Y^{(1)}(y)\right)\,\widetilde f_U^{(2)}\!\left(\widehat F_Y^{(2)}(y')\right)\,w\!\left(\widehat F_Y^{(1)}(y), \widehat F_Y^{(2)}(y')\right)dydy' + \frac{N}{n_2^2} \sum_{i=1}^{n_2}\widehat \eps_i^2.
\end{equation}
The proof has two steps:
\begin{enumerate}
    \item We show that $\widetilde \sigma^2-\sigma^2=o_P(1)$.
    \item We show that $\widehat\sigma^2-\widetilde \sigma^2=o_P(1)$.
\end{enumerate}
\paragraph{First step: $\widetilde \sigma^2-\sigma^2=o_P(1)$.}
Consider the decomposition
\begin{align*}
     \widetilde\sigma^2-\sigma^2&= R_1+R_2,
\end{align*}
where
\begin{align*}
     R_1 & := \frac{N(n_1+n_3)}{n_1n_3}\int_{\R^2}\widetilde f_U^{(1)}\!\left(\widehat F_Y^{(1)}(y)\right)\widetilde f_U^{(2)}\!\left(\widehat F_Y^{(2)}(y')\right) w(\widehat F_Y^{(1)}(y), \widehat F_Y^{(2)}(y'))dydy'\\
     & \quad - (\lambda_1+\lambda_3)\int_{\R^2}f_U(F_Y(y)) \, f_U(F_Y(y')) \, w(F_Y(y), F_Y(y'))dydy', \\[5pt]
      R_2 & :=  \frac{N}{n_2^2} \sum_{i=1}^{n_2}\widehat F_Y^{-1}(U_i)^2 - \lambda_2E\left[F_Y^{-1}(U_1)^2 \right]
\end{align*}
and $w(s,t) = (s \land t) (\bar s \land \bar t)$ for any $s,t \in [0,1]$. 
For any two functions $u,v : [0,1]^2 \to (0,\infty)$, any two functions $a,b : [0,1] \to \R$ and any function $F: \R \to [0,1]$, we define
\begin{align*}
& \left\langle u, v\right\rangle_{F_1\otimes F_2} \, = \, \int_{\R^2} u\big(F_1(y),F_2(y')\big) v\big(F_1(y),F_2(y')\big) dy dy'\\[5pt]
&a \otimes b(y,y') = a(y) b(y'), \quad \forall y,y' \in [0,1]\\[5pt]
&\,a^{\,\otimes \,2} \, = \, a\otimes a.
\end{align*}
We can now rewrite the quantity $R_1$ as 
\begin{align*}
 R_1 & := \frac{N(n_1+n_3)}{n_1n_3} \left\langle \widetilde f_U^{(1)} \otimes \widetilde f_U^{(2)}, w\right\rangle_{\widehat F_{Y}^{(1)}\otimes\widehat F_{Y}^{(2)}} - (\lambda_1+\lambda_3) \left\langle f_U^{\,\otimes \,2}, w\right\rangle_{F_Y^{\,\otimes \,2}}.
\end{align*}
Define
\begin{align*}
     \widetilde R_1 & := \left\langle \widetilde f_U^{(1)} \otimes \widetilde f_U^{(2)}, w\right\rangle_{\widehat F_{Y}^{(1)}\otimes\widehat F_{Y}^{(2)}} - \left\langle f_U^{\,\otimes \,2}, w\right\rangle_{F_Y^{\,\otimes \,2}}\\[5pt]
    \widetilde R_2 & := \frac{1}{n_2} \sum_{i=1}^{n_2}\widehat F_Y^{-1}(U_i)^2 - E\left[F_Y^{-1}(U_1)^2 \right].
\end{align*}
It is sufficient to show that $\widetilde R_1+\widetilde R_2=o_P(1)$. In Step 1 below, we show that $\widetilde R_1= o_P(1)$. In Step 2, we show that $\widetilde R_2= o_P(1)$.

\medskip
\noindent {\bf Step 1: $\widetilde R_1= o_P(1)$.}
Let 
\begin{align*}
    I_{1}&:= E\left[\left\langle \widetilde f_U^{(1)} \otimes \widetilde f_U^{(2)}, w\right\rangle_{\widehat F_{Y}^{(1)}\otimes\widehat F_{Y}^{(2)}}\right] - \left\langle f_U^{\,\otimes \,2}, w\right\rangle_{F_Y^{\,\otimes \,2}}, \\[10pt]
    I_{2}&:=\left\langle \widetilde f_U^{(1)} \otimes \widetilde f_U^{(2)}, w\right\rangle_{\widehat F_{Y}^{(1)}\otimes\widehat F_{Y}^{(2)}}-E\left[\left\langle \widetilde f_U^{(1)} \otimes \widetilde f_U^{(2)}, w\right\rangle_{\widehat F_{Y}^{(1)}\otimes\widehat F_{Y}^{(2)}}\right],
\end{align*}
where the expectation is easily shown to be finite (see further steps). 
Note that $\widetilde R_{1}=I_{1}+I_{2}$ and that $E[\vert \widetilde R_{1}\vert ^2]=I_{1}^2+E[I_{2}^2]$ is a bias-variance decomposition. We will show the stronger result that $E[\vert \widetilde R_{1}\vert^2]\to0$. 

\medskip
\noindent {\bf Step 1.A (Vanishing bias): $I_1=o(1)$.}  For any $t \in \R$, we define  $\Delta_{2h_{n_2,t}} F_U(t):= F_U(t+h_{n_2,t})- F_U(t-h_{n_2,t})$. Since $\widetilde f_U^{(1)}(t)$, $\widetilde f_U^{(2)}(t)$, and $(Y_i)_{i=1}^{n_1}$ are mutually independent and $E\left[\widetilde f_U^{(1)}(t)\right]=E\left[\widetilde f_U^{(2)}(t)\right] = \Delta_{2h_{n_2,t}} F_U(t)/(2h_{n_2,t})$ for any $t \in \R$, the law of total expectation yields
\begin{align*}
	& E \left[\left\langle \widetilde f_U^{(1)} \otimes \widetilde f_U^{(2)}, w\right\rangle_{\widehat F_{Y}^{(1)}\otimes\widehat F_{Y}^{(2)}}\right]\\[10pt]
   =& E\left[\int_{\R^2}\,\widetilde f_U^{(1)}\!\left(\widehat F_Y^{(1)}(y)\right)\,\widetilde f_U^{(2)}\!\left(\widehat F_Y^{(2)}(y')\right) w\!\left(\widehat F_Y^{(1)}(y), \widehat F_Y^{(2)}(y')\right)dydy'\right]\\[6pt]
   =&\int_{\R^2}E\left[\frac{\Delta_{2h_{n_2, \widehat F_Y^{(1)}(y)}}F_U(\widehat F_Y^{(1)}(y))}{2h_{n_2,\widehat F_Y^{(1)}(y)}}\frac{\Delta_{2h_{n_2,\widehat F_Y^{(2)}(y')}}F_U(\widehat F_Y^{(2)}(y'))}{2h_{n_2,\widehat F_Y^{(2)}(y')}}w\!\left(\widehat F_Y^{(1)}(y), \widehat F_Y^{(2)}(y')\right)\right]dydy'.
\end{align*}
For any $y \in \R$, the mean value theorem ensures that there exist
\begin{align*}
    \widetilde F_{Y}^{(1)}(y) &\in \big(\widehat F_Y^{(1)}(y) \, -\,h_{n_2,\widehat F_Y^{(1)}(y)}, \,\widehat F_Y^{(1)}(y) \,+ \,h_{n_2,\widehat F_Y^{(1)}(y)} \big)\\
    \widetilde F_{Y}^{(2)}(y') &\in \big(\widehat F_Y^{(2)}(y')-h_{n_2,\widehat F_Y^{(2)}(y')},\widehat F_Y^{(2)}(y')+h_{n_2,\widehat F_Y^{(2)}(y')} \big)
\end{align*}
 satisfying 
\begin{align*}
\frac{\Delta_{2h_{n_2, \widehat F_Y^{(j)}(y)}}F_U(\widehat F_Y^{(j)}(y))}{2h_{n_2,\widehat F_Y^{(j)}(y)}} = f_U\!\left(\widetilde F_Y^{(j)}(y)\right), \quad j\in\{1,2\}.
\end{align*} 
 Therefore,
 \begin{align}
 E \left[\left\langle \widetilde f_U^{(1)} \otimes \widetilde f_U^{(2)}, w\right\rangle_{\widehat F_{Y}^{(1)}\otimes\widehat F_{Y}^{(2)}}\right] 
   =&\int_{\R^2}E\Big[f_U\!\left(\widetilde F_Y^{(1)}(y)\right)f_U\!\left(\widetilde F_Y^{(2)}(y')\right)w\!\left(\widehat F_Y^{(1)}(y), \widehat F_Y^{(2)}(y')\right)\Big]dydy'.\label{eq_hat_to_tilde}
\end{align}
Next, note that
\[
\left|\widetilde F_Y^{(j)}(y)-\widehat F_Y^{(j)}(y)\right| \leq h_{n_2,\widehat F_Y^{(j)}(y)}=\eps_{n_2}\widehat F_Y^{(j)}(y)\bar{\widehat F}_Y^{(j)}(y)\leq\eps_{n_2}\widehat F_Y^{(j)}(y).
\]
This implies
\[
\left|\widetilde F_Y^{(j)}(y)-\widehat F_Y^{(j)}(y)\right|\leq \eps_{n_2}\left|\widetilde F_Y^{(j)}(y)-\widehat F_Y^{(j)}(y)\right|+\eps_{n_2}\widetilde F_Y^{(j)}(y),
\]
which, since $\eps_{n_2}\leq1/2$, further implies
\begin{equation}\label{eq:control_tilde_hat}
  \left|\widetilde F_Y^{(j)}(y)-\widehat F_Y^{(j)}(y)\right| \leq  \frac{\eps_{n_2}}{1-\eps_{n_2}}\widetilde F_Y^{(j)}(y)\leq 2\eps_{n_2}\widetilde F_Y^{(j)}(y).
\end{equation}
We conclude that $\left|\widetilde F_Y^{(j)}(y)-\widehat F_Y^{(j)}(y)\right| \leq 2\eps_{n_2}(\widetilde F_Y^{(j)}(y)\wedge \widehat F_Y^{(j)}(y))$. 
An analogous reasoning yields $\big|\bar{\widetilde F}_Y^{(j)}(y)-\bar{\widehat F}_Y^{(j)}(y)\big| \leq 2\eps_{n_2}(\bar{\widetilde F}_Y^{(j)}(y)\wedge\bar{\widehat F}_Y^{(j)}(y))$. By Lemma~\ref{lem:useful_ineq}, we then have
\begin{align}
    w\left(\widehat F_Y^{(1)}(y), \widehat F_Y^{(2)}(y')\right) &= \left(\widehat F_Y^{(1)}(y) \land \widehat F_Y^{(2)}(y')\right) \left(\bar{\widehat F}_Y^{(1)}(y) \land \bar{\widehat F}_Y^{(2)}(y')\right) \nonumber\\
    & \leq (1+2\varepsilon_{n_2})\left(\widetilde F_Y^{(1)}(y) \land \widetilde F_Y^{(2)}(y')\right) (1+2\varepsilon_{n_2}) \left(\bar{\widetilde F}_Y^{(1)}(y) \land \bar{\widetilde F}_Y^{(2)}(y')\right) \nonumber\\
    & = \big(1+O(\varepsilon_{n_2})\big) \left(\widetilde F_Y^{(1)}(y) \land \widetilde F_Y^{(2)}(y')\right) \left(\bar{\widetilde F}_Y^{(1)}(y) \land \bar{\widetilde F}_Y^{(2)}(y')\right).\label{eq_upper_bound_w}
\end{align}
Combining~\eqref{eq_hat_to_tilde} and~\eqref{eq_upper_bound_w}, we obtain
\begin{align*}
   & E\left[\left\langle \widetilde f_U^{(1)} \otimes \widetilde f_U^{(2)}, w\right\rangle_{\widehat F_{Y}^{(1)}\otimes\widehat F_{Y}^{(2)}}\right]\\
   =&(1+O(\eps_{n_2}))\int_{\R^2}E\left[\, f_U\!\left(\widetilde F_Y^{(1)}(y)\right)\, f_U\!\left(\widetilde F_Y^{(2)}(y')\right) w\!\left(\widetilde F_Y^{(1)}(y), \widetilde F_Y^{(2)}(y')\right)\right]dydy'\\
   =& (1+O(\eps_{n_2}))E\left[\left\langle f_U^{\,\otimes \,2}, w\right\rangle_{\widetilde F_Y^{(1)}\otimes \widetilde F_Y^{(2)}}\right].
\end{align*}
Plugging this expression into the definition of $I_{1} = E\left[\left\langle \widetilde f_U^{(1)} \otimes \widetilde f_U^{(2)}, w\right\rangle_{\widehat F_{Y}^{(1)}\otimes\widehat F_{Y}^{(2)}}\right] - \left\langle f_U^{\,\otimes \,2}, w\right\rangle_{F_Y^{\,\otimes \,2}}$, we obtain
\begin{align}
I_{1} \leq \left|E\left[\left\langle f_U^{\,\otimes \,2}, w\right\rangle_{\widetilde F_Y^{(1)}\otimes \widetilde F_Y^{(2)}}\right] - \left\langle f_U^{\,\otimes \,2}, w\right\rangle_{F_Y^{\,\otimes \,2}}\right| + O\left(\varepsilon_{n_2}E\left[\left\langle f_U^{\,\otimes \,2}, w\right\rangle_{\widetilde F_Y^{(1)}\otimes \widetilde F_Y^{(2)}}\right]\right).\label{eq:decomposition_var_I1N}
\end{align}
We first prove that the second term tends to zero. To do so, we show that $E\left[\left\langle f_U^{\,\otimes \,2}, w\right\rangle_{\widetilde F_Y^{(1)}\otimes \widetilde F_Y^{(2)}}\right]$ is finite. We recall the definition of the function $g(s,t) = (s \land t)^{2b_1}(\bar s \land \bar t)^{2b_2}f_U(s) f_U(t)$, and, in the sequel, we write 
$$\Omega(s,t) = (s \land t)^{1-2b_1} (\bar s \land \bar t)^{1-2b_2}.$$ By Assumption~\ref{ass:smoothness1}\ref{ass:smoothness13}, there exists a positive constant $C_U$ such that $g \leq C_U^2$ over $(0,1)^2$. Since $\widetilde F_Y^{(j)}(y)-\widehat F_Y^{(j)}(y)\leq \widehat F_Y^{(j)}(y)\bar{\widehat F}_Y^{(j)}(y)/2\leq \widehat F_Y^{(j)}(y)$ by definition of $h_{n_2,\widehat F_Y^{(j)}(y)}$, and, similarly, $\bar{\widetilde F}_Y^{(j)}(y)-\bar{\widehat F}_Y^{(j)}(y)\leq \widehat F_Y^{(j)}(y)\bar{\widehat F}_Y^{(j)}(y)/2\leq \bar{\widehat F}_Y^{(j)}(y)$, it follows that
\begin{align*}
    & \quad  E\left[f_U\!\left(\widetilde F_Y^{(1)}(y)\right)\, f_U\!\left(\widetilde F_Y^{(2)}(y')\right) w\!\left(\widetilde F_Y^{(1)}(y), \widetilde F_Y^{(2)}(y')\right)\right] \\
    & = E\left[g(\widetilde F_Y^{(1)}(y), \widetilde F_Y^{(2)}(y')) \, \Omega\!\left(\widetilde F_Y^{(1)}(y), \widetilde F_Y^{(2)}(y')\right)\right] \\
    & \leq C_U^2E\left[\left(\widetilde F_Y^{(1)}(y)\wedge \widetilde F_Y^{(2)}(y')\right)^{1-2b_1}\left(\bar{\widetilde F}^{(1)}_Y(y)\wedge \bar{\widetilde F}^{(2)}_Y(y')\right)^{1-2b_2}\right] \\
    & \leq 2^{2(1-b_1-b_2)}C_U^2E\left[\left(\widehat F_Y^{(1)}(y)\wedge \widehat F_Y^{(2)}(y')\right)^{1-2b_1}\left(\bar{\widehat F}_Y^{(1)}(y)\wedge \bar{\widehat F}_Y^{(2)}(y')\right)^{1-2b_2}\right],
\end{align*}
which implies that $E\left[\left\langle f_U^{\,\otimes \,2}, w\right\rangle_{\widetilde F_Y^{(1)}\otimes \widetilde F_Y^{(2)}}\right]$ is finite by Lemma~\ref{lem:int_FY_FY'} since $1-2b_j > 2 d_j$ for $j \in \{1,2\}$ by Assumption~\ref{ass:smoothness1}\ref{ass:smoothness14}. 
It follows that 
\begin{equation*}
    \lim_{N\to\infty}O\left(\varepsilon_{n_2}E\left[\left\langle f_U^{\,\otimes \,2}, w\right\rangle_{\widetilde F_Y^{(1)}\otimes \widetilde F_Y^{(2)}}\right]\right)=0.
\end{equation*}
Next, we turn to the first term in \eqref{eq:decomposition_var_I1N}. 
It immediately follows that 
\begin{align*}
    \left\langle f_U^{\,\otimes \,2}, w\right\rangle_{\widetilde F_Y^{(1)}\otimes \widetilde F_Y^{(2)}} &= \int_{\R^2} \left(g \circ \widetilde F_Y^{(1)}\otimes\widetilde F_Y^{(2)}\right)(y,y') \left( \Omega \circ \widetilde F_Y^{(1)}\otimes\widetilde F_Y^{(2)}\right)(y,y') dy dy'\\[5pt]
    & =: \int_{\R^2} \left(g \circ \widetilde F_Y^{(1)}\otimes\widetilde F_Y^{(2)}\right) \left( \Omega \circ \widetilde F_Y^{(1)}\otimes\widetilde F_Y^{(2)}\right)
\end{align*} 
and similarly
\begin{align*}
    \left\langle f_U^{\,\otimes \,2}, w\right\rangle_{F_Y^{\,\otimes \,2}} & = \int_{\R^2} \left(g \circ F_Y^{\,\otimes \,2}\right) \left( \Omega \circ F_Y^{\,\otimes \,2}\right). 
\end{align*} 
By the triangle inequality and Assumption~\ref{ass:smoothness2}\ref{ass:smoothness23}, the first term in~\eqref{eq:decomposition_var_I1N} can be upper bounded as 
\begin{align*}
& \quad \left|\left\langle f_U^{\,\otimes \,2}, w\right\rangle_{\widetilde F_Y^{(1)}\otimes \widetilde F_Y^{(2)}} - \left\langle f_U^{\,\otimes \,2}, w\right\rangle_{ F_Y}\right|  \\[5pt]
& = \left|\int_{\R^2} \left(g \circ \widetilde F_Y^{(1)}\otimes\widetilde F_Y^{(2)}\right) \left( \Omega \circ \widetilde F_Y^{(1)}\otimes\widetilde F_Y^{(2)}\right)  - \int_{\R^2} \left(g \circ  F_Y^{\,\otimes \,2}\right) \left( \Omega \circ  F_Y^{\,\otimes \,2}\right) \right|\\[10pt]
&\leq 
\int_{\R^2} \Big|\!\left( g \circ \widetilde F_Y^{(1)}\otimes\widetilde F_Y^{(2)}\right) - \left(g \circ  F_Y^{\,\otimes \,2}\right)\!\Big| \!\left(\Omega \circ \widetilde F_Y^{(1)}\otimes\widetilde F_Y^{(2)}\right) \\[5pt]
&\quad + \int_{\R^2} \left(g \circ  F_Y^{\,\otimes \,2}\right) \! \Big|\!\left( \Omega \circ \widetilde F_Y^{(1)}\otimes\widetilde F_Y^{(2)}\right) - \left( \Omega \circ  F_Y^{\,\otimes \,2}\right)\!\Big|\\[10pt]
& \leq 
 c_U \int_{\R^2}\left(\left\vert \widetilde F_Y^{(1)}(y)-F_Y(y)\right\vert^\beta+\left\vert\widetilde F_Y^{(2)}(y')-F_Y(y')\right\vert^\beta\right)\Omega\left(\widetilde F_Y^{(1)}(y), \widetilde F_Y^{(2)}(y')\right)dydy'\\[5pt]
& \quad + 
\int_{\R^2} \left(g \circ  F_Y^{\,\otimes \,2}\right) \left|\left( \Omega \circ \widehat F_Y^{(1)}\otimes\widehat F_Y^{(2)}\right) - \left( \Omega \circ  F_Y^{\,\otimes \,2}\right)\right|\\[5pt]
& \quad + \int_{\R^2} \left(g \circ  F_Y^{\,\otimes \,2}\right) \left|\left( \Omega \circ \widetilde F_Y^{(1)}\otimes\widetilde F_Y^{(2)}\right) - \left( \Omega \circ \widehat F_Y^{(1)}\otimes\widehat F_Y^{(2)}\right)\right|. 
\end{align*} 
Given equation~\eqref{eq:control_tilde_hat}, Lemma~\ref{lem:useful_ineq} implies that
\begin{align*}
    \left|\left( \Omega \circ \widetilde F_Y^{(1)}\otimes\widetilde F_Y^{(2)}\right) - \left( \Omega \circ  \widehat F_Y^{(1)}\otimes\widehat F_Y^{(2)}\right)\right| = O(\varepsilon_{n_2}) \cdot \left(\Omega \circ \widehat F_Y^{(1)}\otimes\widehat F_Y^{(2)}\right).
\end{align*}
Therefore, by symmetry of $\Omega$, we obtain
\begin{align*}
& \quad E\left|\left\langle f_U^{\,\otimes \,2}, w\right\rangle_{\widetilde F_Y^{(1)}\otimes \widetilde F_Y^{(2)}} - \left\langle f_U^{\,\otimes \,2}, w\right\rangle_{ F_Y}\right| \\[5pt]
&\leq 
 2c_U E\left[\int_{\R^2} \left\vert \widetilde F_Y^{(1)}(y)-F_Y(y)\right\vert^\beta \Omega\left(\widehat F_Y^{(1)}(y), \widehat F_Y^{(2)}(y')\right)dydy' \right] \\[5pt]
& \quad + E\int_{\R^2} \left(g \circ  F_Y^{\,\otimes \,2}\right) \left|\left( \Omega \circ \widehat F_Y^{(1)}\otimes\widehat F_Y^{(2)}\right) - \left( \Omega \circ F_Y^{\,\otimes \,2}\right)\right| \\[5pt]
& \quad + O(\varepsilon_{n_2}) E\Big[\int_{\R^2} \left\vert \widetilde F_Y^{(1)}(y)-F_Y(y)\right\vert^\beta\Omega\left(\widehat F_Y^{(1)}(y), \widehat F_Y^{(2)}(y')\right) dydy' \\[5pt]
&\qquad\qquad\quad + \int_{\R^2} \left(g \circ F_Y^{\,\otimes \,2}\right)\left(\Omega \circ \widehat F_Y^{(1)}\otimes\widehat F_Y^{(2)}\right)(y,y')dydy'\Big]\\
& = J_1 + J_2 + J_3.
\end{align*} 
Lemma~\ref{lem:21} in Supplementary Appendix Section~\ref{sec:S2} establishes that $J_1,J_2,J_3$ are $o(1)$.

\medskip
\noindent {\bf Step 1.B (Vanishing variance): $E[I_2^2]=o(1)$.} 
Let
\begin{align*}
    T & := \left\langle \widetilde f_U^{(1)} \otimes \widetilde f_U^{(2)}, w\right\rangle_{\widehat F_{Y}^{(1)}\otimes\widehat F_{Y}^{(2)}} \\
    & =\int_{\R^2}\widetilde f_U^{(1)}\big(\widehat F_Y^{(1)}(y)\big)\widetilde f_U^{(2)}\big(\widehat F_Y^{(2)}(y')\big)\big[\widehat F_Y^{(1)}(y)\wedge \widehat F_Y^{(2)}(y')\big]\big[\bar{\widehat F}_Y^{(1)}(y)\wedge \bar{\widehat F}_Y^{(2)}(y')\big]dydy'
\end{align*}
denote the random part of $I_2$.
Note that 
\begin{align*}
E[I_2^2]&=V\left[T\right]\\
&=V[E[T|(Y_i)_{i=1}^{n_1}]]+E[V[T|(Y_i)_{i=1}^{n_1}]]\\
&=:T_{1}+T_{2}.
\end{align*}

\medskip
\noindent {\bf Step 1.B.1: $T_{1}=o(1)$.} By the same steps as above, we have
\[
E[T|(Y_i)_{i=1}^{n_1}] = (1+O(\eps_{n_2}))\int_{\R^2}f_U\left(\widetilde F_Y^{(1)}(y)\right)f_U\left(\widetilde F_Y^{(2)}(y')\right)w\left(\widetilde F_Y^{(1)}(y),\widetilde F_Y^{(2)}(y')\right)dydy'.
\]
It follows that 
\begin{align*}
    T_1 &= V\left[E[T|(Y_i)_{i=1}^{n_1}]\right] \\
    &\leq 2 V\left[\int_{\R^2}f_U\left(\widetilde F_Y^{(1)}(y)\right)f_U\left(\widetilde F_Y^{(2)}(y')\right)w\left(\widetilde F_Y^{(1)}(y),\widetilde F_Y^{(2)}(y')\right)dydy'\right]\\
    & \quad + O(\varepsilon_{n_2}^2) E\left[\left(\int_{\R^2}f_U\left(\widetilde F_Y^{(1)}(y)\right)f_U\left(\widetilde F_Y^{(2)}(y')\right) w\left(\widetilde F_Y^{(1)}(y),\widetilde F_Y^{(2)}(y')\right)dydy'\right)^2\right].
\end{align*}
The second term can be analyzed as follows. By the same steps as above, we have
\begin{align*}
    & \quad E\left[\left(\int_{\R^2}f_U\left(\widetilde F_Y^{(1)}(y)\right)f_U\left(\widetilde F_Y^{(2)}(y')\right) w\left(\widetilde F_Y^{(1)}(y),\widetilde F_Y^{(2)}(y')\right)dydy'\right)^2\right]\\
    & \leq C_U^22^{2(2-b_1-b_2)}E\left[\int_{\R^2}\left( \left(\widehat F_Y^{(1)}(y) \land \widehat F_Y^{(2)}(y')\right)^{1-2b_1} \left(\bar{\widehat{F}}_Y^{(1)}(y) \land \bar{\widehat{F}}_Y^{(2)}(y')\right)^{1-2b_2} dydy'\right)^2\right],
\end{align*}
which is finite by Lemma~\ref{lem:int_FY_FY'}, since $1-2b_j>d_j$ for $j \in \{1,2\}$ by Assumption~\ref{ass:smoothness1}\ref{ass:smoothness14}. This ensures that the second term is $o(1)$. 

We now show that the first term $2V\left[\int_{\R^2}f_U\left(\widetilde F_Y^{(1)}(y)\right)f_U\left(\widetilde F_Y^{(2)}(y')\right)\Omega\left(\widetilde F_Y^{(1)}(y),\widetilde F_Y^{(2)}(y')\right)dydy'\right]$ tends to zero as $N\to\infty$. 
Let $\left(\widetilde{\widetilde F}_Y^{(j)}\right)_{j\in\{1,2\}}$ be  independent copies of $\left(\widetilde F_Y^{(j)}\right)_{j\in\{1,2\}}$ computed on an independent copy $(Y'_{1}, \dots, Y'_{n_1})$ of $(Y_1,\dots, Y_{n_1})$.
We have
\begin{align*}
    &\quad V\left[\int_{\R^2}f_U\left(\widetilde F_Y^{(1)}(y)\right)f_U\left(\widetilde F_Y^{(2)}(y')\right)w\left(\widetilde F_Y^{(1)}(y),\widetilde F_Y^{(2)}(y')\right)dydy'\right] \\[5pt]
    & \leq E\left[\left(\int_{\R^2}f_U\!\left(\widetilde F_Y^{(1)}(y)\right)f_U\!\left(\widetilde F_Y^{(2)}(y')\right)w\left(\widetilde F_Y^{(1)}(y),\widetilde F_Y^{(2)}(y')\right) \right.\right.\\[5pt]
    & \hspace{2cm} \left.\left.- f_U\!\left(\widetilde{\widetilde F}_Y^{(1)}(y)\right)f_U\!\left(\widetilde{\widetilde F}_Y^{(2)}(y')\right) w\left(\widetilde{\widetilde F}_Y^{(1)}(y),\widetilde{\widetilde F}_Y^{(2)}(y')\right)dydy'\right)^2\right]\\[5pt]
    & = E\Bigg[\bigg(\int_{\R^2}g\left(\widetilde F_Y^{(1)}(y), \widetilde F_Y^{(2)}(y')\right)\Omega\left(\widetilde F_Y^{(1)}(y),\widetilde F_Y^{(2)}(y')\right) \\[5pt]
    & \qquad - g\left(\widetilde{\widetilde F}_Y^{(1)}(y), \widetilde{\widetilde F}_Y^{(2)}(y')\right)\Omega\left(\widetilde{\widetilde F}_Y^{(1)}(y),\widetilde{\widetilde F}_Y^{(2)}(y')\right)dydy'\bigg)^2\Bigg]\\[5pt]
    & \leq 2 E \left[\left(\int_{\R^2} \Big|\,g\!\left(\widetilde F_Y^{(1)}(y), \widetilde F_Y^{(2)}(y')\right) - g\!\left(\widetilde{\widetilde F}_Y^{(1)}(y), \widetilde{\widetilde F}_Y^{(2)}(y')\right)\Big| \,\Omega\left(\widetilde F_Y^{(1)}(y), \widetilde F_Y^{(2)}(y')\right)dydy'\right)^2\right]\\[5pt]
    & \quad + 2 E\left[\left(\int_{\R^2} g\!\left(\widetilde{\widetilde F}_Y^{(1)}(y), \widetilde{\widetilde F}_Y^{(2)}(y')\right) \Big|\Omega\left(\widetilde F_Y^{(1)}(y), \widetilde F_Y^{(2)}(y')\right) - \Omega\left(\widetilde{\widetilde F}_Y^{(1)}(y),\widetilde{\widetilde F}_Y^{(2)}(y')\right)\Big|dydy'\right)^2\right]\\[5pt]
    & \leq 2 E \Bigg[\bigg(\int_{\R^2} c_U\left(h_{\widehat F_Y^{(1)}(y)}^{\beta}+h_{\widehat{\widehat F}_Y^{(1)}(y)}^{\beta} + h_{\widehat F_Y^{(2)}(y')}^{\beta}+h_{\widehat{\widehat F}_Y^{(2)}(y')}^{\beta}+\Big| \widehat F_Y^{(1)}(y)-\widehat {\widehat F}_Y^{(1)}(y)\Big|^\beta+\Big| \widehat F_Y^{(2)}(y')-\widehat{\widehat F}_Y^{(2)}(y')\Big|^\beta\right) \\
    &\qquad \times \Omega\left(\widetilde F_Y^{(1)}(y), \widetilde F_Y^{(2)}(y')\right) dydy'\bigg)^2\Bigg]\\[5pt]
    & \quad + 2 C_U^2 E\left[\left(\int_{\R^2} \Big|\Omega\left(\widetilde F_Y^{(1)}(y), \widetilde F_Y^{(2)}(y')\right) - \Omega\left(\widetilde{\widetilde F}_Y^{(1)}(y),\widetilde{\widetilde F}_Y^{(2)}(y')\right)\Big|dydy'\right)^2\right] \\[5pt]
    & =: J_4+J_5
\end{align*}
Lemma~\ref{lem:21} in Supplementary Appendix Section~\ref{sec:S2} establishes that $J_4$ and $J_5$ are $o(1)$.

\medskip
\noindent {\bf Step 1.B.2: $T_2=o(1)$.} 
In this part of the analysis, we fix $\mathbf{Y} = (Y_1,\dots, Y_{n_1})$ and consider the term $V[T|\mathbf Y]$. Conditionally on $\mathbf{Y}$, we define the measure $\widehat \mu_{\mathbf{Y}}$ over $[0,1]^2$ by
\begin{align*}
    d\hspace{.5mm}\widehat \mu_{\mathbf{Y}}(s,t) = (s\land t) \, (\bar s \land \bar t) \, d(\widehat F_Y^{(1)})^{-1}(s) \, d(\widehat F_Y^{(2)})^{-1}(t).
\end{align*}
It follows that, conditionally on $\mathbf{Y}$, the random variable $T$ can be rewritten as
\begin{align*}
    T| \mathbf Y = \int_{[0,1]^2} \widetilde f_U^{(1)}(s) \, \widetilde f_U^{(2)}(t) \; d\widehat \mu_{\mathbf{Y}}(s,t).
\end{align*}
Writing $e_{\mathbf Y}(s) = E\big[\widetilde f_U^{(1)}(s) \, | \, \mathbf Y\big]=E\big[\widetilde f_U^{(1)}(s)\big]$, we obtain
\begin{align*}
    V[T|\mathbf Y] &= E\Big[\big(T - E\big[T| \mathbf Y\big]\big)^2 \, \big| \, \mathbf Y\Big]\\[5pt]
    & = E\Big[\Big(\int_{[0,1]^2} \widehat f_U^{\,(1)}(s) \, \widehat f_U^{\,(2)}(t) - E \big[\widehat f_U^{\,(1)}(s) \, \widehat f_U^{\,(2)}(t) \,\big|\, \mathbf Y\big] \; d\widehat \mu_{\mathbf{Y}}(s,t)\Big)^2 \, \big| \, \mathbf Y\Big]\\[5pt]
    & = E\Big[\Big(\int_{[0,1]^2} \widehat f_U^{\,(1)}(s) \, \widehat f_U^{\,(2)}(t) - e_{\mathbf Y}(s) \, e_{\mathbf Y}(t) \; d\widehat \mu_{\mathbf{Y}}(s,t)\Big)^2 \, \big| \, \mathbf Y\Big]\\[5pt]
    & = \int_{[0,1]^4} E\left[\Big(\widehat f_U^{\,(1)}(s) \, \widehat f_U^{\,(2)}(t) - e_{\mathbf Y}(s) \, e_{\mathbf Y}(t) \Big)\Big(\widehat f_U^{\,(1)}(s') \, \widehat f_U^{\,(2)}(t') - e_{\mathbf Y}(s') \, e_{\mathbf Y}(t') \Big) \,\big|\,\mathbf Y\right] d\widehat\mu_{\mathbf Y}^{\, \otimes 2}(s,t,s',t')\\[5pt]
    & = \int_{[0,1]^4} E\left[\widehat f_U^{\,(1)}(s)\widehat f_U^{\,(1)}(s')\,\big|\,\mathbf Y\right] E\left[\widehat f_U^{\,(1)}(t)\widehat f_U^{\,(1)}(t')\,\big|\,\mathbf Y\right] d\widehat\mu_{\mathbf Y}^{\, \otimes 2}(s,t,s',t')\\[5pt]
    & \quad - \left(\int_{[0,1]^2} e_{\mathbf Y}(s) e_{\mathbf Y}(t)d\widehat\mu_{\mathbf Y}(s,t)\right)^2\\[5pt]
    & = \int_{[0,1]^4} \Big\{E\left[\widehat f_U^{\,(1)}(s)\widehat f_U^{\,(1)}(s')\,\big|\,\mathbf Y\right] E\left[\widehat f_U^{\,(1)}(t)\widehat f_U^{\,(1)}(t')\,\big|\,\mathbf Y\right] \\
    &\qquad -E\left[\widehat f_U^{\,(1)}(s)\widehat f_U^{\,(1)}(t)\,\big|\,\mathbf Y\right] E\left[\widehat f_U^{\,(1)}(s')\widehat f_U^{\,(1)}(t')\,\big|\,\mathbf Y\right] \Big\}d\widehat\mu_{\mathbf Y}^{\, \otimes 2}(s,t,s',t')  \\[5pt]
    &\quad + \left(\int_{[0,1]^2} \Big(E\left[\widehat f_U^{\,(1)}(s)\widehat f_U^{\,(1)}(t)\,\big|\,\mathbf Y\right] - e_{\mathbf Y}(s) e_{\mathbf Y}(t)\Big) d\widehat\mu_{\mathbf Y}(s,t) \right)^2 \\[5pt]
    &  \quad + 2 \int_{[0,1]^2} \Big(E\left[\widehat f_U^{\,(1)}(s)\widehat f_U^{\,(1)}(t)\,\big|\,\mathbf Y\right] - e_{\mathbf Y}(s) e_{\mathbf Y}(t)\Big) d\widehat\mu_{\mathbf Y}(s,t) \cdot \int_{[0,1]^2} e_{\mathbf Y}(s') e_{\mathbf Y}(t') d\widehat\mu_{\mathbf Y}(s',t')\\[5pt]
    & = \int_{[0,1]^4} \Big\{E\left[\widehat f_U^{\,(1)}(s)\widehat f_U^{\,(1)}(s')\,\big|\,\mathbf Y\right] E\left[\widehat f_U^{\,(1)}(t)\widehat f_U^{\,(1)}(t')\,\big|\,\mathbf Y\right] \\[5pt]
    &\qquad -E\left[\widehat f_U^{\,(1)}(s)\widehat f_U^{\,(1)}(t)\,\big|\,\mathbf Y\right] E\left[\widehat f_U^{\,(1)}(s')\widehat f_U^{\,(1)}(t')\,\big|\,\mathbf Y\right] \Big\}d\widehat\mu_{\mathbf Y}^{\, \otimes 2}(s,t,s',t') \\[5pt]
    & \quad +\left(\int_{[0,1]^2} \text{cov}\big(\widehat f_U^{\,(1)}(s), \widehat f_U^{\,(1)}(t)\big) d\widehat\mu_{\mathbf Y}(s,t) \right)^2\\[5pt]
    & \quad + 2 \int_{[0,1]^2} \text{cov}\big(\widehat f_U^{\,(1)}(s), \widehat f_U^{\,(1)}(t)\big) d\widehat\mu_{\mathbf Y}(s,t) \cdot \int_{[0,1]^2} e_{\mathbf Y}(s') e_{\mathbf Y}(t') d\widehat\mu_{\mathbf Y}(s',t') \\[5pt]
    & =: J_6+J_7+J_8.
\end{align*}
To show that $T_2=o(1)$, it suffices to show that $E[\abs{J_6}], E[\abs{J_7}],E[\abs{J_8}]=o_P(1)$. This is established in Lemma~\ref{lem:22} in Supplementary Appendix Section~\ref{sec:S2}.

\medskip
\noindent{\bf Step 2: $\widetilde R_2= o_P(1)$.}
 Decompose $\widetilde R_2$ as follows
\begin{align*}
    \widetilde R_2 & = \int_0^1 \left(\left[\widehat F_Y^{-1}\right]^2 - [F_Y^{-1}]^2 \right) d\widehat F_U  + \int_0^1[F_Y^{-1}]^2 d\left[\widehat F_U  -  F_U \right] \\[5pt]
    & = \underbrace{\int_0^1 \left(\left[F_Y^{-1} \circ \Gn^{-1}\right]^2 - [F_Y^{-1}]^2 \right) d\widehat F_U}_{=:\widetilde R_{2,1}} + \underbrace{\int_0^1[F_Y^{-1}]^2 d\left[\widehat F_U  -  F_U \right]  }_{=:\widetilde R_{2,2}}.
\end{align*}
 By the law of large numbers, $\widetilde R_{2,2} = o_P(1)$. Turning to $\widetilde R_{2,1}$, we show the stronger result that $E[|\widetilde R_{2,1}|]\to 0$. We have
\begin{align}
  |\widetilde R_{2,1}| &  = \left \vert \int_0^1 \left(F_Y^{-1} \circ \Gn^{-1}  - F_Y^{-1}\right)\left( F_Y^{-1} \circ \Gn^{-1} + F_Y^{-1}\right) d\widehat F_U \right \vert \notag \\[5pt]
  & \leq \int_0^1 \left  \vert F_Y^{-1} \circ \Gn^{-1}  - F_Y^{-1}\right \vert \left \vert F_Y^{-1} \circ \Gn^{-1}+ F_Y^{-1}\right \vert d\widehat F_U \notag \\[5pt]
      &  \leq  \left[\max_{i=1,\ldots,n_1}|Y_i| + \max_{i=1,\ldots,n_2} |F_Y^{-1}(U_i)|\right] \int_0^1 \left  \vert F_Y^{-1} \circ \Gn^{-1}  - F_Y^{-1}\right \vert d\widehat F_U. \label{eq:ineq_var1}
\end{align}
Assumption~\ref{ass:smoothness1}\ref{ass:smoothness14} implies that there exists $\nu>0$ such that $b_i + (2+\nu)d_i<1$ for $i \in \{1,2\}$ and $(2+\nu)(d_1 \vee d_2)<1$. From this, it is easy to obtain $E[|F_Y^{-1}(U)|^{2+\nu}]<\infty$ and $E[|Y|^{2+\nu}]<\infty$, which in turn yields that $\Pr(|F_Y^{-1}(U)|>x) = o(x^{-(2+\nu)})$ and $\Pr(|Y|>x) = o(x^{-(2+\nu)})$. In particular,  we have by, e.g., Exercise 2.3.4 in \cite{vaartwellner1996}, $\max_{i=1,...,n_2}|F_Y^{-1}(U_i)| = o_P(n_2^{\frac{1}{2+\nu}})$ and $\max_{i=1,...,n_1}|Y_i| = o_P(n_1^{\frac{1}{2+\nu}})$.\footnote{Note that Lemma~\ref{lem:extrem_moments1} and Assumption~\ref{ass:smoothness1} already imply that $\max_{i=1,...,n_1}|Y_i|=O_P\left(n_1^{d_1\lor d_2}\right)=o_P\left(n_1^{1/2-\gamma}\right)$ for $\gamma$ small enough.} Hence, to show that $\widetilde R_{2,1}=o_P(1)$, it is sufficient to show that the integral in \eqref{eq:ineq_var1} is $O_P((n_1\vee n_2)^{-\frac{1}{2+\nu}})$. We show that it is $O_P((n_1 \vee n_2)^{-\frac{1}{2}})$. Consider the decomposition
\begin{align*}
    & \sqrt{n_1 \vee n_2} \int_0^1  \left \vert F_Y^{-1} \circ \Gn^{-1}  - F_Y^{-1}\right \vert d\widehat F_U 
     \\[5pt]
     & \qquad \lesssim  \sqrt{n_1} \int_0^1  \left \vert F_Y^{-1} \circ \Gn^{-1}  - F_Y^{-1}\right \vert d F_U + \abs{\sqrt{n_1} \int_0^1   F_Y^{-1} \circ \Gn^{-1} d[\widehat F_U-F_U]} \\[5pt]
     &\qquad\qquad  + \abs{\sqrt{n_1} \int_0^1 F_Y^{-1} d[\widehat F_U-F_U]}.
     \end{align*}
     Since $\int_0^1\abs{F_Y^{-1}}^2dF_U<\infty$, the central limit theorem implies that the third term is $O_P(1)$. The second term is bounded by $\sqrt{n_1}\abs{T_2}\lesssim \abs{\sqrt{N}T_2}$, where $\sqrt N T_2$ is defined and shown to be $O_P(1)$ in the proof of Theorem~\ref{thm:an}. Hence, the second term is $O_P(1)$. As for the first term, note that
    \begin{align*}
   &  \sqrt{n_1} \int_0^1  \left \vert F_Y^{-1} \circ \Gn^{-1}  - F_Y^{-1}\right \vert d F_U  \\[5pt]
   &\quad = \sqrt{n_1} \int_{1/n_1}^{1-1/n_1}  \left \vert F_Y^{-1} \circ \Gn^{-1}  - F_Y^{-1}\right \vert d F_U +o_P(1)  \\[5pt]
    &  \quad\leq  \sqrt{n_1} \int_{1/n_1}^{1-1/n_1}\abs{F_Y^{-1}{}'(\tilde t_{n_1})} \left \vert \Gn^{-1}(t)-t\right \vert d F_U(t),
  \end{align*}
for some $\tilde t_{n_1}(t) \in [\Gn^{-1}(t)\wedge t, \Gn^{-1}(t)\vee t]$.  The upper bound can be shown to be $o_P(1)$ by arguments similar to those developed in Step 3 of the proof of Theorem~\ref{thm:an}.

\paragraph{Second step: $\widehat\sigma^2-\widetilde \sigma^2=o_P(1)$.} 
Define
\begin{align*}
      R_3 & := \left\langle (\widehat f_U^{(1)} \otimes \widehat f_U^{(2)})-(\widetilde f_U^{(1)} \otimes \widetilde f_U^{(2)}), w\right\rangle_{\widehat F_{Y}^{(1)}\otimes\widehat F_{Y}^{(2)}},\\[5pt]
     R_4 & := \frac{1}{n_2} \sum_{i=1}^{n_2}\widehat F_Y^{-1}(\widehat U_i)^2-\frac{1}{n_2} \sum_{i=1}^{n_2}\widehat F_Y^{-1}(U_i)^2.
\end{align*}
It is sufficient to show that $R_{3}=o_P(1)$ and $R_4=o_p(1)$. In Step 1 below, we show that $R_3= o_P(1)$. In Step 2, we show that $R_4 = o_P(1)$.

To alleviate the notation, we will write $h_{t} := h_{n_2,t}$ and $n = n_2$ throughout the proof, and we define $ E_{\mathbf{Y}}[\cdot] = E\big[\hspace{.2mm} \cdot \hspace{.2mm}|\, (Y_i)_{i=1}^{n_1}\big]$.

\medskip
\noindent{\bf Step 1: $R_3= o_P(1)$.} We introduce the high-probability event 
\begin{align*}
    \mathcal{A}_n = \mathcal{A}_n^{(1)} \cap \mathcal{A}_n^{(2)}  
\end{align*}
where 
\begin{align*}
    \mathcal{A}_n^{(j)} = \left\{\|\widehat F_Z^{(j)} - F_Z\|_{\infty} \leq \frac{a_n}{\sqrt{n}}\right\}
\end{align*}
for $j=1,2$. Note that $\mathcal{A}_n$  satisfies $P(\mathcal{A}_n^c) = o(1)$ for any $a_n\to \infty$ by the Glivenko--Cantelli lemma. 
In what follows, we assume that $(a_n/\sqrt{n})_{n \geq 1}$ is a decreasing sequence converging to zero such that $a_n = O(\log(n))$ and $a_1$ is sufficiently small in a sense that will become clear later on. 
In the following, we show that $E [\vert  R_3\mathds 1\{\mathcal A_n\}\vert] = o(1)$, which implies that $R_3 = o_P(1)$ by the Markov inequality, as desired.

By the triangle inequality, we have
\begin{align*}
     |R_3|&\leq \int_{\R^2}\frac{w(\widehat F_Y^{(1)}(y),\widehat F_Y^{(2)}(y'))}{n^{2}h_{\widehat F_Y^{(1)}(y)}h_{\widehat F_Y^{(2)}(y')}} \\[5pt]
     &\qquad \times \sum_{i=1}^{n/2} \sum_{j=n/2+1}^{n}\Bigg\vert \mathds 1\left\{\abs{\widehat U_i^{(1)}-\widehat F_Y^{(1)}(y)}\leq h_{\widehat F_Y^{(1)}(y)}, \abs{\widehat U_j^{(2)}-\widehat F_Y^{(2)}(y')}\leq h_{\widehat F_Y^{(2)}(y')}\right\} \\[5pt]
     &\hspace{3cm} - \mathds 1\left\{\abs{ U_i-\widehat F_Y^{(1)}(y)}\leq h_{\widehat F_Y^{(1)}(y)}, \abs{ U_j-\widehat F_Y^{(2)}(y')}\leq h_{\widehat F_Y^{(2)}(y')}\right\}\Bigg\vert dydy'.
\end{align*}
We define $\widehat U^{(1)} = \widehat F_Z^{(1)}(X)$ and $\widehat U^{(2)} = \widehat F_Z^{(2)}(X')$ where $X, X'$ are i.i.d.~with cdf $F_X$. 
We note that the random variable $\widehat U^{(1)}$ (resp.~$\widehat U^{(2)}$) has the same distribution as the random variables $U_i^{(1)}$ for any $i = 1,\dots, n_2/2$ (resp.~$U_j^{(2)}$ for any $j = n_2/2 + 1, \dots n_2$) conditional on $(Z_1,\dots, Z_{n_3/2})$ (resp.~$(Z_{n_3/2+1},\dots, Z_{n_3})$). 
Now, taking the expectation conditional on $(Y_i)_{i=1}^{n_1}$, we obtain
\begin{align*}
   E_{\mathbf Y}\big[R_3\mathds 1\{\mathcal A_n\}\big]
   & \leq \int_{\R^2}\frac{w(\widehat F_Y^{(1)}(y),\widehat F_Y^{(2)}(y'))}{h_{\widehat F_Y^{(1)}(y)}h_{\widehat F_Y^{(2)}(y')}}  \\[5pt]
   &\qquad \times E_{\mathbf Y}\Bigg[\bigg\vert \mathds 1\left\{\abs{\widehat U^{(1)}-\widehat F_Y^{(1)}(y)}\leq h_{\widehat F_Y^{(1)}(y)}, \abs{\widehat U^{(2)}-\widehat F_Y^{(2)}(y')}\leq h_{\widehat F_Y^{(2)}(y')}\right\} \\[5pt]
     &\hspace{1cm} - \mathds 1\left\{\abs{ U_1-\widehat F_Y^{(1)}(y)}\leq h_{\widehat F_Y^{(1)}(y)}, \abs{ U_2-\widehat F_Y^{(2)}(y')}\leq h_{\widehat F_Y^{(2)}(y')}\right\}\bigg\vert\times\mathds 1\{\mathcal A_n\}\Bigg]dydy'.
\end{align*}
By the inequality $|ab - \hat a \hat b| \leq |a - \hat a| \,b + |b- \hat b|\,a + |a - \hat a||b-\hat b|$ that holds for any $a,b,\hat a, \hat b \in \R$, we have
\begin{align*}
     E_{\mathbf{Y}}\big[R_3\mathds 1\{\mathcal A_n\}\big] \leq \int_{\R^2}\frac{w(\widehat F_Y^{(1)}(y),\widehat F_Y^{(2)}(y'))}{h_{\widehat F_Y^{(1)}(y)}h_{\widehat F_Y^{(2)}(y')}} \big[I(y,y')+II(y,y')+III(y,y')\big]dydy',
\end{align*}
where 
\begin{align*}
    I(y,y')&=E_{\mathbf Y}\Big[\Big\vert \mathds 1\left\{\abs{\widehat U^{(1)}-\widehat F_Y^{(1)}(y)}\leq h_{\widehat F_Y^{(1)}(y)}\right\}-\mathds 1\left\{\abs{ U_1-\widehat F_Y^{(1)}(y)}\leq h_{\widehat F_Y^{(1)}(y)}\right\}\\[5pt]
    & \hspace{5cm} \times \mathds 1\left\{\abs{U_2-\widehat F_Y^{(2)}(y')}\leq h_{\widehat F_Y^{(2)}(y')}\right\} \Big\vert\times\mathds 1\{\mathcal A_n^{(1)}\}\Big], \\[5pt]
    II(y,y')&=E_{\mathbf Y}\Big[\Big\vert \mathds 1\left\{\abs{\widehat U^{(2)}-\widehat F_Y^{(2)}(y')}\leq h_{\widehat F_Y^{(2)}(y')}\right\}-\mathds 1\left\{\abs{ U_2-\widehat F_Y^{(2)}(y')}\leq h_{\widehat F_Y^{(2)}(y')}\right\}\\[5pt]
    & \hspace{5cm} \times\mathds 1\left\{\abs{U_1-\widehat F_Y^{(1)}(y)}\leq h_{\widehat F_Y^{(1)}(y)}\right\}\Big\vert\times\mathds 1\{\mathcal A_n^{(2)}\}\Big], \\[5pt]
    III(y,y')&=E_{\mathbf Y}\Big[\Big\vert \mathds 1\left\{\abs{\widehat U^{(1)}-\widehat F_Y^{(1)}(y)\,}\leq \, h_{\widehat F_Y^{(1)}(y)}\right\}-\mathds 1\left\{\abs{ U_1-\widehat F_Y^{(1)}(y)}\leq h_{\widehat F_Y^{(1)}(y)}\right\}\Big\vert \\[5pt]
    &\qquad \times\Big\vert\mathds 1\left\{\abs{\widehat U^{(2)}-\widehat F_Y^{(2)}(y')}\leq h_{\widehat F_Y^{(2)}(y')}\right\}-\mathds 1\left\{\abs{ U_2-\widehat F_Y^{(2)}(y')}\leq h_{\widehat F_Y^{(2)}(y')}\right\}\Big\vert\times\mathds 1\{\mathcal A_n\}\Big].
\end{align*}
Lemma~\ref{lem:25} in Supplementary Appendix Section~\ref{sec:S2} establishes that each of these integrals converges to zero.

\medskip
\noindent{\bf Step 2: $R_4= o_P(1)$.} Consider the decomposition:
\begin{align*}
    R_4&=\int_0^1\left([F_Y^{-1}\circ \Gn^{-1}\circ\Hn]^2-[F_Y^{-1}\circ\Gn^{-1}]^2\right)d\widehat F_U +\widetilde R_{2N,1}+\widetilde R_{2N,2} \\[5pt]
    &=\int_0^1\left([F_Y^{-1}\circ \Gn^{-1}\circ\Hn]^2-[F_Y^{-1}\circ\Gn^{-1}]^2\right)d\widehat F_U+o_P(1) \\
    &=:R_5+o_P(1).
\end{align*}
We show below that $R_5:=o_P(1)$. Note that
\begin{align*}
    \abs{R_5}&=\abs{\int_0^1\left([F_Y^{-1}\circ \Gn^{-1}\circ\Hn]^2-[F_Y^{-1}\circ\Gn^{-1}]^2\right)d\widehat F_U}\\[5pt]
    &=\abs{\int_0^1\left(F_Y^{-1}\circ \Gn^{-1}\circ\Hn-F_Y^{-1}\circ\Gn^{-1}\right)\left(F_Y^{-1}\circ \Gn^{-1}\circ\Hn+F_Y^{-1}\circ\Gn^{-1}\right)d\widehat F_U}\\[5pt]
    &\leq 2 \max_{i=1,\ldots,n_1}|Y_i| \int_0^1 \left  \vert F_Y^{-1} \circ \Gn^{-1}\circ\Hn  - F_Y^{-1}\circ\Gn^{-1}\right \vert d\widehat F_U \\[5pt]
    &= 2\max_{i=1,\ldots,n_1}|Y_i|T_3,\\[5pt]
    &=o_P\left(n_1^{1/2}\right)T_3 \\[5pt]
    &=o_P\left(n_1^{1/2}\right)O_P\left(n_1^{-1/2}\right),
\end{align*}
where the last line follows from Lemma~\ref{lem:extrem_moments1}, and $T_3$ is defined and shown to be $O_P(N^{-1/2})=O_P(n_1^{-1/2})$ in the proof of Theorem~1. The result follows. 
\hfill$\square$
\section{Technical lemmas}\label{sec:tech_lemmas}

In Theorems~\ref{thm:an}--\ref{thm:var_est}, we use the following lemma, which is established in Proposition~1 of \cite{FALKNER2012412}.

\begin{lemma}[\cite{FALKNER2012412}'s Proposition 1]\label{lem:stieltjes_subs}
Let $M:[a,b]\to\R$ be increasing and $f:[a,b]\to\R$ be a bounded Borel function. Let $N: [M(a), M(b)]\to\R$ be increasing and right-continuous. We have
$$\int_a^bf(x)dN(M(x))=\int_{M(a)}^{M(b)}f(M^{-1}(y))dN(y).$$
\end{lemma}

\begin{lemma}~\label{lem:extrem_moments1}
Suppose that Assumption~\ref{ass:smoothness1}\ref{ass:smoothness12} holds. Then,
\begin{align*}
E\left[\left\vert Y_{(1)}\right\vert\right]=O\left(n_1^{d_1}\right), \quad E\left[\left\vert Y_{(1)}\right\vert^2\right]=O\left(n_1^{2d_1}\right), \quad   E\left[\left\vert Y_{(n_1)}\right\vert\right]=O\left(n_1^{d_2}\right), \quad  E\left[\left\vert Y_{(n_1)}\right\vert^2\right]=O\left(n_1^{2d_2}\right).
\end{align*}
\end{lemma}
\noindent\textbf{Proof:} We prove only the first result, as the second, third and fourth are analogous. We recall that the probability density function of the minimum of an i.i.d. $n_1$-sample of uniformly distributed random variables on $[0,1]$ is $f_{\xi_{(1)}}(u)=n_1(1-u)^{n_1-1}\mathbf{1}\{u\in[0,1]\}$. We have
\begin{align*}
    E\left[\left\vert Y_{(1)}\right\vert\right] & =  E\left[\left\vert F_Y^{-1}(\xi_{(1)}) \right\vert\right] \\[5pt]
    & \lesssim n_1\int_0^1u^{-d_1}(1-u)^{-d_2}(1-u)^{n_1-1}du \\[5pt]
    & = n_1\int_0^1v^{n_1-1}v^{-d_2}(1-v)^{-d_1}dv \\[5pt]
    &= n_1E\left[\left(\text{Beta}(1-d_2,1-d_1)\right)^{n_1-1}\right] \\[5pt]
    &=n_1 \frac{\Gamma(n_1-d_2)\Gamma(2-d_1-d_2)}{\Gamma(1-d_2)\Gamma(1-d_1+n_1-d_2)} \\[5pt]
    &= n_1O\left(n_1^{-(1-d_1)}\right) \\[5pt]
    &=O\left(n_1^{d_1}\right).
\end{align*}
\hfill $\square$

\begin{lemma}~\label{lem:extrem_moments2}
Suppose that Assumption~\ref{ass:smoothness1}\ref{ass:smoothness13} hold. Then,
\begin{align*}
E\left[F_U\big(\xi_{(1)}\big)^2\right]=O\left(n_1^{2(1-b_1)}\right) \quad \text{ and } \quad E\left[\left(1-F_U\big(\xi_{(n_1)}\big)\right)^2\right]=O\left(n_1^{2(1-b_2)}\right).
\end{align*}
\end{lemma}
\noindent\textbf{Proof:} We prove only the first result, as the second is analogous.
\begin{align*}
    E\left[F_U\big(\xi_{(1)}\big)^2\right]& = n_1\int_0^1 F_U(u)^2(1-u)^{n_1-1}du \\[5pt]
    & =n_1\int_0^1\left(\int_0^uf_U(t)dt\right)^2(1-u)^{n_1-1}du \\[5pt]
    & = n_1\int_0^1\int_0^1\left(\int_{s\lor t}^1(1-u)^{n_1-1}du\right)f_U(s)f_U(t)dsdt
\end{align*}
where the last equality follows by Fubini--Tonelli’s theorem. Then,
\begin{align*}
    E\left[F_U\big(\xi_{(1)}\big)^2\right]& = n_1\int_0^1\int_0^1\left[\frac{-(1-u)^{n_1}}{n_1}\right]^{u=1}_{u={s\lor t}}f_U(s)f_U(t)dsdt \\[5pt]
    &= \int_0^1\int_0^1(1-s\lor t)^{n_1}f_U(s)f_U(t)dsdt \\[5pt]
    &\leq \int_0^1\int_0^1(1-s)^{n_1/2}(1-t)^{n_1/2}f_U(s)f_U(t)dsdt \\[5pt]
    &=\left(\int_0^1(1-u)^{n_1/2}f_U(u)du\right) \\[5pt]
    &=\left(E\left[\text{Beta}\big(1-b_2,1-b_1\big)^{n_1/2}\right]\right)^2 \\[5pt]
    &\lesssim n_1^{2(1-b_1)},
\end{align*}
where the first inequality follows from $(1-s\lor t)^{n_1}\leq(1-s)^{n_1/2}(1-t)^{n_1/2}$ for all $s,t\in(0,1)$.
\hfill $\square$

\begin{lemma} \label{lem:int_FY}
Suppose that Assumption~\ref{ass:smoothness1}\ref{ass:smoothness12} holds and that $a_1 > d_1$ and $a_2 > d_2$, then $\int_0^1 x^{a_1}(1-x)^{a_2} \, dF_Y^{-1}(x)<\infty$.
\end{lemma}

\noindent\textbf{Proof:} First, we have
\begin{equation*}
    \int_0^1 x^{a_1}(1-x)^{a_2} \, dF^{-1}_Y(x)=\int_{\R} F_Y(u)^{a_1}(1-F_Y(u))^{a_2} \, du.
\end{equation*}
By Assumption~\ref{ass:smoothness1}\ref{ass:smoothness12}, for all $u\in \R$:
\begin{equation*}
    |u|\leq C_Y F_Y(u)^{-d_1}(1-F_Y(u))^{-d_2}.
\end{equation*}
Fix $\eps>0$. Then, for all $u\leq -1\wedge F_Y^{-1}(\eps)$, $F_Y(u) \leq C_Y^{1/d_1}(1-\eps)^{-d_2/d_1}|u|^{-1/d_1}$. Thus:
\begin{equation*}
\int_{-\infty}^{-1\wedge F_Y^{-1}(\eps)} F_Y(u)^{a_1}(1-F_Y(u))^{a_2}\, du \leq C_Y^{a_1/d_1}(1-\eps)^{-a_1d_2/d_1} \int_{-\infty}^{-1\wedge F_Y^{-1}(\eps)}|u|^{-a_1/d_1} \, du<\infty,
\end{equation*}
since $d_1<a_1$. A similar reasoning shows that $\int_{1\vee F_Y^{-1}(1-\eps)}^\infty F_Y(u)^{a_1}(1-F_Y(u))^{a_2} \, du<\infty$, using $d_2<a_2$. \hfill$\square$

\begin{lemma}\label{lem:useful_lem}
(Bounds on moments involving $F_U$) Suppose that Assumption \ref{ass:smoothness1} holds and  random variables $Q_n(x)\in\R$ and $B_n(x)\in\{0,1\}$ satisfy, for some $0<\delta<1/2$ and all $0<x< \delta$, $E[B_n(x)|Q_n(x) - x|] \lesssim x$ and $\Pr(Q_n(x)> 1/2,B_n(x)=1)\lesssim x^{1-b_1}$. Then, for such $x\in(0, \delta)$, $E[B_n(x)|F_U(Q_n(x))-F_U(x)|]\lesssim x^{1-b_1}$. The latter inequality holds with $(1-x)^{1-b_2}$ if we replace $x$ by $1-x$, using possibly another $\delta$.
\end{lemma}

\noindent\textbf{Proof:} First, remark that for $x<1/2$, $F_U(x)\lesssim x^{1-b_1}$. Then,
\begin{align*}
& E[B_n(x)|F_U(x)-F_U(Q_n(x))|] \\ & \quad \leq  E[\ind{x>Q_n(x)}B_n(x)|F_U(x)-F_U(Q_n(x))|] + \Pr(Q_n(x)> 1/2,B_n(x)=1) \\[5pt]
&\qquad  + E\left[\ind{Q_n(x) \in [x,1/2]}B_n(x)|F_U(x)-F_U(Q_n(x))|\right]  \\[5pt]
 & \quad \lesssim F_U(x) + x^{1-b_1} + E\left[\ind{Q_n(x) \in [x,1/2]}B_n(x)|F_U(x)-F_U(Q_n(x))|\right] \\[5pt]
& \quad \lesssim x^{1-b_1} + E\left[\ind{Q_n(x) \in [x,1/2]} B_n(x)|F_U(x)-F_U(Q_n(x))|\right].
\end{align*}
Now, if $Q_n(x) \in [x,1/2]$, by the mean value theorem, there exists $X_n \in (x,1/2)$ such that
$$F_U(x)-F_U(Q_n(x))=f_U(X_n) (x-Q_n(x)).$$
Moreover, by Assumption \ref{ass:smoothness1} and $x<\delta$, $f_U(X_n) \lesssim x^{-b_1}$. Then, using $E[B_n(x)|Q_n(x) - x|] \lesssim x$, 
$$E\left[\ind{Q_n(x) \in [x,1/2]}  B_n(x)|F_U(x)-F_U(Q_n(x))|\right] \lesssim x^{1-b_1}.$$
The result follows. \hfill$\square$

\begin{lemma} (Properties of $\Hn^{-1} \circ \Gn$)
There exists $\delta\in (0,1/2)$ and $n_0\in\N$ such that for all  $0<x< \delta$ and all $n\ge n_0$,  
\begin{align}
E\left [\mathds 1\{\xi_{(1)}<  x<\xi_{(n_1)}\}\left|\Hn^{-1}\circ\Gn(x) - x \right| \right] & \lesssim  x.
\label{eq:cond1_qq}
\end{align}
Moreover, for any $\eta>0$, there exists $n_0'$ such that for all $n\geq n_0'$ and for all $0<x< \delta$, 
\begin{equation}
\Pr(\Hn^{-1}\circ\Gn(x)> 1/2, \;  \xi_{(1)}<  x<\xi_{(n_1)})\lesssim  x^{1-\eta}.
\label{eq:cond2_qq}
\end{equation}
Inequalities \eqref{eq:cond1_qq}--\eqref{eq:cond2_qq} hold if we replace $x$ by $1-x$, using possibly another $\delta$ and $n_0$.
\label{lem:HnGn}
\end{lemma}

\noindent\textbf{Proof of Lemma \ref{lem:HnGn}} 
Observe that for a given $x\in [0,1]$, we have $E[\Gn(x)]=x$. Recall that $\mathds 1_{\mathcal A_{n_1}(x)}:=\mathds 1\{\xi_{(1)}<  x<\xi_{(n_1)}\}$ We now establish~\eqref{eq:cond1_qq}.  By the triangle inequality,
\begin{align}
   & E\left [\mathds 1_{\mathcal A_{n_1}(x)}\left|\Hn^{-1}\circ\Gn(x) - x\right|\right] \notag \\[5pt]
   & \leq  E\left [\mathds 1_{\mathcal A_{n_1}(x)}\left|\Hn^{-1}\circ\Gn(x) - \frac{\lceil n_3 \Gn(x)\rceil }{n_3+1} \right| \right] + E\left[\left | \frac{\lceil n_3 \Gn(x)\rceil }{n_3+1} - x\right|\right].\label{eq:mad_ineq}
\end{align}
Consider the second term first. Suppose first that $n_3 x\leq 1$. Let $\lambda=n_3/n_1$ and $B\sim$Binomial$(n_1,x)$. By the triangle inequality, we have
\begin{equation}
E\left[\left | \frac{\lceil n_3 \Gn(x)\rceil }{n_3+1} - x\right|\right] \leq \frac{E\left[\big|\lceil \lambda B\rceil - \lambda n_1 x\big|\right] + x}{n_3+1}.    
    \label{eq:1st_term_x_sm1}
\end{equation}
Since $\lambda n_1 x\leq 1$, $1\leq \lceil \lambda k\rceil \leq \lambda k +1$ for $k\in\{1,2,\ldots,\}$, and $\Pr(B>0)\leq n_1 x$, we get
\begin{align}
    & E[|\lceil \lambda B\rceil - \lambda n_1 x|] \notag \\[5pt]
    & = \lambda n_1 x \Pr(B=0) + \sum_{k=1}^{n_1} \left(\lceil \lambda k\rceil - \lambda n_1 x\right)\Pr(B=k) \notag \\[5pt]
    & \leq \lambda n_1 x(1-\Pr(B>0)) + \sum_{k=1}^{n_1} (\lambda k + 1 - \lambda n_1 x) \Pr(B=k) \notag \\[5pt]
    &= \lambda n_1 x(1-\Pr(B>0)) + \lambda n_1x +\Pr(B>0) -\lambda n_1x\Pr(B>0) \notag \\[5pt]
    & =  n_1 x(2\lambda-(\lambda+1)\Pr(B>0)) +  \Pr(B>0)\notag \\[5pt]
    & \leq (2\lambda +1) n_1 x. \label{eq:1st_term_x_sm2}
\end{align}
Combining \eqref{eq:1st_term_x_sm1}--\eqref{eq:1st_term_x_sm2} and $n_1\lesssim n_3$ yields
$$E\left[\left | \frac{\lceil n_3 \Gn(x)\rceil }{n_3+1} - x\right|\right] \lesssim x.$$
Now, suppose that $n_3 x>1$. Since $\left | \lceil a \rceil - a \right|\leq 1$ for all $a \in \mathbb R_+$, the triangle inequality implies
\begin{align*}
    E\left[\left | \frac{\lceil n_3 \Gn(x)\rceil }{n_3+1}- x\right|\right] & \leq \frac{1}{n_3+1} + \frac{n_3}{n_3+1} E\left[ \left| \Gn(x) - x \right|\right] \notag \\[5pt]
    & \leq \frac{1}{n_3+1} +  E\left[ \left| \Gn(x) - x \right|\right],
\end{align*}
Then, using $n_3+1 > 1/x$ and 
\eqref{eq:cauchy_schwarz}, which holds for all $x\in (0, \widetilde{\delta})$, we obtain 
    \begin{align}
    E\left[\left | \frac{\lceil n_3 \Gn(x)\rceil }{n_3+1}- x\right|\right] &  \lesssim x. \label{eq:ineq_mad_easy}
\end{align}

Next, let us bound the first term of \eqref{eq:mad_ineq}. Since  $\Hn^{-1}(x)\sim $Beta$(\lceil n_3 x \rceil,n_3+1-\lceil n_3 x \rceil)$ for all $x\in(0,1)$, we have 
$$E\left[\mathds 1_{\mathcal A_{n_1}(x)}\Hn^{-1} \circ \Gn(x) \,|(\xi_i)_i \right] = \mathds 1_{\mathcal A_{n_1}(x)}\frac{\lceil n_3 \Gn(x) \rceil}{n_3+1}.$$ 
Moreover, any $Z\sim$Beta$(a,b)$ satisfies $E[|Z-E(Z)|]=2a^ab^b/(B(a,b)(a+b)^{a+b+1})$. Thus,
\begin{align*}
& E\left[1_{\mathcal A_N(x)}\left|\Hn^{-1} \circ \Gn(x) - \frac{\lceil n_3 \Gn(x) \rceil}{n_3+1} \right| \bigg| (\xi_i)_i\right] \notag \\[5pt]
= & 1_{\mathcal A_N(x)}\frac{2\lceil n_3 \Gn(x) \rceil^{\lceil n_3 \Gn(x) \rceil} (n_3+1- \lceil n_3 \Gn(x) \rceil)^{n_3+1- \lceil n_3 \Gn(x) \rceil}}{{\rm B}(\lceil n_3 \Gn(x) \rceil,n_3+1- \lceil n_3 \Gn(x) \rceil)(n_3+1)^{n_3+2}}.
\end{align*}
Let $f(u)=2u^u (n_3+1-u)^{n_3+1-u} /[{\rm B}(u,n_3+1-u) (n_3+1)^{n_3+2}]$ for $u\in [0,n_3]$. It follows that
$$E\left[\mathds 1_{\mathcal A_{n_1}(x)}\left|\Hn^{-1} \circ \Gn(x) - \frac{\lceil n_3 \Gn(x) \rceil}{n_3+1} \right| \right] = E\left[\mathds 1_{\mathcal A_{n_1}(x)}f\left(\lceil n_3 \Gn(x) \rceil\right) \right].$$
Now, Stirling's formula gives the following bound for the beta function ${\rm B}(\cdot, \cdot)$ \cite[see, e.g., p.~263, Ex.~45,][]{whittaker_watson_1996}:
\begin{equation} \label{eq:bound_beta}
\frac{1}{{\rm B}(x,y)} < \frac{1}{\sqrt{2 \pi}} \frac{(x+y)^{x+y-1/2}}{x^{x-1/2}y^{y-1/2}}, \quad \forall x,y >0.
\end{equation}
Plugging \eqref{eq:bound_beta} for $x=u$ and $y=n_3 + 1 - u$ in the definition of $f(u)$, we have for all $1 \leq u \leq n_3$
\begin{align*}
    f(u) & \leq \frac{1}{\sqrt{2\pi}}\frac{2 u^u (n_3 + 1 - u)^{n_3 + 1 - u}(n_3 + 1)^{n_3 + 1/2}}{u^{u-1/2}(n_3 + 1 - u)^{n_3 + 1/2 - u}(n_3 + 1)^{n_3+2}} \\[5pt]
    & \lesssim \frac{u^{1/2}(n_3 + 1 - u)^{1/2}}{(n_3 + 1)^{3/2}} \\[5pt]
    & \leq \frac{u}{n_3},
\end{align*}
where the last inequality uses $u^{1/2} \leq u$ and $(n_3 + 1 - u) \leq n_3 + 1$ for all $1 \leq u \leq n_3$. Hence, 
\begin{align*}
E\left[\mathds 1_{\mathcal A_{n_1}(x)}\left|\Hn^{-1} \circ \Gn(x) - \frac{\lceil n_3 \Gn(x) \rceil}{n_3+1} \right| \right] 
& \lesssim \frac{1}{n_3} E\left[ \mathds 1_{\mathcal A_{n_1}(x)}\lceil n_3 \Gn(x) \rceil \right] \\[5pt]
& \leq \frac{1}{n_3} \left( n_3 E\left[\Gn(x)\right]+E\left[\mathds 1_{\mathcal A_{n_1}(x)}\right] \right) \\
& = \frac{1}{n_3} (n_3x+\Pr(\Gn(x) >0)) \\[5pt]
& \lesssim x.
\end{align*}
where we have used $\lceil a\rceil \leq a+1$ for all $a\in\{0,1,2,...\}$, $\mathds 1_{\mathcal A_{n_1}(x)}\leq \mathds 1\{\Gn(x) >0\}\leq 1$, $E\left[\Gn(x)\right]=x$, and  $\Pr(\Gn(x) >0)\leq n_1 x$.

\medskip
We now turn to Equation \eqref{eq:cond2_qq}.  Since $\Hn^{-1}(t)>u$ implies $t\geq \Hn(u)$ for all $(t,u)\in(0,1)$, we have
\begin{align*}
\Pr(\Hn^{-1}\circ \Gn(x)>1/2, \; \xi_{(1)}<  x<\xi_{(n_1)}) 
& \leq E\left[\Pr(\Gn(x) \geq \Hn(1/2)|\Hn(1/2))\right] \\[5pt]
& \leq E\left[\left(xe\right)^{n_1(\Hn(1/2) -x)^2}\right] \\[5pt]
& \leq xe + \Pr\left(\Hn(1/2)-x < 1/\sqrt{n_1}\right) \\[5pt]
& \leq xe + \exp\left(-2(\sqrt{n_3}(x-1/2)+\sqrt{n_3/n_1})^2\right) \\[5pt]
& = xe + \exp\left(-2n_3(x-1/2+1/\sqrt{n_1})^2\right),
\end{align*}
where the second and fourth inequalities follow from Kiefer's and Hoeffding's inequalities, respectively. Let $\bar{\delta}\in (0,e^{-1}]$ and fix $\delta=\bar{\delta}/2$ and $n_0\geq (2/\bar{\delta})^2$. Then, for all $n_1\geq n_0$ and 
any $0<x\leq \delta$, we have
\begin{align*}
    \left|x - 1/2 + \frac{1}{\sqrt{n_1}}\right| & = \frac{1}{2} - (x+ 1/\sqrt{n_1}) \\[5pt]
    & \geq \frac{1}{2} - \bar{\delta}.
\end{align*}
Note that there exists a constant $\underbar c>0$ such that for $N$ sufficiently large,  $\underbar c \leq n_3/n_1$. Let $C=2\underbar c(1/2-\bar{\delta})^2$ and suppose first that $x\geq \exp(A-Cn_1)$ for some $A$. Then some algebra shows that, for $N$ sufficiently large,
$$\Pr(\Hn^{-1}\circ \Gn(x)>1/2, \; \xi_{(1)}<x<\xi_{(n_1)}) \lesssim x.$$
Next, assume that $x< \exp(A-Cn_1)$. Then,
\begin{align*}
\Pr(\Hn^{-1}\circ \Gn(x)>1/2, \; \xi_{(1)}<x<\xi_{(n_1)}) \leq & \Pr(\Gn(x)\geq  1/n_1) \\[5pt]
=& 1 - (1-x)^{n_1} \\[5pt]
\leq & n_1 x \\[5pt]
\leq & \frac{A - \ln x}{C} x.
\end{align*}
For any $\eta>0$, we have $-\ln x \lesssim x^{-\eta}$. Thus, $\Pr(\Hn^{-1}\circ \Gn(x)>1/2, \; \xi_{(1)}<x<\xi_{(n_1)}) \lesssim x^{1-\eta}$.
\hfill $\square$

\begin{lemma}\label{lem:int_FY_FY'}
    Let Assumption~\ref{ass:smoothness1}\ref{ass:smoothness12} hold and let $a_1,a_2 \in(0,\infty)$ be such that $2 > a_1 > 2d_1$ and $2 > a_2 > 2d_2$. 
    Then there exists a constant $C>0$ depending only on $a_1,a_2, d_1,d_2$ such that
    \begin{align*}
    &\int_{\R^2}E\left(\left[\widehat F_Y^{(1)}(y) \land \widehat F_Y^{(2)}(y')\right]^{a_1}\left[\bar{\widehat F}_Y^{(1)}(y) \land \bar{\widehat F}^{(2)}_Y(y')\right]^{a_2}\right) dydy' \leq C,\\
    &E\left[\left(\int_{\R} \left(\left[\widehat F_Y^{(j)}(y) \right]^{a_1/2}\left[\bar{\widehat F}_Y^{(j)}(y) \right]^{a_2/2}\right) dy \right)^2\right] \leq C, \quad j\in\{1,2\}.
    \end{align*}
\end{lemma}
\textbf{Proof:} Using the inequality $(x\land y)^{a_1} (\bar x \land \bar y)^{a_2} \leq x^{a_1/2} \, \bar x^{a_2/2} \, y^{a_1/2} \, \bar y^{a_2/2}$ for any $x,y \in [0,1]$, we obtain 
\begin{align*}
    &\hspace{6mm}\int_{\R^2}E\left(\left[\widehat F_Y^{(1)}(y) \land \widehat F_Y^{(2)}(y')\right]^{a_1}\left[\bar{\widehat F}_Y(y) \land \bar{\widehat F}_Y^{(2)}(y')\right]^{a_2}\right) dydy' \\[5pt]
    & \leq \left\{E\int_{\R} \left(\left[\widehat F_Y^{(1)}(y) \right]^{a_1/2}\left[\bar{\widehat F}_Y^{(1)}(y) \right]^{a_2/2}\right) dy  \right\}^2 \\[5pt]
    &\leq  
     E\left[\left(\int_{\R} \left(\left[\widehat F_Y^{(1)}(y) \right]^{a_1/2}\left[\bar{\widehat F}_Y^{(1)}(y) \right]^{a_2/2}\right) dy \right)^2\right]\\[5pt]
     & = \int_{\R^2} E\left(\left[\widehat F_Y^{(1)}(y) \right]^{a_1/2}\left[\bar{\widehat F}_Y^{(1)}(y) \right]^{a_2/2}\left[\widehat F_Y^{(2)}(y') \right]^{a_1/2}\left[\bar{\widehat F}_Y^{(2)}(y') \right]^{a_2/2} \right) dydy'\\[5pt]
     & \leq \int_{\R^2} \left[E\left(\widehat F_Y^{(1)}(y)\right)\right]^{a_1/2}\left[E\left(\bar{\widehat F}^{(1)}_Y(y)\right)\right]^{a_2/2} \left[E\left(\widehat F_Y^{(2)}(y')\right)\right]^{a_1/2}\left[E\left(\bar{\widehat F}_Y^{(2)}(y')\right)\right]^{a_2/2} dy dy'\\[5pt]
     & = \int_{\R^2} \Big[F_Y(y)\Big]^{a_1/2}\Big[\bar{ F}_Y(y)\Big]^{a_2/2} \Big[F_Y(y')\Big]^{a_1/2}\Big[\bar{ F}_Y(y')\Big]^{a_2/2} dy dy'.
    \end{align*}
    In the last inequality, we used Jensen's inequality since the function $(x,y) \to x^{a_1/2} \, (1-x)^{a_2/2} \, y^{a_1/2} \, (1-y)^{a_2/2}$ is concave over $[0,1]^2$. 
    The last integral is finite by Lemma~\ref{lem:int_FY} since $a_1/2>d_1$ and $a_2/2>d_2$ by assumption.
\hfill $\square$

\begin{lemma}\label{lem:useful_ineq}
    Let $\eps>0$ and $s,s',t,t'\in\R $ such that $ \abs{s-s'}\leq \eps (s\wedge s')$ and $ \abs{t-t'}\leq \eps(t\wedge t')$. Then, $ \abs{s\wedge t-s'\wedge t'}\leq \eps(s\wedge t)$.
\end{lemma}
\noindent\textbf{Proof:} 
Without loss of generality, suppose that $s\leq t$. \\
{\em First case:} Suppose $s'\leq t'$. Then 
\[
\abs{s\wedge t-s'\wedge t'}=\abs{s-s'}\leq \eps(s\wedge s')\leq \eps s=\eps(s\wedge t).
\]
{\em Second case:} Suppose $s'>t'$. Then 
\[
\abs{s\wedge t-s'\wedge t'}=\abs{s-t'}.
\]
{\em First subcase:} Suppose $s<t'$ and $s'<t$. Then, $s<t'<s'<t$ and
\[
\abs{s-t'}\leq \abs{s-s'}\leq \eps(s\wedge s')=\eps s=\eps(s\wedge t).
\]
{\em Second subcase:} Suppose $s\geq t'$ or $s'\geq t$. If $s\geq t'$, then
\[
\abs{s-t'}\leq \abs{t-t'}\leq \eps(t\wedge t')=\eps t'\leq \eps s=\eps(s\wedge t).
\]
Else, $s<t'$ and $s'\geq t$ imply
\[
\abs{s-t'}\leq\abs{s-s'}\leq\eps(s\wedge s')=\eps s=\eps(s\wedge t).
\]
The result follows. \hfill$\square$

\begin{lemma}\label{lem_Xi_Gn+}
For any $t \in [0,1]$, let 
\begin{align*}
    \Xi(t) &= \sup\left\{y \in [0,1]: \Hninv \circ \Gn (y) \leq t\right\},\\
    \mathbb G_{+,n_1}^{-1}(t) &= \sup \left\{y \in [0,1]: \Gn(y) \leq t\right\}.
\end{align*}
    Then, it holds that $\Xi = \mathbb G_{+,n_1}^{-1} \circ \Hn$. 
\end{lemma}

\noindent\textbf{Proof.}
    Let $t \in [0,1]$. We recall that, for any $y \in [0,1]$, $\Hninv \circ \Gn(y) = \inf \left\{z \in [0,1]: \Hn(z) \geq \Gn(y)\right\}$. We have
    \begin{align*}
        & \hspace{1cm}\, \Big[\Xi(t) = \mathbb G_{+,n_1}^{-1} \circ \Hn(t)\Big] \\[2pt]
        &\Longleftrightarrow ~~ \sup\left\{y\in [0,1]: \Hninv \circ \Gn (y) \leq t\right\} = \sup \left\{y \in [0,1]: \Gn(y) \leq \Hn(t)\right\}\\[5pt]
        & \Longleftarrow ~~\,\forall y \in [0,1]: \left[ \Hninv \circ \Gn (y) \leq t \Longleftrightarrow \Gn(y) \leq \Hn(t)\right]\\[5pt]
        & \Longleftrightarrow ~~\forall y \in [0,1]: \Big[ \inf\left\{z \in [0,1]: \Hn(z) \geq \Gn(y)\right\}\leq t \Longleftrightarrow \Gn(y) \leq \Hn(t)\Big].
    \end{align*}
    Let $y \in [0,1]$, and assume first that $\Gn(y) \leq \Hn(t)$. Then, we have $$\inf\left\{z \in [0,1]: \Hn(z) \geq \Gn(y)\right\}\leq t.$$ 

    Conversely, let $\ell = \inf\left\{z \in [0,1]: \Hn(z) \geq \Gn(y)\right\}$, and assume that $\ell \leq t$. 
    We can define a non-increasing sequence $(z_k)_k$ of elements of $\left\{z \in [0,1]: \Hn(z) \geq \Gn(y)\right\}$ that converges to $\ell$. For any $k \in \N$, by definition of $z_k$, it holds that
    \begin{align*}
        \Gn(y) \leq \Hn(z_k) &\to \Hn(\ell) \qquad \text{ since $\Hn$ is right-continuous}\\
        &\leq \Hn(t) \qquad \text{ since $\Hn$ is non-decreasing.} 
    \end{align*}
This concludes the proof.
\hfill$\square$

\begin{lemma}\label{lem_Gn+_equiv_Gn}
    Let $\mathcal I_{0}:=[0,\Hn^{-1}(n_1^{-1})]$ and  $\mathcal I_{1}:=[\zeta_{(n_3)}, 1]$. It holds that 
    \begin{align}
        \int_{(\mathcal I_0 \cup \mathcal{I}_1)^c} \Big(F_Y^{-1} \circ \mathbb G_{+,n_1}^{-1}  \circ \Hn - F_Y^{-1} \circ \Gn^{-1} \circ \Hn\Big) d\widehat F_U = o_{P}\left(\frac{1}{\sqrt{N}}\right) \label{eq:diff_1}
    \end{align}
    and
        \begin{align}
        \int_{(\mathcal I_0 \cup \mathcal{I}_1)^c} \Big(F_Y^{-1} \circ \mathbb G_{+,n_1}^{-1}  \circ \Hn - F_Y^{-1} \circ \left(\Hn^{-1} \circ \Gn\right)^{-1}\Big) d\widehat F_U = o_{P}\left(\frac{1}{\sqrt{N}}\right). \label{eq:diff_2}
    \end{align}
\end{lemma}

\noindent\textbf{Proof.}
We note that $\mathbb G_{+,n_1}^{-1} \geq \Gn^{-1}$, so that the integrand 
\begin{align*}
    \Big(F_Y^{-1} \circ \mathbb G_{+,n_1}^{-1}  \circ \Hn - F_Y^{-1} \circ \Gn^{-1} \circ \Hn\Big) 
\end{align*}
is always non-negative. 
To prove the result, it therefore suffices to show that
\begin{align*}
        E\left[\int_{(\mathcal I_0 \cup \mathcal{I}_1)^c} \Big(F_Y^{-1} \circ \mathbb G_{+,n_1}^{-1}  \circ \Hn - F_Y^{-1} \circ \Gn^{-1} \circ \Hn\Big) d\widehat F_U\right] = o\left(\frac{1}{\sqrt{N}}\right).
    \end{align*}
Let $t \in {(\mathcal I_0 \cup \mathcal{I}_1)^c}$. Then
\begin{align*}
    \mathbb G_{n_1}^{-1}(y) \leq t \quad \Longleftrightarrow \quad  \mathbb G_{n_1}^{-1}(y) \leq  \frac{\lfloor tn_1 \rfloor}{n_1} \quad \Longleftrightarrow \quad y < Y_{\lfloor tn_1 \rfloor+1}.
\end{align*}
Hence for any $t \in {(\mathcal I_0 \cup \mathcal{I}_1)^c}$, we have
\begin{align*}
   \mathbb G_{+,n_1}^{-1}(t) = \sup \left\{y \in \R : \mathbb G_{n_1}(y) \leq t\right\} = Y_{\lfloor tn_1 \rfloor+1}.
\end{align*}
Therefore, we have
\begin{align*}
    & \hspace{.5cm}E\left[\int_{(\mathcal I_0 \cup \mathcal{I}_1)^c} \Big(F_Y^{-1} \circ \mathbb G_{+,n_1}^{-1}  \circ \Hn - F_Y^{-1} \circ \Gn^{-1} \circ \Hn\Big) d\widehat F_U\right]\\
    &\leq E\left[\int_{(\mathcal I_0 \cup \mathcal{I}_1)^c} \left(Y_{(\lceil \Hn(u) n_1 + 1\rceil)} - Y_{(\lceil \Hn(u) n_1 \rceil)}\right)\ind{\Hn(u) \in \{k/n_1: k = 1,\dots, n_1\}} d\widehat F_U(u)\right] \\
    & = E\left[\frac{1}{n_2} \sum_{i=1}^{n_2} \left(Y_{(\lceil \Hn(U_i) n_1 + 1\rceil)} - Y_{(\lceil \Hn(U_i) n_1 \rceil)}\right) \ind{\Hn(U_i) \in \{k/n_1: k = 1,\dots, n_1-1\}}\right] \\
    & = E\left[\frac{1}{n_2} \sum_{i=1}^{n_2} \sum_{k=1}^{n_1-1} \left(Y_{(\lceil \Hn(U_i) n_1 + 1\rceil)} - Y_{(\lceil \Hn(U_i) n_1 \rceil)}\right) \ind{\Hn(U_i) = k/n_1}\right]\\
    & = E\left[\frac{1}{n_2}  \sum_{k=1}^{n_1-1} \sum_{i=1}^{n_2} \left(Y_{(k + 1)} - Y_{(k)}\right) \ind{\Hn(U_i) = k/n_1}\right]\\
    & = \sum_{k=1}^{n_1-1} E\left(Y_{(k + 1)} - Y_{(k)}\right) \mathbb \Pr\left(\Hn(U_i) = k/n_1\right).
\end{align*}
Moreover, there exists an integer $k \in \{1,\dots, n_1-1\}$ such that $\Hn(U_i) = \frac{k}{n_1}$ if, and only if,
\begin{align*}
    & \qquad \frac{1}{n_3} \sum_{j=1}^{n_3} \ind{F_Z(Z_j) \leq U_i} = \frac{k}{n_1}\\
    & \text{i.e. }\quad  \sum_{j=1}^{n_3} \ind{\zeta_j \leq U_i} = \frac{kn_3}{n_1} \qquad \text{ where } \quad \zeta_j = F_Z(Z_j)\\
    & \text{i.e. }\quad \text{Card}\Big\{j \in \{1, \dots, n_3\}: \zeta_j \leq U_i\Big\} = \frac{kn_3}{n_1}\\
    & \text{i.e. }\quad U_i \in \left[\zeta_{(\lceil \frac{kn_3}{n_1}\rceil)}, \zeta_{(\lfloor \frac{kn_3}{n_1} + 1\rfloor)}\right).
\end{align*}
Note that, if $\frac{kn_3}{n_1}$ is not an integer, this condition is never satisfied, since the interval $\left[\zeta_{(\lceil \frac{kn_3}{n_1}\rceil)}, \zeta_{(\lfloor \frac{kn_3}{n_1} + 1\rfloor)}\right)$ is empty. 

We thus have
\begin{align*}
    & \hspace{.5cm}E\left[\int_{0}^1 \Big(F_Y^{-1} \circ \mathbb G_{+,n_1}^{-1}  \circ \Hn - F_Y^{-1} \circ \Gn^{-1} \circ \Hn\Big) d\widehat F_U\right]\\
    & \leq \sum_{k=1}^{n_1-1} E\left(Y_{(k + 1)} - Y_{(k)}\right) \mathbb \Pr\left(U_i \in \left[\zeta_{(\lceil \frac{kn_3}{n_1}\rceil)}, \zeta_{(\lfloor \frac{kn_3}{n_1} + 1\rfloor)}\right)\right).
\end{align*}
Moreover,
\begin{align*}
    \mathbb \Pr\left(U_i \in \left[\zeta_{(\lceil \frac{kn_3}{n_1}\rceil)}, \zeta_{(\lfloor \frac{kn_3}{n_1} + 1\rfloor)}\right)\right) &= E\left[E\Bigg[\int_{\zeta_{(\lceil \frac{kn_3}{n_1}\rceil)}}^{\zeta_{(\lfloor \frac{kn_3}{n_1} + 1\rfloor)}} f_U(u)du\, \bigg| \, \zeta_1, \dots, \zeta_{n_3}\Bigg]\right]\\
    & \leq E\left[E\Bigg[\int_{\zeta_{(\lceil \frac{kn_3}{n_1}\rceil)}}^{\zeta_{(\lfloor \frac{kn_3}{n_1} + 1\rfloor)}} C_U u^{-b_1} (1-u)^{-b_2} du\, \bigg| \, \zeta_1, \dots, \zeta_{n_3}\Bigg]\right]\\
    & \leq C_UE\left[\left(\zeta_{(\lfloor \frac{kn_3}{n_1} + 1\rfloor)} - \zeta_{(\lceil \frac{kn_3}{n_1}\rceil)}\right)\left(\zeta_{(\lceil \frac{kn_3}{n_1}\rceil)}\right)^{-b_1} \left(1-\zeta_{(\lfloor \frac{kn_3}{n_1} + 1\rfloor)}\right)^{-b_2} \right]\\
    & \lesssim \frac{1}{n_3} \left(\frac{n_3}{\lceil \frac{kn_3}{n_1}\rceil}\right)^{b_1} \left(\frac{n_3}{n_3-\lceil \frac{kn_3}{n_1}\rceil}\right)^{b_2} \quad \text{ by Lemma~\ref{lem_control_moments_unif}}\\
    & \asymp \frac{1}{n_1} \left(\frac{n_1}{k}\right)^{b_1} \left(\frac{n_1}{n_1-k}\right)^{b_2}\\
    & \asymp \frac{1}{n_1} \left[\left(\frac{n_1}{k}\right)^{b_1} + \left(\frac{n_1}{n_1-k}\right)^{b_2}\right].
\end{align*}
We can now conclude the proof of~\eqref{eq:diff_1} as follows.
\begin{align*}
    & \hspace{.5cm}E\left[\int_{(\mathcal I_0 \cup \mathcal{I}_1)^c} \Big(F_Y^{-1} \circ \mathbb G_{+,n_1}^{-1}  \circ \Hn - F_Y^{-1} \circ \Gn^{-1} \circ \Hn\Big) d\widehat F_U\right]\\
    & \leq \sum_{k=1}^{n_1-1} E\left(Y_{(k + 1)} - Y_{(k)}\right) \mathbb \Pr\left(U_i \in \left[\zeta_{(\lceil \frac{kn_3}{n_1}\rceil)}, \zeta_{(\lfloor \frac{kn_3}{n_1} + 1\rfloor)}\right)\right)\\
    & \lesssim \sum_{k=1}^{n_1-1} E\left(Y_{(k + 1)} - Y_{(k)}\right) \frac{1}{n_1} \left[\left(\frac{n_1}{k}\right)^{b_1} + \left(\frac{n_1}{n_1-k}\right)^{b_2}\right]\\
    & \lesssim \sum_{k=1}^{\lfloor n_1/2\rfloor} E\left(Y_{(k + 1)} - Y_{(k)}\right) \frac{1}{n_1} \left(\frac{n_1}{k}\right)^{b_1} + \sum_{k=\lfloor n_1/2\rfloor +1}^{n_1-1} E\left(Y_{(k + 1)} - Y_{(k)}\right) \frac{1}{n_1} \left(\frac{n_1}{n_1-k}\right)^{b_2}\\
    & \leq n^{b_1 - 1} E\left[Y_{(\lfloor n_1/2\rfloor)} - Y_{(1)}\right] + n^{b_2 - 1} E\left[Y_{(n_1)} - Y_{(\lfloor n_1/2\rfloor + 1)} \right]\\
    & \lesssim n^{-1+b_1 + d_1} + n^{-1 + b_2 + d_2} \quad \text{ by Lemma~\ref{lem:extrem_moments1}}\\
    & = o(1).
\end{align*}

As for \eqref{eq:diff_2}, we note that $\mathbb G_{+,n_1}^{-1}\circ \Hn \geq \left(\Hn^{-1}\circ\Gn\right)^{-1}$ since $\mathbb G_{+,n_1}^{-1}\circ \Hn$ is the right inverse of $\Hn^{-1}\circ\Gn$ by Lemma~\ref{lem_Xi_Gn+}, so that the integrand 
\begin{align*}
    \Big(F_Y^{-1} \circ \mathbb G_{+,n_1}^{-1}  \circ \Hn - F_Y^{-1} \circ \left(\Hn^{-1}\circ \Gn\right)^{-1}\Big) 
\end{align*}
is always non-negative. 
To prove the result, it therefore suffices to show that
\begin{align*}
        E\left[\int_{(\mathcal I_0 \cup \mathcal{I}_1)^c}\Big(F_Y^{-1} \circ \mathbb G_{+,n_1}^{-1}  \circ \Hn - F_Y^{-1} \circ \left(\Hn^{-1}\circ \Gn\right)^{-1}\Big) d\widehat F_U\right] = o\left(\frac{1}{\sqrt{N}}\right).
    \end{align*}
Since $\mathbb G_{+,n_1}^{-1}\circ \Hn$ is the right inverse of $\Hn^{-1}\circ\Gn$, we have
\begin{align*}
    & \hspace{.5cm}E\left[\int_{(\mathcal I_0 \cup \mathcal{I}_1)^c} \Big(F_Y^{-1} \circ \mathbb G_{+,n_1}^{-1}  \circ \Hn - F_Y^{-1} \circ \left(\Hn^{-1}\circ \Gn\right)^{-1}\Big) d\widehat F_U\right]\\
    &\leq  E\left[\frac{1}{n_2} \sum_{i=1}^{n_2} \sum_{k=1}^{n_3-1} \left(Y_{(\lceil \frac{k+1}{n_3} n_1\rceil)} - Y_{(\lceil \frac{k}{n_3}n_1\rceil)}\right) \ind{\Hn(U_i) = k/n_3}\right]\\
    & = \sum_{k=1}^{n_3-1} E\left(Y_{(\lceil \frac{k+1}{n_3} n_1\rceil)} - Y_{(\lceil \frac{k}{n_3}n_1\rceil)}\right) \mathbb \Pr\left(\Hn(U_i) = k/n_3\right).
\end{align*}
The result follows from arguments similar to those above.

\hfill$\square$

\begin{lemma}\label{lem_control_moments_unif}
    Let \(1\le k<n\), and let \(a,b>0\) with \(a<k\) and \(b<n-k\) and $U_1,\dots, U_n \overset{iid}{\sim} \text{Uniform}(0,1)$. Let $U_{(1)} \leq \dots \leq U_{(n)}$ denote the corresponding order statistics. Then it holds that
\[
\mathbb E\!\left[
(U_{(k+1)}-U_{(k)})U_{(k)}^{-a}(1-U_{(k+1)})^{-b}
\right]
\lesssim
\frac1n
\left(\frac nk\right)^a
\left(\frac n{n-k}\right)^b ,
\]
where the implicit constant depends only on \(a\) and \(b\).

\end{lemma}

\textbf{Proof.}

The joint density of \((U_{(k)},U_{(k+1)})\) is
\[
f_{U_{(k)},U_{(k+1)}}(x,y)
=
\frac{n!}{(k-1)!(n-k-1)!}
x^{k-1}(1-y)^{n-k-1},
\qquad 0<x<y<1.
\]
Therefore,
\begin{align*}
&\mathbb E\!\left[
(U_{(k+1)}-U_{(k)})U_{(k)}^{-a}(1-U_{(k+1)})^{-b}
\right] \\
&=
\frac{n!}{(k-1)!(n-k-1)!}
\int_0^1\int_0^y
(y-x)x^{k-1-a}(1-y)^{n-k-1-b}
\,dx\,dy.
\end{align*}
For \(a<k\),
\[
\int_0^y (y-x)x^{k-1-a}\,dx
=
\frac{y^{k+1-a}}{(k-a)(k+1-a)}.
\]
Hence
\begin{align*}
&\mathbb E\!\left[
(U_{(k+1)}-U_{(k)})U_{(k)}^{-a}(1-U_{(k+1)})^{-b}
\right] \\
&=
\frac{n!}{(k-1)!(n-k-1)!(k-a)(k+1-a)}
\int_0^1
y^{k+1-a}(1-y)^{n-k-1-b}
\,dy \\
&=
\frac{n!}{(k-1)!(n-k-1)!(k-a)(k+1-a)}
B(k+2-a,n-k-b).
\end{align*}
Using \(B(p,q)=\Gamma(p)\Gamma(q)/\Gamma(p+q)\), we get
\[
\mathbb E\!\left[
(U_{(k+1)}-U_{(k)})U_{(k)}^{-a}(1-U_{(k+1)})^{-b}
\right]
=
\frac{n!\Gamma(k+2-a)\Gamma(n-k-b)}
{(k-1)!(n-k-1)!(k-a)(k+1-a)\Gamma(n+2-a-b)}.
\]
Since
\[
\Gamma(k+2-a)=(k+1-a)(k-a)\Gamma(k-a),
\]
this becomes
\[
\mathbb E\!\left[
(U_{(k+1)}-U_{(k)})U_{(k)}^{-a}(1-U_{(k+1)})^{-b}
\right]
=
\frac{\Gamma(n+1)}{\Gamma(n+2-a-b)}
\frac{\Gamma(k-a)}{\Gamma(k)}
\frac{\Gamma(n-k-b)}{\Gamma(n-k)}.
\]

We now use the standard Gamma-ratio estimate: for each fixed \(c>0\), there exists a constant
\(C_c>0\) such that, for all \(m>c\),
\[
\frac{\Gamma(m-c)}{\Gamma(m)}
\le C_c m^{-c}.
\]
Similarly, for fixed \(a,b>0\), there exists \(C_{a,b}>0\) such that
\[
\frac{\Gamma(n+1)}{\Gamma(n+2-a-b)}
\le C_{a,b} n^{a+b-1}.
\]
Applying these bounds gives
\begin{align*}
\mathbb E\!\left[
(U_{(k+1)}-U_{(k)})U_{(k)}^{-a}(1-U_{(k+1)})^{-b}
\right]
&\le
C_{a,b}
n^{a+b-1} k^{-a}(n-k)^{-b} \\
&=
C_{a,b}
\frac1n
\left(\frac nk\right)^a
\left(\frac n{n-k}\right)^b .
\end{align*}
This proves the desired bound.

\hfill$\square$

\end{document}


\renewcommand{\thesection}{S.\arabic{section}}
\renewcommand{\thelemma}{\thesection.\arabic{lemma}}
\numberwithin{equation}{section}
\renewcommand{\theequation}{\thesection.\arabic{equation}}

\title{Supplementary Appendix for ``Asymptotic Properties of Empirical Quantile-Based Estimators''}

\author{Julien Chhor\thanks{Toulouse School of Economics, University of Toulouse Capitole, Toulouse, France, julien.chhor@tse-fr.eu.} \and Xavier D'Haultf\oe{}uille\thanks{CREST-ENSAE, xavier.dhaultfoeuille@ensae.fr.} \and J\'er\'emy L'Hour\thanks{CFM \& CREST-ENSAE, jeremy.l.hour@ensae.fr.} \and Martin Mugnier\thanks{Paris School of Economics, martin.mugnier@psemail.eu.}}

\maketitle

\begin{abstract}
This appendix contains additional lemmas used to prove Theorems 1 and 2. 
\end{abstract}

\section{Lemmas for Theorem 1}\label{sec:S1}

\begin{lemma}\label{lem:11}
   For all $j\in\{1,2,3,4\}$: $R_j=o_P(1)$ as $N\to\infty$.
\end{lemma}
{\bf Proof:}
First, we show that $R_1=o_P(1)$.  By the triangle inequality, we have
\begin{equation}\label{eq:T1n1_rem11}
\begin{aligned}
 \abs{R_1} & \leq  \sqrt{N}\abs{F_Y^{-1}\big(\xi_{(n_1)}\big)}\abs{F_U(1)-F_U\big(\xi_{(n_1)}\big)} 
  +\sqrt{N}\abs{F_Y^{-1}\big(\xi_{(1)}\big)}\abs{F_U\big(1/n_1\big)-F_U\big(\xi_{(1)}\big)} \\[5pt]
 &=\sqrt{N}\abs{Y_{(n_1)}}\abs{1-F_U(\xi_{(n_1)})} +\sqrt{N}\abs{Y_{(1)}}\abs{F_U\big(1/n_1\big)-F_U\big(\xi_{(1)}\big)}.
\end{aligned}
\end{equation}
Lemma~\ref{lem:extrem_moments1} and Markov's inequality imply
\[
    \sqrt{N}\abs{Y_{(1)}} =O_P\left(n_1^{\frac12 +d_1}\right) \quad \text{ and } \quad  \sqrt{N}\abs{Y_{(n_1)}} =O_P\left(n_1^{\frac12 +d_2}\right),
\]
and it suffices to show that 
\[
 \quad \abs{F_U\big(1/n_1\big)-F_U\big(\xi_{(1)}\big)} = o_P\left(n_1^{-\frac12-d_1}\right) \quad \text{ and } \quad  \abs{1-F_U(\xi_{(n_1)})}=o_P\left(n_1^{-\frac12-d_2}\right).
\]
We show the second result only (the first can de obtained analogously). By the triangle inequality and the mean value theorem, there exists $c_{n_1}\in((n_1-1)/n_1\land \eps_{(n_1)}, (n_1-1)/n_1\lor \eps_{(n_1)})$ such that
\begin{align*}
     \abs{1-F_U\big(\xi_{(n_1)}\big)}&\leq \abs{1- F_U\big((n_1-1)/n_1\big)} +\abs{F_U\big((n_1-1)/n_1\big) -F_U\big(\xi_{(n_1)}\big)} \\[5pt]
     &=\int_{\frac{n_1-1}{n_1}}^{1}f_U(u)du + f_U(c_{n_1})\left(1-\xi_{(n_1)} +\frac{1}{n_1}\right) \\[5pt]
     &\lesssim \int_{\frac{n_1-1}{n_1}}^{1}u^{-b_1}(1-u)^{-b_2}du + c_{n_1}^{-b_1}(1-c_{n_1})^{-b_2}\left(1-\xi_{(n_1)} +\frac{1}{n_1}\right)\\[5pt]
     &\lesssim \int_0^{\frac{1}{n_1}}v^{-b_2}dv +  \left(2^{b_1}+\xi_{(n_1)}^{-b_1}\right)\left(n_1^{b_2}+(1-\xi_{(n_1)})^{-b_2}\right)\left(1-\xi_{(n_1)} +\frac{1}{n_1}\right)\\[5pt]
     &\lesssim n_1^{b_2-1}+n_1^{b_2-1}\left(2^{b_1}+\xi_{(n_1)}^{-b_1}\right)\left(1+[n_1(1-\xi_{(n_1)})]^{-b_2}\right)\left(n_1\big(1-\xi_{(n_1)}\big) +1\right) \\
     &=n_1^{b_2-1}\left(1+O_P(1)\right),
\end{align*}
where the last line follows from $\xi_{(n_1)}\overset{p}{\to}1$, $n_1(1-\xi_{(n_1)})\overset{d}{\to}$Exponential(1), and the continuous mapping theorem. The result follows from $b_2+d_2<1/2$ by Assumption~\ref{ass:smoothness1}\ref{ass:smoothness14}.

\medskip
Second, we show that $R_2=o_P(1)$. By Assumption~\ref{ass:smoothness1}\ref{ass:smoothness12}--\ref{ass:smoothness13},
\begin{align*}
\sqrt{N}\left|\int_0^{\xi_{(1)}} F_Y^{-1}dF_U\right| &\lesssim \sqrt{n_1}\ind{\xi_{(1)}\geq 1/2} \left|\int_0^1 F_Y^{-1}dF_U\right| + \sqrt{n_1}\ind{\xi_{(1)}< 1/2} \int_0^{\xi_{(1)}} t^{-b_1-d_1}dt \\[5pt]
&\lesssim \sqrt{n_1}\ind{\xi_{(1)}\geq 1/2} + n_1^{b_1+d_1-1/2}\left[n_1\xi_{(1)}\right]^{1-b_1-d_1}.
\end{align*}
The first term is $o_P(1)$ since it converges to zero in $L_1$:
\begin{equation*}
    E\left[\sqrt{n_1}\ind{\xi_{(1)}\geq 1/2}\right]=\sqrt{n_1}\Pr\left(\xi_{(1)}\geq 1/2\right)=\frac{\sqrt{n_1}}{2^{n_1}}=o(1).
\end{equation*}
That the second term is $o_P(1)$ follows from $E[\xi_{(1)}]=1/(n_1+1)$ and Jensen's inequality since $b_1+d_1<1/2$. 
This shows that
\begin{equation}\label{eq:tail_up}
\sqrt{N}\left|\int_0^{\xi_{(1)}} F_Y^{-1}dF_U\right|=o_P(1).
\end{equation}
An analogous reasoning shows that
\begin{equation}\label{eq:tail_low}
\sqrt{N}\left|\int_{\xi_{(n_1)}}^1 F_Y^{-1}dF_U\right|=o_P(1).
\end{equation}
It follows that $R_2=o_P(1)$. 

\medskip
Third, we show that $R_3=o_P(1)$. By the mean value theorem, there exists $u_{n_1}(t)\in (\Gn(t), t)$ such that 
$$R_3 = \sqrt{N}\int_{\xi_{(1)}}^{\xi_{(n_1)}} \underbrace{\left[f_U - f_U \circ u_{n_1} \right]}_{=:\nu_{n_1}}\left[\Gn - \I\right]dF_Y^{-1}.
$$ 
By Assumption~\ref{ass:smoothness1}\ref{ass:smoothness14}, there exists $\delta>0$ such that $b_j+d_j<1/2-\delta$. Further, let $\delta_j>0$ be such that $\delta>2(\delta_1\vee\delta_2)$ and
\begin{equation*}
b_j+d_j<1/2-\delta - \delta_j.    
\end{equation*}
Then, let $q(t) = t^{1/2-\delta_1}(1-t)^{1/2-\delta_2}$. From what precedes, we have
\begin{equation}
 \left \vert R_3\right \vert  \leq  \sup_{t\in(0,1)} \left \vert \frac{\sqrt{n_1}(\Gn(t) - t )}{q(t)} \right\vert \int_{\xi_{(1)}}^{\xi_{(n_1)}}|\nu_{n_1}(t)|q(t)dF^{-1}_Y(t).    
\label{eq:ineg_Rn}
\end{equation}
We now show that the integral term in \eqref{eq:ineg_Rn} tends to 0 in $L_1$. Let $f_{n_1}(t):=\ind{\xi_{(1)}<t<\xi_{(n_1)}}|\nu_{n_1}(t)|$ and let $\mu$ denote the measure defined as the product of the probability measure on $\Omega$, the probability space on which all random variables are defined, with the measure $\rho$ on $(0,1)$ such that $d\rho/dF_Y^{-1}=q$. By Lemma~\ref{lem:int_FY}, $\mu(\Omega\times(0,1))=\int_0^1q(t)dF_Y^{-1}(t)<\infty$. Moreover, 
\begin{equation*}
    E\left[\int_{\xi_{(1)}}^{\xi_{(n_1)}}|\nu_{n_1}(t)|q(t)dF^{-1}_Y(t)\right] = \int_{\Omega\times(0,1)}f_{n_1}d\mu.
\end{equation*}
First, we show that $f_{n_1}$ tends to 0 $\mu$-almost surely. By almost-sure convergence of $\Gn(t)$ to $t$, we have, for all $t\in(0,1)$, $u_{n_1}(t)\convAS t$. Then, by continuity of $f_U$, $\nu_{n_1}(t)\convAS 0$ for all $t\in (0,1)$. This shows that $f_{n_1}\to0$ $\mu$-almost-surely. Second, by Corollary(i)--(ii) of Theorem 16.14 in \cite[][p.~218]{Billingsley1995}, it is sufficient to show that $f_{n_1}$ is uniformly integrable, i.e., that
\begin{equation*}
\lim_{M\to\infty}\limsup_{N\to\infty}\int\ind{f_{n_1}>M}f_{n_1}d\mu =0.
\end{equation*}
We will show the stronger result that, for all $M>0$,
\begin{equation}\label{eq:UI2}
    \limsup_{N\to\infty}\int\ind{f_{n_1}>M}f_{n_1}d\mu =0.
\end{equation}
Let $M>0$ and $\nu\in(1,\infty)$ such that 
\begin{equation}\label{eq:cond_nu}
    \nu<\frac{b_1+\delta-\delta_1}{b_1} \land \frac{b_2+\delta-\delta_2}{b_2}.
\end{equation}
By H\"older's inequality, 
\begin{equation}\label{eq:hold_change_measure}
     \int\ind{f_{n_1}>M}f_{n_1}d\mu 
     \leq \left(\int f_{n_1}^\nu d\mu\right)^{1/\nu}\left(\int \ind{f_{n_1}>M} d\mu\right)^{(\nu-1)/\nu}.
\end{equation}
First, we show that the first term in the right-hand side of \eqref{eq:hold_change_measure} is $O(1)$. 
\begin{align*}
    \int f_{n_1}^\nu d\mu & = \int_0^1E\left[\ind{\xi_{(1)}<t<\xi_{(n_1)}}\vert \nu_{n_1}(t)\vert ^\nu\right]q(t)dF_Y^{-1}(t).
\end{align*} 
Let $B(t):=C_Ut^{-b_1}(1-t)^{-b_2}$. By the triangle inequality, Assumption \ref{ass:smoothness1}, and since $B^\nu$ is convex and $u_{n_1}(t)$ lies between $\Gn(t)$ and $t$, for all $t\in (\xi_{(1)},  \xi_{(n_1)})$, almost surely,
\begin{align*}
 |\nu_{n_1}(t)|^\nu & \leq B(t)^\nu + B(u_{n_1}(t))^\nu \\[5pt]
 & \leq 2B(t)^\nu + B(\Gn(t))^\nu.
\end{align*}
Hence,
\begin{align*}
    \int f_{n_1}^\nu d\mu & \lesssim \int_0^1B(t)^\nu q(t)dF_Y^{-1}(t) +\int_0^1E\left[\ind{\xi_{(1)}<t<\xi_{(n_1)}}B(\Gn(t))^\nu\right]q(t)dF_Y^{-1}(t).
\end{align*}
By Lemma~\ref{lem:int_FY}, for $\nu$ sufficiently small, the first integral is finite and does not depend on $n_1$. Consider the second integral. By Lemma~\ref{lem:int_FY} again, to show that this integral is $O(1)$ it suffices to show that
\begin{equation}\label{eq:in1_bounded}
v_{n_1}(t):=E\left[\ind{\xi_{(1)}<t<\xi_{(n_1)}}B(\Gn(t))^\nu\right]q(t) \lesssim t^{1/2-b_1-\delta}(1-t)^{1/2-b_2-\delta},
\end{equation}
which we prove below. We have
\begin{align*}
v_{n_1}(t)&= E\left[\ind{1<n_1\Gn(t)<n_1}B\left(\Big(n_1\Gn(t)\Big)/n_1\right)^\nu\right]q(t) \\[5pt]
&=C_Ut^{1/2-\delta_1}(1-t)^{1/2-\delta_2}E\left[\ind{1<W<n_1}\left(\frac{W}{n_1}\right)^{-b_1\nu}\left(1-\frac{W}{n_1}\right)^{-b_2\nu}\right] \\[5pt]
& \lesssim t^{1/2-b_1-\delta}(1-t)^{1/2-b_2-\delta} w_{n_1}(t),
\end{align*}
where $W$$\sim$Binomial$(n_1,t)$ and
\begin{align*}
    w_{n_1}(t)&=t^{b_1+\delta-\delta_1}(1-t)^{b_2+\delta-\delta_2}E\left[\ind{1<W<n_1}\left(\frac{W}{n_1}\right)^{-b_1\nu}\left(1-\frac{W}{n_1}\right)^{-b_2\nu}\right]. 
\end{align*}
Hence, it suffices to show that $w_{n_1}(t)=O(1)$ uniformly in $t$. Since
\begin{align*}
   &  E\left[\ind{1<W<n_1}\left(\frac{W}{n_1}\right)^{-b_1\nu}\left(1-\frac{W}{n_1}\right)^{-b_2\nu}\right] \\[5pt]
   & \leq E^{1/2}\left[\ind{1<W<n_1}\left(\frac{W}{n_1}\right)^{-2b_1\nu}\right]E^{1/2}\left[\ind{1<W<n_1}\left(1-\frac{W}{n_1}\right)^{-2b_2\nu}\right],
\end{align*}
it suffices to show
\begin{align}
    E\left[\ind{1<W<n_1}\left(\frac{W}{n_1}\right)^{-2b_1\nu}\right]&\lesssim t^{-2b_1\nu}, \label{eq:to_Prove} \\
     E\left[\ind{1<W<n_1}\left(1-\frac{W}{n_1}\right)^{-2b_2\nu}\right]&\lesssim (1-t)^{-2b_2\nu}. \notag
\end{align}
We show the first result only (the second can be obtained analogously). We have
\begin{align*}
   & E\left[\ind{1<W<n_1}\left(\frac{W}{n_1}\right)^{-2b_1\nu}\right] \\[5pt]
   & = E\left[\ind{1<W<n_1}\ind{W\geq \frac{n_1t}{2}}\left(\frac{W}{n_1}\right)^{-2b_1\nu}\right]  \\[5pt]
   &\qquad + E\left[\ind{1<W<n_1}\ind{W< \frac{n_1t}{2}}\left(\frac{W}{n_1}\right)^{-2b_1\nu}\right] \\[5pt]
   &\lesssim  t^{-2b_1\nu} + E\left[\ind{W< \frac{n_1t}{2}}\left(\frac{1}{n_1}\right)^{-2b_1\nu}\right] \\[5pt]
   &= t^{-2b_1\nu} + n_1^{2b_1\nu} \Pr\left(W< \frac{n_1t}{2}\right).
\end{align*}
By Hoeffding's inequality, we have
\[
\Pr\left(W< \frac{n_1t}{2}\right) \leq \exp\left(-\frac{1}{8}n_1t\right).
\]
Since $x^\delta\exp(-x/8)$ is bounded on $[0,\infty)$ for any $\delta>0$, we have
\[
n_1^{2b_1\nu} \Pr\left(W< \frac{n_1t}{2}\right)\lesssim t^{-2b_1\nu}
\]
 and thus
 \[
  E\left[\ind{1<W<n_1}\left(\frac{W}{n_1}\right)^{-2b_1\nu}\right]\lesssim t^{-2b_1\nu},
 \]
 which proves \eqref{eq:to_Prove}, $w_{n_1}(t)=O(1)$, \eqref{eq:in1_bounded}, and in turn that the first term in~\eqref{eq:hold_change_measure} is $O(1)$.
 
\medskip
Next, we show that the second term in~\eqref{eq:hold_change_measure} converges to zero. We have
\begin{align}
   \int \ind{f_{n_1}>M} d\mu & =  \int_0^1\Pr\left(\ind{\xi_{(1)}<t<\xi_{(n_1)}}\vert \nu_{n_1}(t)\vert >M\right)q(t)dF_Y^{-1}(t) \notag \\[5pt]
   &\leq  \int_0^1\Pr\left(\vert \nu_{n_1}(t)\vert >M\right)q(t)dF_Y^{-1}(t). \label{eq:DCT_change_measure}
\end{align}
We have already shown that, for all $t\in(0,1)$, $\vert \nu_{n_1}(t)\vert \convAS 0$. It follows that, for all $t\in(0,1)$, $\abs{\nu_{n_1}}(t)=o_P(1)$. This implies that
\[
\lim_{N\to\infty}\Pr\left(\vert \nu_{n_1}(t)\vert >M\right)q(t)=0, \quad \forall t\in(0,1).
\]
Moreover, for all $t\in(0,1)$, $\Pr\left(\vert \nu_{n_1}(t)\vert >M\right)q(t)\leq q(t)$ with $\int_0^1q(t)dF_Y^{-1}(t)<\infty$ by Lemma~\ref{lem:int_FY}. By the dominated convergence theorem,
\[
\lim_{N\to\infty}\int_0^1\Pr\left(\vert \nu_{n_1}(t)\vert >M\right)q(t)dF_Y^{-1}(t)=0,
\]
and \eqref{eq:DCT_change_measure} implies that \eqref{eq:UI2} holds. This completes the proof that 
\begin{equation*}
    E\left[\int_{\xi_{(1)}}^{\xi_{(n_1)}}|\nu_{n_1}(t)|q(t)dF^{-1}_Y(t)\right]=o(1),
\end{equation*}
and thus 
\begin{equation}
    \int_{\xi_{(1)}}^{\xi_{(n_1)}}|\nu_{n_1}(t)|q(t)dF^{-1}_Y(t) =o_P(1).
    \label{eq:for_conv_Rn}
\end{equation}
To prove $R_3=o_P(1)$, it is thus sufficient to show that the first term in \eqref{eq:ineg_Rn} is bounded in probability. By Equation~(2) in Chapter 2, Section 7 (p.~141) in \cite{ShorackWellner1986}, there exists a Brownian bridge $\mathbb U$ such that
$$ \sup_{t \in (0,1)} \abs{\frac{\sqrt{n_1}(\Gn(t) -t) - \mathbb U(t) }{q(t)}} = o_P(1).$$
Since $||./q||$ is a norm, the triangular inequality yields
\begin{align*}
    \sup_{t\in(0,1)} \abs{\frac{\sqrt{n_1}(\Gn(t) - t )}{q(t)}} \leq \sup_{t\in(0,1)} |\mathbb U(t)/q(t)| + o_P(1) = O_P(1) + o_P(1) = O_P(1),
\end{align*}
by inequality $(17)$ p.~451 in \cite{ShorackWellner1986} with $a=0, b=1/2$, and by noticing that $\mathbb U/q$ has the same distribution on $[0, 1/2]$ and $[1/2, 1]$, and that the integral on the right-hand side of the inequality is finite for $q(t) = [t(1-t)]^a$ with $a < 1/2$. This result, together with \eqref{eq:for_conv_Rn} and \eqref{eq:ineg_Rn}, implies that $R_3=o_P(1)$.

\medskip
Fourth, we show that $R_4=o_P(1)$. We have
$$|R_4| \leq \sqrt{N} \int_0^{\xi_{(1)}}x d \Lambda(x) + \sqrt{N} \int_{\xi_{(n_1)}}^1(1-x) d \Lambda(x).$$
Consider the first term (the second term can be controlled analogously). Let $\delta\in (b_1+d_1,1/2)$. With probability tending to 1, $\xi_{(1)}\leq 1/2$. Moreover, on this event, we have, by Assumption \ref{ass:smoothness1}, 
\begin{align*}
\sqrt{N} \int_0^{\xi_{(1)}}x d \Lambda(x) &\lesssim \sqrt{N}\xi_{(1)}^{1-\delta}\int_0^{1/2} x^\delta d\Lambda(x) \\[5pt]
&\lesssim \sqrt{\frac{N}{n_1}}n_1^{\delta-1/2}\left[n_1\xi_{(1)}\right]^{1-\delta}\int_0^{1/2} x^{\delta-b_1}dF_Y^{-1}(x).
\end{align*}
By Lemma \ref{lem:int_FY}, $\int_0^{1/2} x^{\delta-b_1}dF_Y^{-1}(x)<\infty$. The result follows from $[n_1\xi_{(1)}]^{1-\delta}=O_P(1)$, $\sqrt{N/n_1}\leq 1$, and $\delta<1/2$. 

\hfill$\square$

\begin{lemma}\label{lem:12}
    $E[\abs{R_5}]\to0$ as $N\to\infty$.
\end{lemma}
{\bf Proof:}
 For that purpose, let $I_{n_1}(x)=(x,\Gn(x)]$ if $\Gn(x)>x$, $I_{n_1}(x)=[\Gn(x),x)$ if $\Gn(x)<x$, and $\emptyset$ otherwise. Finally, let $S_{n_1}(x)=\text{sgn}(\Gn(x)-x)$. Observe first that
\begin{equation}
\Vn\circ F_U\circ \Gn(x)-\Vn\circ F_U(x) = S_{n_1}(x)Z_N(x),    
    \label{eq:expr_diff_V}
\end{equation}
with 
$$Z_N(x) = \frac{1}{\sqrt{n_2}} \sum_{i=1}^{n_2} \left[\indic{U_i\in I_{n_1}(x)} - P_U(I_{n_1}(x)) \right],$$
where $P_U([a,b]) = P_U((a,b]) = P_U([a,b)) = P_U((a,b)) = F_U(a) - F_U(b)$ for all $(a,b) \in [0,1], a\leq b$.
Then, 
\begin{align}
E\left[|R_5| \; |(\xi_i)_i\right] & \leq E\left[\int_0^1 \left|\Vn \circ F_U \circ \Gn - \Vn\circ F_U\right| \, dF_Y^{-1} \bigg |(\xi_i)_i\right] \notag \\[5pt]
& = \int_0^1 E\left[\left|\Vn \circ F_U \circ \Gn - \Vn\circ F_U\right|\; |(\xi_i)_i\right]  dF_Y^{-1} \notag \\[5pt]
& \leq \int_0^1 E\left[Z_N(x)^2 |(\xi_i)_i\right]^{1/2} dF_Y^{-1}(x) \notag \\[5pt]
& = \int_0^1 V\left[\ind{U_1\in I_{n_1}(x)}|(\xi_i)_i\right]^{1/2} dF_Y^{-1}(x) \notag \\[5pt]
& \leq \int_0^1 |P_U(I_{n_1}(x))|^{1/2} dF_Y^{-1}(x).\label{eq:ineg1_Rn}
\end{align}
The first equality follows by Fubini--Tonelli's theorem, the second inequality uses \eqref{eq:expr_diff_V} and the Cauchy--Schwarz inequality, and the second equality holds since conditional on the $(\xi_i)_{i}$, the variables $\indic{U_i\in I_{n_1}(x)} - P_U(I_{n_1}(x))$ are i.i.d. with mean zero. As a result, 
\begin{align}
E\left[|R_5|\right] & \leq \int_0^1 E\left[|P_U(I_{n_1}(x)|^{1/2}\right] dF_Y^{-1}(x) \notag \\[5pt]
& \leq \int_0^1 E\left[|P_U(I_{n_1}(x)|\right]^{1/2} dF_Y^{-1}(x), \label{eq:ineg2_Rn}
\end{align}
where the first inequality follows by \eqref{eq:ineg1_Rn} and Fubini--Tonelli's theorem, whereas the second is due to Jensen's inequality. Now, by the law of large numbers and the continuous mapping theorem, $|F_U(\Gn(x))-F_U(x)|\convP 0$ for all $x\in[0,1]$. Moreover, $|F_U(\Gn(x))-F_U(x)|\leq 1$. Hence, for all $x\in[0,1]$,
$$E\left[|F_U(x) -F_U(\Gn(x))|\right] \to 0.$$
We now apply the dominated convergence theorem to prove that $E[|R_5|]\to 0$. Because $x \mapsto E[|P_U(I_{n_1}(x))|]^{1/2}$ is bounded by 1 for all $n_1$, it is actually enough to bound this function for $x$ close to 0 and close to 1. Also, by symmetry, we can focus without loss of generality on the neighborhood of 0. We prove that
\begin{equation}
E\left[|P_U(I_{n_1}(x))|\right] \lesssim x^{1-b_1}.
    \label{eq:for_DCT_Rn}
\end{equation}
Then the result follows by Lemma \ref{lem:int_FY} combined with Assumption \ref{ass:smoothness1}\ref{ass:smoothness14}. To prove \eqref{eq:for_DCT_Rn}, we apply Lemma \ref{lem:useful_lem} with $Q_n(x):=\Gn(x)$, $B_n(x)=1$,  and $\delta < \exp(-1)/2$. If $x\geq 1/n_1$, the Cauchy--Schwarz inequality yields 
\begin{equation}\label{eq:cauchy_schwarz}
    E\left[\abs{\Gn(x)-x} \right] \leq \left[\frac{x(1-x)}{n_1}\right]^{1/2} \leq 2 x,
\end{equation}
since $n_1^{1/2}\geq x^{-1/2}$. If $x<1/n_1$, \eqref{eq:cauchy_schwarz} holds as well by Theorem 1 in \cite{berend2013sharp}. Hence, \eqref{eq:cauchy_schwarz} holds for all $x\in (0,\bar \delta/2)$. Next, let $n_0\in \N$, $n_0\geq 4/(1-\bar \delta)^2$. By Kiefer's inequality \citep[see, e.g.][Corollary A.6.3]{vaartwellner1996}, we have, for all $x\in [0,\delta]$ and all $n_1\geq n_0$,
\begin{equation*}
\Pr(\Gn(x)>1/2) \leq (e x)^{n_1(1-\bar \delta)^2/4} \lesssim x.   
\end{equation*}
Thus, we can apply Lemma~\ref{lem:useful_lem}, which yields \eqref{eq:for_DCT_Rn}. 

\hfill $\square$

\begin{lemma}\label{lem:13}
    For all $j\in\{6,7,8,9\}$: $R_j=o_P(1)$ as $N\to\infty$.
\end{lemma}
{\bf Proof:} First, we show that $R_6=o_p(1)$. We split $R_6$ into three integrals: $R_{6,1}$ is obtained by integrating along the segment $[\xi_{(1)}, \zeta_{(1)}\vee \xi_{(1)}]$, $R_{6,2}$ is obtained by integrating along the segment $(\zeta_{(1)}\vee \xi_{(1)}, \zeta_{(n_3)}\wedge \xi_{(n_1)}]$, and $R_{6,3}$ is obtained by integrating on the segment $(\zeta_{(n_3)}\wedge \xi_{(n_1)}, \xi_{(n_1)}]$. Let us first focus on $R_{6,2}$. By the mean value theorem, for all $t\in [\zeta_{(1)}\vee \xi_{(1)},\zeta_{(n_3)}\wedge \xi_{(n_1)}]$, there exists $u_N(t)\in (\Gn(t), \Hn^{-1}\circ\Gn(t))$ such that 
\begin{align*}
    R_{6,2} & = \sqrt{\frac{N}{n_3}} \sqrt{n_3}\int_{\zeta_{(1)}\vee \xi_{(1)}}^{\zeta_{(n_3)}\wedge \xi_{(n_1)}} \underbrace{\left[f_U - f_U \circ u_N \right]}_{=:\nu_N} \left[\Hn^{-1}\circ \Gn - \Gn \right] dF_Y^{-1}.
\end{align*}
By Assumption \ref{ass:smoothness1}\ref{ass:smoothness14}, there exists $\delta>0$ such that $b_j+d_j<1/2-\delta$. Further, let $\delta_j>0$ be such that $\delta>2(\delta_1\vee\delta_2)$ and
\begin{equation}
b_j+d_j<1/2-\delta - \delta_j.    
    \label{eq:delta_securite}
\end{equation}
Then let $q(t) = t^{1/2-\delta_1}(1-t)^{1/2-\delta_2}$. From what precedes, we have
\begin{equation}
 \left \vert R_{6,2}\right \vert  \leq \sqrt{\frac{N}{n_3}} \sup_{t\in(1/n_1,1-1/n_1)} \left \vert \frac{\sqrt{n_3}(\Hn^{-1}(t) - t )}{q(t)} \right\vert \int_{\zeta_{(1)}\vee \xi_{(1)}}^{\zeta_{(n_3)}\wedge \xi_{(n_1)}} |\nu_N(t)|q(t)  dF^{-1}_Y(t).    
\label{eq:ineg_R1n}
\end{equation}
We now show that the last term in~\eqref{eq:ineg_R1n} tends to 0 in $L_1$. Let 
\[
f_{N}(t):=\ind{\zeta_{(1)}\vee \xi_{(1)}<t<\zeta_{(n_3)}\wedge \xi_{(n_1)}}|\nu_N(t)|,
\]
and let $\mu$ denote the measure defined as the product of the probability measure on $\Omega$, the probability space on which all random variables are defined, with the measure $\rho$ on $(0,1)$ such that $d\rho/dF_Y^{-1}=q$. By Lemma~\ref{lem:int_FY}, $\mu(\Omega\times(0,1))=\int_0^1q(t)dF_Y^{-1}(t)<\infty$. Moreoever, 
\begin{equation*}
    E\left[\int_{\zeta_{(1)}\vee \xi_{(1)}}^{\zeta_{(n_3)}\wedge \xi_{(n_1)}} |\nu_N(t)|q(t)dF^{-1}_Y(t)\right] = \int_{\Omega\times(0,1)}f_Nd\mu.
\end{equation*}
First, we show that $f_N$ tends to 0 $\mu$-almost surely. By convergence of $\Gn(t)$ and $\Hn^{-1}(t)$ to $t$, we have, for all $t\in(0,1)$, $u_N(t)\convAS t$. Then, by continuity of $f_U$, $\nu_N(t)\convAS 0$ for all $t\in (0,1)$. This shows that $f_N\to0$ $\mu$-almost-surely. Second, by Corollary(i)--(ii) of Theorem 16.14 in \cite[][p. 218]{Billingsley1995}, it is sufficient to show that $f_N$ is uniformly integrable, i.e., that
\begin{equation}\label{eq:UI1_bis}\lim_{M\to\infty}\limsup_{N\to\infty}\int\ind{f_N>M}f_Nd\mu =0.
\end{equation}
We will show the stronger result that, for all $M>0$,
\begin{equation}\label{eq:UI2_bis}
    \limsup_{N\to\infty}\int\ind{f_N>M}f_Nd\mu =0.
\end{equation}
Let $M>0$ and $\nu\in(1,\infty)$ such that \eqref{eq:cond_nu} holds. By H\"older's inequality, 
\begin{equation}\label{eq:hold_change_measure_bis}
     \int\ind{f_N>M}f_Nd\mu 
     \leq \left(\int f_N^\nu d\mu\right)^{1/\nu}\left(\int \ind{f_N>M} d\mu\right)^{(\nu-1)/\nu}.
\end{equation}
Consider the first term. 
\begin{align*}
    \int f_N^\nu d\mu & = \int_0^1E\left[\int_{\zeta_{(1)}\vee \xi_{(1)}}^{\zeta_{(n_3)}\wedge \xi_{(n_1)}}\vert \nu_N(t)\vert ^\nu q(t)\right]dF_Y^{-1}(t).
\end{align*} 
Let $B(t):=C_Ut^{-b_1}(1-t)^{-b_2}$. By Assumption \ref{ass:smoothness1} and since $B^\nu$ is convex and $u_N(t)$ lies between $\Gn(t)$ and $t$, for all $t\in (\xi_{(1)}, \xi_{(n_1)})$, 
\begin{align*}
 |\nu_N(t)|^\nu & \leq 2^{\nu-1}[B(t)^\nu + B(u_N(t))^\nu] \\[5pt]
 & \leq 2^{\nu-1}[2B(t)^\nu + B(\Gn(t))^\nu + B(\Hn^{-1}(\Gn(t)))^\nu], \quad \text{a.s.}.
\end{align*}
Hence,
\begin{align*}
    \int f_N^\nu d\mu & \lesssim \int_0^1B(t)^\nu q(t)dF_Y^{-1}(t) +\int_0^1E\left[\ind{\xi_{(1)}<t<\xi_{(n_1)}}B(\Gn(t))^\nu q(t)\right]dF_Y^{-1}(t) \\[5pt]
    & \quad +\int_0^1E\left[\ind{\zeta_{(1)}\vee \xi_{(1)}<t<\zeta_{(n_3)}\wedge \xi_{(n_1)}}B(\Hn^{-1}(\Gn(t)))^\nu q(t)\right]dF_Y^{-1}(t).
\end{align*}
The first two integrals are uniformly bounded as already shown in the proof of Lemma~\ref{lem:11}. Let us focus on the third integral. 
\begin{align*}
   &\int_0^1E\left[\ind{\zeta_{(1)}\vee \xi_{(1)}<t<\zeta_{(n_3)}\wedge \xi_{(n_1)}}B(\Hn^{-1}(\Gn(t)))^\nu q(t)\right]dF_Y^{-1}(t) \\[5pt]
   &\quad \leq \int_0^1E\left[\zeta_{(1)}^{\delta/2+(1-\nu) b_1-\delta_1}(1-\zeta_{(n_3)})^{\delta/2+(1-\nu)b_2-\delta_2}\right]t^{1/2-\delta-b_1}(1-t)^{1/2-\delta-b_2}dF_Y^{-1}(t)  \\[5pt]
   &\quad =n_3^{(b_1+b_2)(\nu-1)+\delta_1+\delta_2-\delta}E\left[\big(n_3\zeta_{(1)}\big)^{\delta/2+(1-\nu)b_1-\delta_1}\big(n_3(1-\zeta_{(n_3)})\big)^{\delta/2+(1-\nu)b_2-\delta_2}\right]\\[5pt]
   &\qquad \times \int_0^1t^{1/2-\delta-b_1}(1-t)^{1/2-\delta-b_2}dF_Y^{-1}(t).
\end{align*}
For $\nu$ sufficiently small, $n_3^{(b_1+b_2)(\nu-1)+\delta_1+\delta_2-\delta}\leq 1$ and $x\mapsto x^{\delta/2+(1-\nu)b_j-\delta_j}$ is concave. By the Cauchy--Schwarz and Jensen inequalities,
\begin{align*}
   & E\left[\big(n_3\zeta_{(1)}\big)^{\delta/2+(1-\nu)b_1-\delta_1}\big(n_3(1-\zeta_{(n_3)})\big)^{\delta/2+(1-\nu)b_2-\delta_2}\right] \\[5pt]
   &\quad \leq \left(E\left[\big(n_3\zeta_{(1)}\big)^{2(\delta/2+(1-\nu)b_1-\delta_1)}\right]\right)^{1/2}\left(E\left[\big(n_3(1-\zeta_{(n_3)})\big)^{2(\delta/2+(1-\nu)b_2-\delta_2)}\right]\right)^{1/2} \\[5pt]
   &\quad \leq E\left[n_3\zeta_{(1)}\right]^{\delta/2+(1-\nu)b_1-\delta_1}E\left[n_3(1-\zeta_{(n_3)})\right]^{\delta/2+(1-\nu)b_2-\delta_2} \\[5pt]
   &\quad = \left(\frac{n_3}{n_3+1}\right)^{\delta/2+(1-\nu)b_1-\delta_1}\left(\frac{n_3}{n_3+1}\right)^{\delta/2+(1-\nu)b_2-\delta_2} \\[5pt]
   &\quad \leq 1.
\end{align*}
Combining all pieces together yields
\begin{align*}
   &\int_0^1E\left[\ind{\zeta_{(1)}\vee \xi_{(1)}<t<\zeta_{(n_3)}\wedge \xi_{(n_1)}}B(\Hn^{-1}(\Gn(t)))^\nu q(t)\right]dF_Y^{-1}(t) \\[5pt]
   &\quad \leq \int_0^1t^{1/2-\delta-b_1}(1-t)^{1/2-\delta-b_2}dF_Y^{-1}(t),
\end{align*}
which does not depend on $N$ and is finite by Lemma~\ref{lem:int_FY}. Next, let us show that the second term in~\eqref{eq:hold_change_measure_bis} converges to zero.
\begin{align}
   \int \ind{f_N>M} d\mu & = \int_0^1\Pr\left(\ind{\zeta_{(1)}\vee\xi_{(1)}<t<\zeta_{(n_3)}\wedge \xi_{(n_1)}}\vert \nu_N(t)\vert >M\right)q(t)dF_Y^{-1}(t) \notag \\[5pt]
   &\leq  \int_0^1\Pr\left(\vert \nu_N(t)\vert >M\right)q(t)dF_Y^{-1}(t). \label{eq:DCT_change_measure_bis}
\end{align}
Since we have already shown that, for all $t\in(0,1)$, $\vert \nu_N(t)\vert \convAS 0$, it follows that, for all $t\in(0,1)$, $\vert \nu_N(t)\vert \convP 0$. This implies that
\[
\lim_{N\to\infty}\Pr\left(\vert \nu_N(t)\vert >M\right)q(t)=0, \quad \forall t\in(0,1).
\]
Moreover, for all $t\in(0,1)$, $\Pr\left(\vert \nu_N(t)\vert >M\right)q(t)\leq q(t)$ with $\int_0^1q(t)dF_Y^{-1}(t)<\infty$ by Lemma~\ref{lem:int_FY}. By the dominated convergence theorem,
\[
\lim_{N\to\infty}\int_0^1\Pr\left(\vert \nu_N(t)\vert >M\right)q(t)dF_Y^{-1}(t)=0,
\]
and \eqref{eq:DCT_change_measure_bis} implies that \eqref{eq:UI2_bis} holds. This completes the proof that 
\begin{equation*}
    E\left[\int_{\zeta_{(1)}\vee \xi_{(1)}}^{\zeta_{(n_3)}\wedge \xi_{(n_1)}} |\nu_N(t)|q(t)dF^{-1}_Y(t)\right]=o(1),
\end{equation*}
and thus 
\begin{equation}
   \int_{\zeta_{(1)}\vee \xi_{(1)}}^{\zeta_{(n_3)}\wedge \xi_{(n_1)}}|\nu_N(t)|q(t)dF^{-1}_Y(t) \convP 0.
    \label{eq:for_conv_R1n}
\end{equation}
Next, by Corollary 4.3.1 and Theorem 3.4 in \cite{csorgo1986}, 
$$\sup_{t\in(1/n_3,1-1/n_3)} \abs{\frac{\sqrt{n_3}(\Hn^{-1}(t) - t )}{q(t)}} = O_P(1).$$

Let $\eps>0$. For $N$ sufficiently large, $1/n_1>1/n_3-\eps$. It follows that
\begin{align*}
     \sup_{t\in(1/n_1,1-1/n_1)} \left \vert \frac{\sqrt{n_3}(\Hn^{-1}(t) - t )}{q(t)} \right\vert &\leq \sup_{t\in(1/n_3-\eps,1-1/n_3+\eps)} \left \vert \frac{\sqrt{n_3}(\Hn^{-1}(t) - t )}{q(t)} \right\vert \\[5pt]
     &= \sup_{t\in(1/n_3,1-1/n_3)} \left \vert \frac{\sqrt{n_3}(\Hn^{-1}(t) - t )}{q(t)} \right\vert+O_P(1) \\[5pt]
     & = O_P(1).
\end{align*}
This, together with \eqref{eq:ineg_R1n}, $\sqrt{\frac{N}{n_3}} =O_P(1)$, and \eqref{eq:for_conv_R1n}, implies that $R_{6,2}=o_P(1)$. Now, let us show that $R_{6,1}=o_P(1)$.
\begin{align*}
   \abs{ R_{6,1}}  & :=\abs{ \sqrt{\frac{N}{n_3}} \sqrt{n_3}  \int_{\xi_{(1)}}^{\zeta_{(1)}\vee\xi_{(1)}} A_N(t)\left[\Hninv \circ \Gn(t) - \Gn(t) \right] dF_Y^{-1}(t)},\\[5pt]
    &\leq  \sqrt{\frac{N}{n_3}} \abs{F_Y\big(\zeta_{(1)}\big)-F_Y\big(\xi_{(1)}\big)} \ind{\xi_{(1)}<\zeta_{(1)}} \\[5pt]
    &\quad \times \sqrt{n_3}(\zeta_{(j_N)}+n_1^{-1}) \sup_{\xi_{(1)}<t<\zeta_{(1)}}\left[2B(t)+B(\Gn(t))+B(\Hn^{-1}(\Gn(t)))\right] \\[5pt]
    &\leq O_P(1)o_P(1)\times n_1^{b_1-1/2} \left((n_1/n_3)^{1/2}\left[n_3\zeta_{(j_N)}\right]+(n_1/n_3)^{-1/2}\right)\left(\left[n_1\xi_{(1)}\right]^{-b_1}+1+(n_1/n_3)^{-b_1}\left[n_3\zeta_{(1)}\right]^{-b_1}\right)
\end{align*}
where $j_N= \lceil n_1/n_3 \rceil$. By assumption, there exists $\eps>0$ such that, for $N$ sufficiently large, $|\frac{n_1}{n_3} - \frac{\lambda_3}{\lambda_1}|<\eps$. By choosing $\eps$ sufficiently small, this ensures that for $N$ sufficiently large $N$, $j_N \leq \lceil \eps + \lambda_3/\lambda_1 \rceil \leq \lceil \lambda_3/\lambda_1 \rceil + 1 =: \bar j$. In particular, $n_3\zeta_{(j_N)}\leq n_3\zeta_{(\bar j)}=O_P(1)$. Since $[n_1\xi_{(1)}]^{-b_1}=O_P(1)$ and  $[n_3\zeta_{(1)}]^{-b_1}=O_P(1)$, and $b_1<1/2$ by Assumption~\ref{ass:smoothness1}\ref{ass:smoothness14},
\begin{align*}
    R_{6,1}  \leq O_P(1)o_P(1)o_P(1)O_P(1) = o_P(1).
\end{align*}
An analogous reasoning yields $R_{6,3}=o_P(1)$.

\medskip
Second, we show that $R_{7}$ is defined in \eqref{eq:def_R2n}. We actually prove the stronger result that $R_{7}$ converges to $0$ in $L^1$. Let $\Wn = \sqrt{n_3}(\Hn^{-1} - \I)$ and $B_{n_1}=\ind{\xi_{(1)}\leq x<\xi_{(n_1)}}$. We have, by Fubini--Tonelli's theorem,
\begin{align*}
E[|R_7|] & \leq \sqrt{\frac{N}{n_3}} \int_0^1 E[|\Wn \circ \Gn(x) - \Wn(x)|\times B_{n_1}]f_U(x)dF_Y^{-1}(x) \\[5pt]
& \leq \sqrt{\frac{N}{n_3}}\int_0^1 E[(\Wn \circ \Gn(x) - \Wn(x))^2 \times B_{n_1}]^{1/2} f_U(x)dF_Y^{-1}(x).    
\end{align*}
We apply the dominated convergence theorem to prove the result. First, note that for all $(x,y) \in (0,1]^2$,
$$|\Hninv(x) - \Hninv(y)| \sim \text{Beta}(|\lceil n_3x \rceil -\lceil n_3y \rceil|, n_3-|\lceil n_3x \rceil -\lceil n_3y \rceil|+1),$$

with the convention that the $\text{Beta}(0,n_3+1)$ is the Dirac distribution at 0. Hence, for any $k\in\{1,\ldots,n_1-1\}$,
\begin{align}
    & E\left[(\Wn \circ \Gn(x) - \Wn(x))^2\; |\Gn(x)=k/n_1\right] \notag \\[5pt]
    & = n_3 \left\{E\left[\left(\Hninv(k/n_1) -\Hninv(x) - (k/n_1-x)\right)^2\right]\right\} \notag \\[5pt]
    &= n_3 \left\{E\left[\left(\Hninv(k/n_1) -\Hninv(x) - \frac{1}{n_3+1}\left (\left \lceil (n_3k)/n_1 \right \rceil - \lceil n_3x \rceil  \right) \right. \right.\right. \notag \\[5pt]
    & \quad \left.\left. \left. + \frac{1}{n_3+1}\left (\left \lceil (n_3k)/n_1 \right \rceil - \lceil n_3x \rceil  \right) - (k/n_1-x) \right)^2\right]\right\} \notag \\
    &= n_3 \left\{V\left[\Hninv(k/n_1) -\Hninv(x) \right] +  \frac{1}{(n_3+1)^2}(\lceil (n_3k)/n_1 \rceil - (n_3k)/n_1 + n_3x - \lceil n_3x \rceil + x - k/n_1)^2\right\} \notag \\[5pt]
    & = \frac{n_3}{(n_3+1)^2(n_3+2)}\left|\lceil(n_3k)/n_1\rceil  - \left \lceil{n_3x}\right \rceil \right|\left(n_3+1-\left|\lceil(n_3k)/n_1\rceil - \left \lceil{n_3x}\right \rceil \right| \right) \\[5pt]
    & \quad +  \frac{n_3}{(n_3+1)^2}(\lceil (n_3k)/n_1 \rceil - (n_3k)/n_1 + n_3x - \lceil n_3x \rceil + x - k/n_1)^2 \notag \\[5pt]
    & \leq  \frac{1}{n_3}\left|\lceil (n_3k)/n_1 \rceil - \left \lceil{n_3x}\right \rceil \right| + \frac{2}{n_3}\left[2\left(\lceil (n_3k)/n_1 \rceil -(n_3k)/n_1 \right)^2 + 2\left(\lceil n_3x \rceil - n_3x \right)^2 + \left(\frac{k}{n_1} - x\right)^2 \right]  \notag \\[5pt]
    & \leq \left|\frac{k}{n_1} - x\right| + \frac{1}{n_3}\left[10 + 2 \left(\frac{k}{n_1} - x\right)^2 \right], \label{eq:inter_for_R2n} 
\end{align}
where the first inequality follows by convexity and the last by the triangle inequality and because by definition, $\left|a - \left \lceil{a}\right \rceil \right|\leq 1$ for all $a \in \R_+$. Now, remark that $B_{n_1}=1$ iff $n_1\Gn(x)\in\{1,\ldots,n_1-1\}$. Then, by what precedes,
\begin{align}
E\left[(\Wn \circ \Gn(x) - \Wn(x))^2\times B_{n_1}\right]& \leq E\left[\left|\Gn(x) - x\right|\right] + \frac{1}{n_3}\left[10 +  2 V(\Gn(x))\right] \label{eq:ineg_for_R2n}\\
&\to 0.   \notag 
\end{align}
To apply the dominated convergence theorem, we show that there exists $q(.)$ such that for all $n_1\geq n_0$ and all $x\in [0,1]$,
\begin{equation}
E[(\Wn \circ \Gn(x) - \Wn(x))^2 \times B_{n_1}]^{1/2}\leq q(x),    
    \label{eq:for_DCT_R2n}
\end{equation}
with $\int_0^1 q(x) f_U(x)dF_Y^{-1}(x)<\infty$. As above, we focus on a neighborhood of 0. If $x>1/n_3$, we have, by \eqref{eq:ineg_for_R2n} and \eqref{eq:cauchy_schwarz},
$$E\left[(\Wn \circ \Gn(x) - \Wn(x))^2\times B_{n_1}\right] \leq 14x.$$
Now suppose that $x<1/n_3$. Remark that $E(B_{n_1})\leq 1-(1-x)^{n_1} \leq n_1x$. Then, integrating \eqref{eq:inter_for_R2n}, we obtain 
\begin{align*}
E\left[(\Wn \circ \Gn(x) - \Wn(x))^2\times B_{n_1}\right]&  \leq
E\left[|\Gn(x) - x|\right] + \frac{1}{n_3}\left[10n_1x + 2V(\Gn(x))\right]\\[5pt]
& \leq (\lceil \lambda_3/\lambda_1 \rceil +1) 14x. 
\end{align*}
Then we can choose $q(x)=((\lceil \lambda_3/\lambda_1 \rceil +1) 14x)^{1/2}$ in \eqref{eq:for_DCT_R2n}. By Assumption \ref{ass:smoothness1} and Lemma \ref{lem:int_FY}, we have $\int_0^{1/2} q(x) f_U(x)dF_Y^{-1}(x)<\infty$. The same reasoning applies to the interval $[1/2,1]$. The result follows.

\bigskip
Third, we show that $R_{8}=o_P(1)$. Let 
$\Lambda$ denote the measure on $(0,1)$ such that 
$d\Lambda/dF_Y^{-1} =f_U$. Note that $\Hn^{-1}(x) \sim$Beta$(\lceil n_3x \rceil, n_3 + 1 -\lceil n_3x \rceil)$, thus $E[\Hn^{-1}(x)]=\lceil n_3x \rceil/(n_3+1)$. Then

$$E[|R_8|]\leq \sqrt{\frac{N}{n_3}} \int_0^1 [1-x^{n_1}-(1-x)^{n_1}] 
\left|\frac{\lceil n_3x \rceil - (n_3+1)x}{(n_3+1)n_3^{-1/2}}\right| d \Lambda(x).$$
Let $f_N(x)$ denote the integrand. We have $\lim_{N\to\infty} f_N(x)=0$. Moreover, using $1-x^{n_1} - (1-x)^{n_1} \leq n_1x$ and since $n_1 \lesssim n_3$, we obtain, when $x<1/n_3$, 
$$f_N(x) \leq 2 n_3^{1/2}x \leq x^{1/2} \lesssim [x(1-x)]^{1/2}.$$ 
When $x\in [1/n_3,1-1/n_3]$, 
$$f_N(x) \leq \frac{2}{n_3^{1/2}} \lesssim [x(1-x)]^{1/2}.$$ 
Finally, when $x>1-1/n_3$, then $x>1-1/n_1$, and thus using $1-x^{n_1}\leq n_1(1-x)$,
$$f_N(x) \leq n_1(1-x) \frac{2}{(n_3+1)n_3^{-1/2}} \leq 2(1-x)^{1/2} \lesssim [x(1-x)]^{1/2}.$$
Moreover, $\int_0^1 [x(1-x)]^{1/2}d\Lambda<\infty$ by Lemma \ref{lem:int_FY}. Thus, by the dominated convergence theorem, $R_8= o_P(1)$. 

\medskip
Fourth, we show that $R_{9}=o_P(1)$. We prove the stronger result that $R_{9}$ converges to $0$ in $L^1$. By Fubini--Tonelli's theorem combined with Assumption~\ref{ass:sampling_est}\ref{ass:sampling_est2} and Jensen's inequality, we have
\begin{align}
    E[|R_9|] & \leq \sqrt{\frac{N}{n_3}} \int_0^1\sqrt{n_3}E\left[\left|\indic{x \in [\xi_{(1)}, \xi_{(n_1)}]} - \indic{x \in [1/n_3, (n_3-1)/n_3]} \right| \right] \notag \\[5pt]
    & \qquad \times E\left[\left(\Hninv(x) - \frac{\lceil n_3x\rceil}{(n_3+1)}  \right)^2 \right]^{1/2} d \Lambda(x). \label{eq:bound_R5}
\end{align}
Since $\Hninv(x) \sim$Beta$(\lceil n_3x \rceil, n_3+1 - \lceil n_3x \rceil)$, we have
$$E\left[\left(\Hninv(x) - \frac{\lceil n_3x\rceil}{(n_3+1)}  \right)^2 \right]^{1/2} = \sqrt{\frac{\lceil n_3x\rceil(n_3+1 - \lceil n_3x\rceil)}{(n_3+1)^2(n_3+2)}} \lesssim \sqrt{\frac{x(1-x)}{n_3}}.$$
Let $q_N(x)$ denote the first expectation in the integrand in \eqref{eq:bound_R5}. By letting $p_{n_1}(x):=1-x^{n_1} - (1-x)^{n_1}$, we have
\begin{align*}
    q_N(x) & = \Pr(\xi_{(1)} \leq x \leq \xi_{(n_1)}, x<1/n_3) + \Pr(\xi_{(1)} \leq x \leq \xi_{(n_1)}, x>(n_3-1)/n_3)\\[5pt]
    &  + \Pr(\xi_{(1)} > x \cup x > \xi_{(n_1)}, 1/n_3 \leq x \leq (n_3-1)/n_3) \\[5pt]
    & = p_{n_1}(x)\left[\indic{x<1/n_3} + \indic{1-x<1/n_3}\right] + (1-p_{n_1}(x))\indic{ 1/{n_3} \leq x \leq (n_3-1)/n_3}.
\end{align*}
Let $f_N(x)$ denote the integrand in the right-hand side of \eqref{eq:bound_R5}. For all $x\in(0,1)$, $\lim_{N \to \infty} p_{n_1}(x)= 1$ so from what precedes,  $\lim_{N \to \infty}f_N(x)=0$ for all $x\in(0,1)$. Moreover, using $q_N(x)\leq 1$, we get
$$f_N(x) \lesssim [x(1-x)]^{1/2},$$ with $\int_0^1[x(1-x)]^{1/2} d \Lambda<\infty$ by Lemma~\ref{lem:int_FY}. The result follows by the dominated convergence theorem.

\hfill $\square$

\begin{lemma}\label{lem:14}
As $N\to\infty$:
    $$\sqrt{n_3}J_{2}\convD \mathcal{N}(0,\varsigma^2),$$
    where
    $$
    \varsigma^2 = \int_0^1\int_0^1 \left[s \wedge t - st\right]f_U(s)f_U(t)\, dF_Y^{-1}(s) \,dF_Y^{-1}(t).
    $$
\end{lemma}
{\bf Proof:} Let $I_{in_3}=[(i-1)/n_3, i/n_3)$. First, note that 
\begin{equation}
-\sqrt{n_3}J_2 =\sum_{i=1}^{n_3}a_{in_3}\left(\zeta_{(i)} - \frac{i}{(n_3+1)}\right),    
    \label{eq:J2_Lstat}
\end{equation}
where $a_{1n_3} = a_{n_3n_3} =0$, and, for all $i\in\{2, \dotsc, n_3-1\}$, $a_{in_3}  = \sqrt{n_3}\Lambda\left(I_{in_3}\right)$.  We now verify that the necessary and sufficient conditions given by \cite{hecker1976} for the asymptotic normality of the L-statistic in \eqref{eq:J2_Lstat} hold in our case. Let us define
$$\sigma_{n_3}^2 = \frac{1}{n_3+2}\sum_{j=1}^{n_3}\sum_{k=1}^{n_3}a_{jn_3}a_{kn_3} \left[ \left( \frac{j}{n_3+1} \wedge \frac{k}{n_3+1} \right) - \frac{jk}{(n_3+1)^2} \right].$$
We have to prove that
\begin{equation}
    \lim_{n_3\to\infty} \frac{\max_{1\leq i\leq n_3}\left | \sum_{j=i}^{n_3}a_{jn_3} \right|}{n_3\sigma_{n_3}}=0.
    \label{eq:cns_hecker}
\end{equation}
First, by Assumption~\ref{ass:smoothness1} and  Lemma~\ref{lem:int_FY}, there exists $\delta<1/2$ such that 
$$\int_0^1 t^{\delta-b_1}(1-t)^{\delta-b_2}  dF_Y^{-1}(t) <\infty.$$
Now, because $a_{in_3}\geq 0$, we have, for all $n_3\geq 2$,
\begin{align*}
    \max_{1\leq i\leq n_3}\left | \sum_{j=i}^{n_3}a_{jn_3} \right | & = \sqrt n_3 \sum_{j=2}^{n_3-1} \Lambda\left(I_{jn_3}\right) \\[5pt]
    & =  \sqrt n_3\int_{1/n_3}^{(n_3-1)/n_3} f_U(t)  dF_Y^{-1}(t) \\[5pt]
    & \leq C_U\sqrt{n_3} \int_{1/n_3}^{(n_3-1)/n_3} t^{-b_1}(1-t)^{-b_2}  dF_Y^{-1}(t) \\[5pt]
    & \leq C_U 2^\delta n_3^{1/2+\delta} \int_{1/n_3}^{(n_3-1)/n_3} t^{\delta-b_1}(1-t)^{\delta-b_2}  dF_Y^{-1}(t) \\[5pt]
    & \leq C_U 2^\delta n_3^{1/2+\delta} \int_0^1 t^{\delta-b_1}(1-t)^{\delta-b_2}  dF_Y^{-1}(t),
    \end{align*}
where the first inequality follows by Assumption~\ref{ass:smoothness1} and the second uses the fact that $[t(1-t)]^\delta \geq 1/(2n_3)^\delta$ for all $t\in[1/n_3,1-1/n_3]$. Therefore,
\begin{equation}
\max_{1\leq i\leq n}\left | \sum_{j=i}^{n_3}a_{jn_3} \right | = O(n_3^{1/2+\delta}).    
    \label{eq:max_an}
\end{equation}
Next, we have
\begin{align*}
    \sigma^2_{n_3} & = \frac{n_3}{n_3+2} \sum_{j=2}^{n_3-1}
    \sum_{k=2}^{n_3-1}
    \Lambda\left(I_{jn_3}\right)\Lambda\left(I_{kn_3}\right) \left( \frac{j}{n_3+1} \wedge \frac{k}{n_3+1} - \frac{jk}{(n_3+1)^2} \right) \\
    & = \frac{n}{n_3+2} \int_0^1\int_0^1 f_{n_3}(x,y)d\Lambda(x)d\Lambda(y),
\end{align*}
where $f_{n_3}(x,y)= \frac{j}{n_3+1} \wedge \frac{k}{n_3+1} - \frac{jk}{(n_3+1)^2}$ when $(x,y)\in I_{jn_3}\times I_{kn_3}$, $1<j\wedge k\leq j\vee k<n_3$, $f_{n_3}(x,y)=0$ otherwise. For any $(x,y)\in(0,1)^2$, $f_{n_3}(x,y)\to f(x,y):=x\wedge y - xy$. Moreover, for any $(x,y)\in I_{jn_3}\times I_{kn_3}$, $1<j\wedge k\leq j\vee k<n_3$,
\begin{align*}
\frac{j}{n_3+1}\wedge \frac{k}{n_3+1} & \leq 2( x\wedge y),\\
1 - \frac{j}{n_3+1}\vee \frac{k}{n_3+1} & \leq 2\left(1- x\vee y\right).
\end{align*}
Thus, $f_{n_3}(x,y)\leq 4 f(x,y)$ for all $(x,y)\in[1/n_3,1-1/n_3]^2$. This inequality also holds for $(x,y)\in[0,1]^2\backslash [1/n_3,1-1/n_3]^2$ since $f_{n_3}(x,y)=0$ for such $(x,y)$. Because $x\wedge y \leq (xy)^{1/2}$ and $1-x\vee y\leq [(1-x)(1-y)]^{1/2}$, we have $f(x,y)\leq [x(1-x)y(1-y)]^{1/2}$. Moreover, by Lemma \ref{lem:int_FY}, $\int_0^1 [\I(1-\I)]^{1/2}d\Lambda<\infty$. Thus, by the dominated convergence theorem,
\begin{equation}
\lim_{n_3\to\infty} \sigma^2_{n_3} = \sigma^2 :=\int_0^1\int_0^1 (x\wedge y - xy)d\Lambda(x)d\Lambda(y)>0.    
    \label{eq:conv_sigma}
\end{equation}
Combined with \eqref{eq:max_an}, this implies \eqref{eq:cns_hecker}. Thus, by Theorem 1 of \cite{hecker1976} and \eqref{eq:conv_sigma} again,
$$-\sqrt{n_3}J_{2}\convD \mathcal{N}(0,\varsigma^2).$$
\hfill$\square$

\section{Lemmas for Theorem 2}\label{sec:S2}

\begin{lemma}\label{lem:21}
    For all $k\in\{1,2,3,4,5\}$: $J_k=o(1)$ as $N\to\infty$.
\end{lemma}
{\bf Proof:} First, we show that $J_1=o(1)$. By sub-additivity of $u\mapsto \vert u\vert^\beta$ on $\R$ (since $\beta\in(0,1]$) and using that $\vert \widetilde F_Y^{(1)}(y)-\widehat F_Y^{(1)}(y)\vert\leq h_{n_2,\widehat F_Y^{(1)}(y)}$, we have
\begin{align}
   J_1&\leq 2c_UE\left[\int_{\R^2}\!\left(\left\vert \widehat F_Y^{(1)}(y)-F_Y(y)\right\vert^\beta \right)\left(\Omega \circ \widehat F_Y^{(1)}\otimes \widehat F_Y^{(2)}\right)\!(y,y') \, dydy'\right] \notag \\[5pt]
   & \quad +2c_UE\left[\int_{\R^2}\left(\Big\vert h_{n_2,\widehat F_Y^{(1)}\left(y\right)}\Big\vert^\beta\right)\left(\Omega \circ \widehat F_Y^{(1)}\otimes \widehat F_Y^{(2)}\right)\!(y,y') \, dydy'\right] \notag\\[5pt]
   & = 2c_UE\left[\int_{\R^2}\!\left(\left\vert \widehat F_Y^{(1)}(y)-F_Y(y)\right\vert^\beta \right)\left(\Omega \circ \widehat F_Y^{(1)}\otimes \widehat F_Y^{(2)}\right)\!(y,y') \, dydy'\right] \notag \\[5pt]
   & \quad +O(\varepsilon_{n_2})E\left[\int_{\R^2}\left(\Omega \circ \widehat F_Y^{(1)}\otimes \widehat F_Y^{(2)}\right) \right]. \label{eq:J1N_bias_var1}
\end{align}
The second term tends to zero since $\varepsilon_{n_2} \to 0$ and $E\left[\int_{\R^2}\left(\Omega \circ \widehat F_Y^{(1)}\otimes \widehat F_Y^{(2)}\right) \right] < \infty$ by Lemma~\ref{lem:int_FY_FY'}. 
We now focus on the first term in~\eqref{eq:J1N_bias_var1}.  Letting $\delta \in (0,1)$ be such that $\delta \leq 1/2 - (b_1 + d_1) \lor (b_2 + d_2)$ (using Assumption~\ref{ass:smoothness1}\ref{ass:smoothness14}), we have
\begin{align*}
    &\hspace{6mm} E\int_{\R^2}\!\left(\left\vert \widehat F_Y^{(1)}(y)-F_Y(y)\right\vert^\beta \right)\left(\Omega \circ \widehat F_Y^{(1)}\otimes \widehat F_Y^{(2)}\right)\!(y,y') \, dydy'\\[5pt]
    &= E\int_{\R^2}\!\left(\left\vert \widehat F_Y^{(1)}(y)-F_Y(y)\right\vert^\beta \right)\left(\widehat F_Y^{(1)}(y) \land \widehat F_Y^{(2)}(y')\right)^{1-2b_1}\left(\bar{\widehat{F}}_Y^{(1)}(y) \land \bar{\widehat{F}}_Y^{(2)}(y')\right)^{1-2b_2} \, dydy'\\[5pt]
    & \leq E\left[\int_{\R}\!\left\vert \widehat F_Y^{(1)}(y)-F_Y(y)\right\vert^\beta \Big(\widehat F_Y^{(1)}(y)\Big)^{d_1} \Big(\bar{\widehat{F}}_Y^{(1)}(y)\Big)^{d_2} \mathds 1_{\big\{\widehat F^{(1)}_Y(y) \bar{\widehat F}^{(1)
    }_Y(y) \neq 0\big\}}dy\right] \\
    & \quad \times E\left[ \int_{\R} \Big(\widehat F_Y^{(2)}(y)\Big)^{d_1 + \delta} \Big(\bar{\widehat{F}}_Y^{(2)}(y)\Big)^{d_2 + \delta} \, dy\right]\\[5pt]
    & \leq  \int_{\R}\!E^{1/2}\left[\left\vert \widehat F_Y^{(1)}(y)-F_Y(y)\right\vert^{2\beta} \mathds 1_{\big\{\widehat F^{(1)}_Y(y) \bar{\widehat F}^{(1)
    }_Y(y) \neq 0\big\}}\right] E^{1/2}\left[\Big(\widehat F_Y^{(1)}(y)\Big)^{2d_1} \Big(\bar{\widehat{F}}_Y^{(1)}(y)\Big)^{2d_2}  \right]dy \times C \quad  \text{by Lemma~\ref{lem:int_FY_FY'}}\\
    & \leq 2^{\beta/2}C \int_{\R}\! \left(\frac{F_Y(y) \bar F_Y(y)}{n_1}\right)^{\beta/2} E^{1/2}\left[\Big(\widehat F_Y^{(1)}(y)\Big)^{2d_1} \Big(\bar{\widehat{F}}_Y^{(1)}(y)\Big)^{2d_2}  \right]dy \quad \text{ by Lemma~\ref{lem_control_|F-hat_F|}}\\
    & \leq 2^{\beta/2}C \int_{\R}\! \left(\frac{F_Y(y) \bar F_Y(y)}{n_1}\right)^{\beta/2} \Big( F_Y(y)\Big)^{d_1} \Big(\bar{F}_Y(y)\Big)^{d_2}  dy \quad \text{ by concavity of $x \mapsto x^{2d_1} (1-x)^{2d_2}$}\\
    & = o(1) \qquad \text{ by Lemma~\ref{lem:int_FY}}.
    \end{align*}
    Note that we used Lemma~\ref{lem_control_|F-hat_F|} with $\tilde b_1 = (1 - \beta)/2$.
    This yields $J_1 \to 0$.

    \medskip
    Second, we show that $J_2=o(1)$. In the analysis of this term, we will write
\begin{align*}
    \alpha_1 &= 1-2b_1\\
    \alpha_2 &= 1-2b_2\\
    m(s,t) &= s \land t, \qquad \forall \alpha, s,t \in [0,1].
\end{align*}
It follows that $\Omega (s,t) = m(s,t)^{\alpha_1}\,m(\bar s, \bar t)^{\alpha_2}$. Then, by Assumption~\ref{ass:smoothness1}\ref{ass:smoothness23}, there exists a positive constant $C_U$ such that $g\leq C_U^2$. We have
\begin{align*}
    J_2 &= E\int_{\R^2} \left(g \circ  F_Y^{\,\otimes \,2}\right) \left|\left( \Omega \circ \widehat F_Y^{(1)}\otimes\widehat F_Y^{(2)}\right) - \left( \Omega \circ F_Y^{\,\otimes \,2}\right)\right|\\[10pt]
    & \lesssim E\int_{\R^2} \left|\left( \Omega \circ \widehat F_Y^{(1)}\otimes\widehat F_Y^{(2)}\right) - \left( \Omega \circ F_Y^{\,\otimes \,2}\right)\right|\hspace{2cm} \\[10pt]
    & = E \int_{\R^2} \left|\left(m\circ \widehat F_Y^{(1)}\otimes\widehat F_Y^{(2)}\right)^{\alpha_1} \left(m\circ \bar{\widehat F}_Y^{(1)}\otimes\bar{\widehat F}_Y^{(2)}\right)^{\alpha_2} - \left(m\circ F_Y^{\,\otimes \,2}\right)^{{\alpha_1}} \left(m\circ \bar{F}_Y^{\,\otimes \,2}\right)^{\alpha_2}\right|\\[10pt]
    & \leq E \int_{\R^2} \Big|\left(m\circ \widehat F_Y^{(1)}\otimes\widehat F_Y^{(2)}\right)^{\alpha_1} - \left(m\circ F_Y^{\,\otimes \,2}\right)^{\alpha_1} \Big|\left(m\circ \bar{\widehat F}_Y^{(1)}\otimes\bar{\widehat F}_Y^{(2)}\right)^{\alpha_2} \\[5pt]
    &
    \quad + E \int_{\R^2}\left(m\circ F_Y^{\,\otimes \,2}\right)^{\alpha_1} \left|\left(m\circ \bar{\widehat F}_Y^{(1)}\otimes\bar{\widehat F}_Y^{(2)}\right)^{\alpha_2} - \left(m\circ \bar{F}_Y^{\,\otimes \,2}\right)^{\alpha_2}\right| \\[10pt]
    & \leq E \int_{\R^2} \Big|\left(m\circ \widehat F_Y^{(1)}\otimes\widehat F_Y^{(2)}\right)^{\alpha_1}  - \left(m\circ F_Y^{\,\otimes \,2}\right)^{\alpha_1} \Big|\left(m\circ \bar{F}_Y^{\,\otimes \,2}\right)^{\alpha_2} 
    \\[5pt]
    & \quad + E \int_{\R^2}\left|\left(m\circ \bar{\widehat F}_Y^{(1)}\otimes\bar{\widehat F}_Y^{(2)}\right)^{\alpha_2} - \left(m\circ \bar{F}_Y^{\,\otimes \,2}\right)^{\alpha_2}\right| \left(m\circ F_Y^{\,\otimes \,2}\right)^{\alpha_1} \\[5pt]
    & \quad  + E \int_{\R^2}\Big|\left(m\circ \widehat F_Y^{(1)}\otimes\widehat F_Y^{(2)}\right)^{\alpha_1}  - \left(m\circ F_Y^{\,\otimes \,2}\right)^{\alpha_1} \Big| \, \Big|\left(m_{\alpha_2}\circ \bar{\widehat F}_Y^{(1)}\otimes\bar{\widehat F}_Y^{(2)}\right)^{\alpha_2} - \left(m\circ \bar{F}_Y^{\,\otimes \,2}\right)^{\alpha_2} \Big|\\[10pt]
    & =: A + A' + B.
\end{align*} 
To show that $J_2 = o(1)$, it suffices to show that $A, A', B = o(1)$. 
Since the terms $A$ and $A'$ can be analyzed in a similar fashion, we only show that $A, B = o(1)$. 

\medskip
\noindent \textit{Term $A$:} We have
$$\begin{aligned}
    A = \int_{\R^2} E \left(\Big|\big(\widehat F_Y^{(1)}(y) \land \widehat F_Y^{(2)}(y')\big)^{\alpha_1} - \big(F_Y(y) \land F_Y(y')\big)^{\alpha_1}\Big|\right) \left( \bar F_Y(y) \land \bar F_Y(y')\right)^{\alpha_2} dy dy'.
\end{aligned}$$
By Lemma~\ref{lem_assumptions_DCT_J2N}, the term $E \left(\big|\big(\widehat F_Y^{(1)}(y) \land \widehat F_Y^{(2)}(y')\big)^{\alpha_1} - \big(F_Y(y) \land F_Y(y')\big)^{\alpha_1}\big|\right)$ goes to $0$ as $N \to \infty$ for any fixed $y,y' \in \R$, and is dominated by the quantity $6 \big(F_Y(y) \land F_Y(y')\big)^{\alpha_1}$ independently of $N$. 
Moreover,
\begin{align*}
    & \quad \int_{\R^2} \big(F_Y(y) \land F_Y(y')\big)^{\alpha_1}\left( \bar F_Y(y) \land \bar F_Y(y')\right)^{\alpha_2} dydy'\\
    & \leq \int_{\R^2} \big(F_Y(y)  F_Y(y')\big)^{\alpha_1/2}\big(\bar F_Y(y)  \bar F_Y(y')\big)^{\alpha_2/2}dydy'\\
    & = \left(\int_{\R^2} F_Y(y) \, ^{\alpha_1/2}\bar F_Y(y)^{\alpha_2/2} \, dy \right)^2.
\end{align*}
The latter integral converges since $\alpha_1/2 > d_1$ and $\alpha_2/2 >d_2$ by Assumption~\ref{ass:smoothness1}\ref{ass:smoothness14}. By the dominated convergence theorem, we obtain that $A \to 0$ as $N \to \infty$.

\medskip
\noindent\textit{Term $B_N$:}
By the Cauchy--Schwarz inequality, we have
\begin{align}
    B &= E \int_{\R^2}\Big|\left(\widehat F_Y^{(1)}(y) \land \widehat F_Y^{(2)}(y')\right)^{\alpha_1}  - \left(F_Y(y) \land F_Y(y')\right)^{\alpha_1} \Big| dydy' \nonumber \\
    &\hspace{10.7mm} \times \Big|\left(\bar{\widehat F}_Y^{(1)}(y) \land \bar{\widehat F}_Y^{(2)}(y')\right)^{\alpha_2} - \left(\bar{F}_Y(y) \land \bar{F}_Y(y')\right)^{\alpha_2}  \Big|dydy' \nonumber \\[5pt]
    & \leq \int_{\R^2}E^{1/2} \left[\Big|\left(\widehat F_Y^{(1)}(y) \land \widehat F_Y^{(2)}(y')\right)^{\alpha_1}  - \left(F_Y(y) \land F_Y(y')\right)^{\alpha_1} \Big|^2\right] \nonumber\\
    &\hspace{10.7mm} \times E^{1/2}\left[\Big|\left(\bar{\widehat F}_Y^{(1)}(y) \land \bar{\widehat F}_Y^{(2)}(y')\right)^{\alpha_2} - \left(\bar{F}_Y(y) \land \bar{F}_Y(y')\right)^{\alpha_2}  \Big|^2 \right]dydy'. \label{eq_B_N}
\end{align}
Now, by Lemma~\ref{lem_decomposition_Omega-hat_Omega} and Jensens's inequality, we have
$$\begin{aligned}
    &\quad E^{1/2} \left[\Big|\left(\widehat F_Y^{(1)}(y) \land \widehat F_Y^{(2)}(y')\right)^{\alpha_1}  - \left(F_Y(y) \land F_Y(y')\right)^{\alpha_1} \Big|^2\right] \\
    & \leq E^{\alpha_1/2} \left[\Big|\widehat F_Y^{(1)}(y)  - F_Y(y) \Big|^2 \land F_Y(y')^{2}\right] + E^{\alpha_1/2} \left[\Big|\widehat F_Y^{(2)}(y')  - F_Y(y')\Big|^2 \land F_Y(y)^{2}\right] \\
    & \quad + E^{\alpha_1/2} \left[\Big|\widehat F_Y^{(1)}(y)  - F_Y(y) \Big|^2 \land \Big|\widehat F_Y^{(2)}(y')-F_Y(y')\Big|^2\right]\\
    &\leq \left(\frac{2F_Y(y) \bar F_Y(y)}{n_1}\right)^{\alpha_1/2} \land F_Y(y')^{\alpha_1} + \left(\frac{2F_Y(y') \bar F_Y(y')}{n_1}\right)^{\alpha_1/2} \land F_Y(y)^{\alpha_1} \\
    & \quad + \left(\frac{2F_Y(y) \bar F_Y(y)}{n_1}\right)^{\alpha_1/2} \land \left(\frac{2F_Y(y') \bar F_Y(y')}{n_1}\right)^{\alpha_1/2},
\end{aligned}$$
where, in the last inequality, we used that, for any two random variables $U,V$, it holds that $E(U\land V) \leq E(U) \land E(V)$. The second factor can be controlled analogously. Therefore, by~\eqref{eq_B_N}, we have
$$\begin{aligned}
    & \quad B \lesssim \int_{\R^2} \Bigg\{\left(\frac{F_Y(y) \bar F_Y(y)}{n_1}\right)^{\alpha_1/2} \land F_Y(y')^{\alpha_1} + \left(\frac{F_Y(y') \bar F_Y(y')}{n_1}\right)^{\alpha_1/2} \land F_Y(y)^{\alpha_1} \\
    & \hspace{1.5cm} + \left(\frac{F_Y(y) \bar F_Y(y)}{n_1}\right)^{\alpha_1/2} \land \left(\frac{F_Y(y') \bar F_Y(y')}{n_1}\right)^{\alpha_1/2}\Bigg\}\\
    & \qquad \times \Bigg\{\left(\frac{F_Y(y) \bar F_Y(y)}{n_1}\right)^{\alpha_2/2} \land \bar F_Y(y')^{\alpha_2} + \left(\frac{F_Y(y') \bar F_Y(y')}{n_1}\right)^{\alpha_2/2} \land \bar F_Y(y)^{\alpha_2} \\
    & \hspace{1.5cm} + \left(\frac{F_Y(y) \bar F_Y(y)}{n_1}\right)^{\alpha_2/2} \land \left(\frac{F_Y(y') \bar F_Y(y')}{n_1}\right)^{\alpha_2/2}\Bigg\}dydy'\\
    &\quad\quad=: 6\int_{\R^2} \bigg\{M_1(y,y') + M_2(y,y') + M_3(y,y')\bigg\} \bigg\{N_1(y,y') + N_2(y,y') +  N_3(y,y')\bigg\} dydy'.
\end{aligned}$$
We now prove that $\displaystyle\int_{\R^2}\! M_i(y,y') N_j(y,y') \, dydy' = o(1)$ for any $i,j \in \{1,2,3\}$. 
Using that $x\land y\leq(xy)^{1/2}$ for all $x,y\in[0,1]$, we have
\begin{align*}
    &\int_{\R^2} M_1(y,y') N_1(y,y') \, dy\,dy' \\
    & = \int_{\R^2} M_2(y,y') N_2(y,y') \, dy\,dy' \\
    & = \int_{\R^2} \left[\left(\frac{F_Y(y) \bar F_Y(y)}{n_1}\right)^{\alpha_1/2} \land F_Y(y')^{\alpha_1} \right] \times \left[\left(\frac{F_Y(y) \bar F_Y(y)}{n_1}\right)^{\alpha_2/2} \land \bar F_Y(y')^{\alpha_2}\right] dydy'\\
    & \leq \int_{\R^2} \left(\frac{F_Y(y) \bar F_Y(y)}{n_1}\right)^{(\alpha_1+\alpha_2)/2}  F_Y(y')^{\alpha_1/2}  \bar F_Y(y')^{\alpha_2/2} dydy'\\
& = o(1),
\end{align*}
where the convergence follows from Assumption~\ref{ass:smoothness2}\ref{ass:smoothness24} and Lemma~\ref{lem:int_FY}. Next, 
\begin{align*}
    \int_{\R^2} M_1(y,y') N_3(y,y') dydy' &= \int_{\R^2} M_2(y,y') N_3(y,y') dydy'\\
    &  \leq \int_{\R^2}\left(\frac{F_Y(y')\bar F_Y(y')}{n_1}\right)^{\alpha_1/2} \left(\frac{F_Y(y) \bar F_Y(y)}{n_1}\right)^{\alpha_2/2}dydy' \\
    & = o(1),
\end{align*}
where, again, the convergence follows from Assumption~\ref{ass:smoothness2}\ref{ass:smoothness24} and Lemma~\ref{lem:int_FY}. The terms $\displaystyle\int_{\R^2} M_2(y,y') N_1(y,y') \, dy\,dy' = \int_{\R^2} M_1(y,y') N_2(y,y') \, dy\,dy'$, $\displaystyle\int_{\R^2} M_3(y,y') N_1(y,y') \, dy\,dy' = \int_{\R^2} M_3(y,y') N_2(y,y') \, dy\,dy'$ and $\displaystyle\int_{\R^2} M_3(y,y') N_3(y,y') \, dy\,dy'$ can be controlled analogously, which concludes the proof of the claim $B = o(1)$. 

\medskip
Third, we show that $J_3=o(1)$. This follows from $\left\vert \widetilde F_Y^{(1)}(y)-F_Y(y)\right\vert^\beta\leq 1$, $g\leq C_U^2$, and Lemma~\ref{lem:int_FY_FY'}. 

\medskip
Fourth, we show that $J_4=o(1)$. Notice that 
\[
 \sup_{y\in\R}\left(h_{\widehat F_Y^{(1)}(y)}^{\beta}+h_{\widehat{\widehat F}_Y^{(1)}(y)}^{\beta} + h_{\widehat F_Y^{(2)}(y')}^{\beta}+h_{\widehat{\widehat F}_Y^{(2)}(y')}^{\beta}\right)\leq 4\eps_{n_2}=o(1).
 \]
 Hence,
 \begin{align*}
     J_4 & \leq 64c_U^2\eps_{n_2}^2 E \left[\left(\int_{\R^2} \,\Omega\left(\widetilde F_Y^{(1)}(y), \widetilde F_Y^{(2)}(y')\right) dydy'\right)^2\right] \\
     &\quad + 4E \left[\left(\int_{\R^2} \left(\Big| \widehat F_Y^{(1)}(y)-\widehat {\widehat F}_Y^{(1)}(y)\Big|^\beta+\Big| \widehat F_Y^{(2)}(y')-\widehat{\widehat F}_Y^{(2)}(y')\Big|^\beta\right)\,\Omega\left(\widetilde F_Y^{(1)}(y), \widetilde F_Y^{(2)}(y')\right) dydy'\right)^2\right] \\
     & = o(1) + 4E \left[\left(\int_{\R^2} \left(\Big| \widehat F_Y^{(1)}(y)-\widehat {\widehat F}_Y^{(1)}(y)\Big|^\beta+\Big| \widehat F_Y^{(2)}(y')-\widehat{\widehat F}_Y^{(2)}(y')\Big|^\beta\right)\,\Omega\left(\widetilde F_Y^{(1)}(y), \widetilde F_Y^{(2)}(y')\right) dydy'\right)^2\right],
 \end{align*}
 where the last equality follows from
 \[
 E \left[\left(\int_{\R^2} \,\Omega\left(\widetilde F_Y^{(1)}(y), \widetilde F_Y^{(2)}(y')\right) dydy'\right)^2\right] = \left(1+O(\eps_{n_2})\right)^2E \left[\left(\int_{\R^2} \,\Omega\left(\widehat F_Y^{(1)}(y), \widehat F_Y^{(2)}(y')\right) dydy'\right)^2\right]
 \]
 which is finite by Lemma~\ref{lem:int_FY_FY'} since $1-2b_j > 2 d_j$ for $j \in \{1,2\}$ by Assumption~\ref{ass:smoothness1}\ref{ass:smoothness14}. The remaining term can be shown to converge to zero by adding and subtracting $F_Y(y)$ and $F_Y(y')$ inside each absolute value, respectively, by using the sub-additivity of $x\mapsto\vert x\vert^\beta$, and by following a reasoning similar to Step 1.A above for showing that $J_1=o(1)$ as $N\to\infty$.

 \medskip
 Fifth, that $J_5=o(1)$ follows by adding and subtracting $\Omega(F_Y(y),F_Y(y'))$ inside the absolute value and by following a reasoning similar to the reasoning above for showing that $J_2=o(1)$ as $N\to\infty$.
\hfill $\square$

\begin{lemma}\label{lem:22}
    For $k\in\{6,7,8\}$: $E[\abs{J_k}]=o(1)$ as $N\to\infty$.
\end{lemma}
{\bf Proof:}
First, we show that $E[\abs{J_6}]=o(1)$. Using symmetry, we have the decomposition
\begin{align*}
  &  E[\abs{J_6}] \\
  \leq & E\left[\int_{[0,1]^4} {\rm cov}\big(\widehat f_U^{\,(1)}(s), \widehat f_U^{\,(1)}(s')\big) {\rm cov}\big(\widehat f_U^{\,(1)}(t), \widehat f_U^{\,(1)}(t')\big) d\widehat\mu_{\mathbf Y}(s,t,s',t')\right] \\[5pt]
  &\quad + E\left[\left(\int_{[0,1]^2} {\rm cov}\big(\widehat f_U^{\,(1)}(s), \widehat f_U^{\,(1)}(t)\big) d\widehat\mu_{\mathbf Y}(s,t)\right)^2\right] \\[5pt]
  & \quad + 2E\left[\int_{[0,1]^4} E\big(\widehat f_U^{\,(1)}(s)\big)E\big(\widehat f_U^{\,(1)}(s')\big) {\rm cov}\big(\widehat f_U^{\,(1)}(t), \widehat f_U^{\,(1)}(t')\big) d\widehat\mu_{\mathbf Y}(s,t,s',t')\right]\\[5pt]
  &\quad + 2E\left[\int_{[0,1]^4}E\big(\widehat f_U^{\,(1)}(s)\big)E\big(\widehat f_U^{\,(1)}(t)\big) {\rm cov}\big(\widehat f_U^{\,(1)}(s'), \widehat f_U^{\,(1)}(t')\big) d\widehat\mu_{\mathbf Y}(s,t,s',t')\right] \\[5pt]
  =& o(1),
\end{align*}
where the convergence follows by Lemma~\ref{lem_cov_hat_fU}.

\medskip
Second, we show that $E\left[\abs{J_7}+\abs{J_8}\right]=o(1)$. By the Cauchy--Schwarz inequality, we have 
\begin{align*}
   &\qquad  E\left[\abs{J_7}+\abs{J_8}\right] \\[5pt]
   & \quad \leq E\left[\left(\int_{[0,1]^2} \text{cov}\big(\widehat f_U^{\,(1)}(s), \widehat f_U^{\,(1)}(t)\big) d\widehat\mu_{\mathbf Y}(s,t)\right)^2\right]\\[5pt]
    & \quad + E^{1/2}\left[\left(\int_{[0,1]^2} \text{cov}\big(\widehat f_U^{\,(1)}(s), \widehat f_U^{\,(1)}(t)\big) d\widehat\mu_{\mathbf Y}(s,t)\right)^2\right]
    E^{1/2} \left[\left(\int_{[0,1]^2} e_{\mathbf Y}(s) e_{\mathbf Y}(t) d\widehat\mu_{\mathbf Y}(s,t)\right)^2\right].
\end{align*}
By Lemma~\ref{lem_cov_hat_fU}\ref{lem_cov_hat_fU1} below, we have $E\left[\left(\int_{[0,1]^2} {\rm cov}\big(\widehat f_U^{\,(1)}(s), \widehat f_U^{\,(1)}(t)\big) d\widehat\mu_{\mathbf Y}(s,t)\right)^2\right] = o(1)$. To show that $E\left[\abs{J_7}+\abs{J_8}\right]=o(1)$,  
it now remains to prove that $E \left[\left(\int_{[0,1]^2} e_{\mathbf Y}(s) e_{\mathbf Y}(t) d\widehat\mu_{\mathbf Y}(s,t)\right)^2\right] = O(1)$. 
We recall the definition of $\widetilde f_U^{(1)}$ from~\eqref{eq_def_hat_f_U} and will slightly abuse notation by writing $h_s$ instead of $h_{n_2,s}$.
We have, for any $s \in [0,1]$,
$$
\begin{aligned}
    e_{\mathbf Y}(s) &= E\left[\widetilde f_U^{(1)}(s)\,\big|\,\mathbf Y\right] = \frac{1}{2 h_s} \Pr\left(|U - s|\leq h_s\right) \\[5pt]
    & = \frac{F_U(s+h_s) - F_U(s-h_s)}{2h_s}\\[5pt]
    & = f_U(\tilde s)\\
    & \leq \frac{\widetilde C_U}{s^{\,b_1} \bar s^{\,b_2}},
\end{aligned}
$$
for some absolute constant $\widetilde C_U>0$ where we used the mean value theorem for some $\tilde s \in [s-h_s, s+h_s]$ in the third line, and Lemma~\ref{lem_fU_tilde_s} in the last line. We therefore obtain
\begin{align*}
   & E \left[\left(\int_{[0,1]^2} e_{\mathbf Y}(s) \, e_{\mathbf Y}(t) \, d\widehat\mu_{\mathbf Y}(s,t)\right)^2\right] \\[5pt]
   & \leq E \left[\left(\int_{[0,1]^2} \frac{\widetilde C_U^2}{(st)^{b_1} (\bar s \bar t)^{b_2}} (s\land t) (\bar s \land \bar t) \, d (\widehat F_Y^{(1)})^{-1}(s) \, d (\widehat F_Y^{(2)})^{-1}(t)\right)^2 \right]\\[5pt]
    & \leq E \left[\left(\int_{[0,1]^2} \frac{\widetilde C_U^2}{(s\land t)^{2b_1} (\bar s \land \bar t)^{2b_2}} (s\land t) (\bar s \land \bar t) \, d (\widehat F_Y^{(1)})^{-1}(s) \, d (\widehat F_Y^{(2)})^{-1}(t)\right)^2 \right]\\[5pt]
    & = \widetilde C_U^4 E \left[\left(\int_{[0,1]^2} (s\land t)^{1-2b_1} (\bar s \land \bar t)^{1-2b_2} \, d (\widehat F_Y^{(1)})^{-1}(s) \, d (\widehat F_Y^{(2)})^{-1}(t)\right)^2 \right],
\end{align*}
which is finite by Lemma~\ref{lem:int_FY_FY'} since $1-2b_j > 2 d_j$ for $j \in \{1,2\}$ by Assumption~\ref{ass:smoothness1}\ref{ass:smoothness14}.
\hfill$\square$

\begin{lemma}\label{lem_cov_hat_fU}
    It holds that
    \begin{enumerate}[label=(\roman*)]
        \item \label{lem_cov_hat_fU1} $E\left[\left(\int_{[0,1]^2} {\rm cov}\big(\widehat f_U^{\,(1)}(s), \widehat f_U^{\,(1)}(t)\big) d\widehat\mu_{\mathbf Y}(s,t)\right)^2\right] = o(1)$.
        \item \label{lem_cov_hat_fU2} $E\left[\int_{[0,1]^4} {\rm cov}\big(\widehat f_U^{\,(1)}(s), \widehat f_U^{\,(1)}(s')\big) {\rm cov}\big(\widehat f_U^{\,(1)}(t), \widehat f_U^{\,(1)}(t')\big) d\widehat\mu_{\mathbf Y}(s,t,s',t')\right] = o(1)$.
        \item \label{lem_cov_hat_fU3} $E\left[\int_{[0,1]^4} E\big(\widehat f_U^{\,(1)}(s)\big)E\big(\widehat f_U^{\,(1)}(s')\big) {\rm cov}\big(\widehat f_U^{\,(1)}(t), \widehat f_U^{\,(1)}(t')\big) d\widehat\mu_{\mathbf Y}(s,t,s',t')\right] = o(1)$.
        \item \label{lem_cov_hat_fU4} $E\left[\int_{[0,1]^4}E\big(\widehat f_U^{\,(1)}(s)\big)E\big(\widehat f_U^{\,(1)}(t)\big) {\rm cov}\big(\widehat f_U^{\,(1)}(s'), \widehat f_U^{\,(1)}(t')\big) d\widehat\mu_{\mathbf Y}(s,t,s',t')\right] = o(1)$.
    \end{enumerate}
    
\end{lemma}
\textbf{Proof:}
For ease of notation, we will write $n$ and $ \widehat F^{(j)}$ instead of $n_2$ and $\widehat F_Y^{(j)}$ respectively throughout the proof.  
We have, for any $s,t \in (0,1)$, 
$$
\begin{aligned}
    & \text{cov}\big(\widehat f_U^{\,(1)}(s), \widehat f_U^{\,(1)}(t)\big)\\[5pt]
    &\quad = E\left[\frac{1}{n^2 h_s h_t}\left(\sum_{i=1}^{n/2} \mathds{1}_{|U_i - s|\leq h_s} \right) \left(\sum_{j=1}^{n/2} \mathds{1}_{|U_j - t|\leq h_t}\right)\right] 
    - \frac{1}{n^2 h_s h_t}E\left[\sum_{i=1}^{n/2} \mathds{1}_{|U_i - s|\leq h_s}  \right]E\left[\sum_{j=1}^{n/2} \mathds{1}_{|U_j - t|\leq h_t}\right] \\[5pt]
    & \quad = \frac{1}{n^2 h_s h_t} E\left[\sum_{i \neq j} \mathds{1}_{|U_i - s|\leq h_s} \mathds{1}_{|U_j - t|\leq h_t} + \sum_{i=1}^{n}\mathds{1}_{|U_i - s|\leq h_s} \mathds{1}_{|U_i - t|\leq h_t} \right]\\[5pt]
    & \quad \quad - \frac{1}{h_s h_t} \Pr\big(|U-s| \leq h_s\big) \Pr\big(|U-t| \leq h_t\big)\\[5pt]
    & \quad \leq \frac{1}{h_s h_t} \Pr\big(|U-s| \leq h_s\big) \Pr\big(|U-t| \leq h_t\big)  + \frac{1}{n h_s h_t} \Pr\big(|U-s|\leq h_s \text{ and } |U-t| \leq h_t\big)\\[5pt]
    & \quad \quad - \frac{1}{h_s h_t} \Pr\big(|U-s| \leq h_s\big)\Pr\big(|U-t| \leq h_t\big)\\[5pt]
    & \quad = \frac{1}{n h_s h_t}\Pr\big(|U-s|\leq h_s \text{ and } |U-t| \leq h_t\big). 
\end{aligned}
$$
Moreover, by Lemma~\ref{lem_control_s_bar_s_t_bar_t} below, we have that $|s-t| \leq 7h_s$ and 
$h_t \geq h_s/6$ whenever $|s-t| \leq h_s + h_t$. 
Therefore, 
$$
\begin{aligned}
    \frac{1}{n h_s h_t}\Pr\big(|U-s|\leq h_s \text{ and } |U-t| \leq h_t\big) &\leq \frac{1}{n h_s h_t} \Pr\big(|U-s|\leq h_s \text{ and } |s-t| \leq h_s+h_t\big) \\[5pt]
    & = \frac{1}{n h_s h_t} \Pr\big(|U-s|\leq h_s \big) \mathbf 
    1\left\{|s-t| \leq h_s+h_t\right\}\\[5pt]
    & \leq \frac{1}{n h_s \cdot (h_s/6)} \cdot 2h_s f_U(\tilde s) \mathds 1\{|s-t| \leq 7h_s\}\\[5pt]
    & \leq \frac{12}{n h_s} \frac{\widetilde C_U}{s^{b_1} \bar s^{\, b_2}}\mathds 1\{|s-t| \leq 7h_s\}\\[5pt]
    & =: \frac{C}{n h_s} \frac{1}{s^{b_1} \bar s^{\, b_2}}\mathds 1\{|s-t| \leq 7h_s\}.
\end{aligned}$$
Therefore,
$$
\begin{aligned}
    \text{cov}\big(\widehat f_U^{\,(1)}(s), \, \widehat f_U^{\,(1)}(t)\big) \leq \frac{C}{n h_s} \frac{1}{s^{b_1} \bar s^{\, b_2}}\mathds 1\{|s-t| \leq 7h_s\}. 
\end{aligned}$$
%

Whenever $|s-t| \leq 7h_s$, we have $t \leq s+ 7h_s \leq 5 s$ since $h_s \leq \frac{1}{2} s\bar s$, and similarly, $\bar t \leq 5 \bar s$. 
Letting $n_1'=n_1/2$, we have
\begin{align*}
 & E\left[\left(\int_{[0,1]^2} \text{cov}\big(\widehat f_U^{\,(1)}(s), \widehat f_U^{\,(1)}(t)\big) d\widehat\mu_{\mathbf Y}(s,t)\right)^2\right] \\[5pt]
 \leq & E\left[\left(\int_{[0,1]^2} \frac{C}{n h_s} \frac{1}{s^{b_1} \bar s^{\, b_2}}\mathds 1\{|s-t| \leq 7h_s\} 
 d\widehat\mu_{\mathbf Y}(s,t)\right)^2\right]\\
    = &E\left[\left(\int_{[0,1]^2} \frac{C}{n h_s} \frac{1}{s^{b_1} \bar s^{\, b_2}}\mathds 1\big\{|s-t| \leq 7h_s\big\}
    (s \land t) (\bar s \land \bar t) \, d(\widehat F^{(1)})^{-1}(s) \, d(\widehat F^{(2)})^{-1}(t)\right)^2\right]\\[5pt]
   \leq &E\left[\left(\int_{[0,1]^2} \frac{C}{n h_s} \frac{1}{s^{b_1} \bar s^{\, b_2}} \mathds 1 \big\{|s-t| \leq 7h_s \big\} \cdot 5 s \cdot 5\bar s \cdot \, d(\widehat F^{(1)})^{-1}(s) \, d(\widehat F^{(2)})^{-1}(t)\right)^2\right]\\[5pt]
    =& (25C)^2 E\left[\left(\frac{1}{n}\int_{[0,1]} s^{1 - b_1} \bar s^{1 - b_2} \left(\frac{1}{h_s}\int_{[0,1]}\mathds 1 \big\{|s-t| \leq 7h_s \big\} d(\widehat F^{(2)})^{-1}(t) \right) d(\widehat F^{(1)})^{-1}(s)\right)^2  \right] \\[5pt]
   \lesssim& E\left[\left(\frac{1}{n}\int_{[0,1]} s^{1 - b_1} \bar s^{1 - b_2} \frac{(\widehat F^{(2)})^{-1} \left(s + 7h_s\right) - (\widehat F^{(2)})^{-1} \left(s - 7 h_s\right)}{h_s} d(\widehat F^{(1)})^{-1}(s)\right)^2  \right] \\[5pt]
    =& E\left[\left(\frac{1}{n}\int_{[0,1]} s^{1 - b_1} \bar s^{1 - b_2} \frac{Y_{(\lceil n_1's + 7n_1'h_{s}\rceil)}^{(2)} - Y_{(\lceil n_1's - 7n_1'h_{s}\rceil)}^{(2)}}{h_s} d(\widehat F^{(1)})^{-1}(s)\right)^2  \right] \\[5pt]
    =&E\left[\left(\frac{1}{n}\sum_{k=1}^{n_1'-1}\left(\frac{k}{n_1'}\right)^{1 - b_1} \left(\frac{n_1'-k}{n_1'}\right)^{1 - b_2} \frac{Y_{(\lceil k + 7n_1'h_{k/n_1'}\rceil)}^{(2)} - Y_{(\lceil k - 7n_1'h_{k/n_1'}\rceil)}^{(2)}}{h_{k/n_1'}}  \left(Y_{(k+1)}^{(1)} - Y_{(k)}^{(1)}\right)\right)^2 \right] \\[5pt]
    =& \frac{1}{n^2\varepsilon_n^2} \sum_{k,\ell = 1}^{n_1'-1} \left( \frac{k\ell}{(n_1')^2}\right)^{-b_1} \left( \frac{(n_1'-k)(n_1'-\ell)}{(n_1')^2}\right)^{-b_2}  E \Big[\left(Y_{(\lceil k + 7n_1'h_{k/n_1'}\rceil)}^{(2)} - Y_{(\lceil k - 7n_1'h_{k/n_1'}\rceil)}^{(2)}\right) \left(Y_{(k+1)}^{(1)} - Y_{(k)}^{(1)}\right)\\[5pt]
    & \hspace{7.2cm} \times
    \left(Y_{(\lceil \ell + 7n_1'h_{\ell/n_1'}\rceil)}^{(2)} - Y_{(\lceil \ell - 7n_1'h_{\ell/n_1'}\rceil)}^{(2)}\right) \left(Y_{(\ell+1)}^{(1)} - Y_{(\ell)}^{(1)}\right)\Big]\\[5pt]
    \leq & \frac{1}{n^2\varepsilon_n^2} \sum_{k,\ell = 1}^{n_1'-1} \left( \frac{k\ell}{(n_1')^2}\right)^{-b_1} \left(\frac{(n_1'-k)(n_1'-\ell)}{(n_1')^2}\right)^{-b_2}  E \Big[\left(Y_{(\lceil k + 7n_1'h_{k/n_1'}\rceil)}^{(2)} - Y_{(\lceil k - 7n_1'h_{k/n_1'}\rceil)}^{(2)}\right) ^2\\[5pt]
    & \hspace{7.2cm} \times 
    \left(Y_{(\lceil \ell + 7n_1'h_{\ell/n_1'}\rceil)}^{(1)} - Y_{(\lceil \ell - 7n_1'h_{\ell/n_1'}\rceil)}^{(1)}\right)^2 \Big]\\[5pt]
    =& \left(\frac{1}{n\varepsilon_n} \sum_{k=1}^{n_1'-1} \left( \frac{k}{n_1'}\right)^{-b_1} \left( \frac{(n_1'-k)}{n_1'}\right)^{-b_2}  E \Bigr[\left(Y_{(\lceil k + 7n_1'h_{k/n_1'}\rceil)}^{(1)} - Y_{(\lceil k - 7n_1'h_{k/n_1'}\rceil)}^{(1)}\right)^2 \Bigr]\right)^2\\[5pt]
    \lesssim&\bigg(\frac{1}{n\varepsilon_n} \sum_{k=1}^{n_1'-1} \left( \frac{k}{n_1'}\right)^{-b_1} \left( \frac{n_1'-k}{n_1'}\right)^{-b_2}  \left(\frac{h_{k/n_1'} + \frac{1}{n_1'}}{\left(\frac{k}{n_1'}\right)^{1+d_1}\left(\frac{n_1'-k}{n_1'}\right)^{1+d_2}}\right)^2\bigg)^2,
    \end{align*}
where we used the independence between $(Y_i^{(1)})_i$ and $(Y_j^{(2)})_j$ to obtain the second to last line, and Lemma~\ref{lem_product_spacings} below in the last inequality.

Next, since $n_1\asymp n_2$, we have
\begin{align*}
    &\frac{1}{n\varepsilon_n} \sum_{k = 1}^{n_1'-1} \left( \frac{k}{n_1'}\right)^{-b_1} \left( \frac{n_1'-k}{n_1'}\right)^{-b_2}  \left(\frac{h_{k/n_1'} + \frac{1}{n_1'}}{\left(\frac{k}{n_1'}\right)^{1+d_1}\left(\frac{n_1'-k}{n_1'}\right)^{1+d_2}}\right)^2\\[5pt]
    \leq & \frac{2}{n\varepsilon_n} \sum_{k = 1}^{n_1'-1} \left( \frac{k}{n_1'}\right)^{-b_1-2(1+d_1)} \left( \frac{n_1'-k}{n_1'}\right)^{-b_2-2(1+d_2)} \left(h_{\frac{k}{n_1'}}^2 + \frac{1}{(n_1')^2}\right) \\[5pt]
    \asymp & \frac{2\varepsilon_n}{n} \sum_{k = 1}^{n_1'-1} \left( \frac{k}{n_1'}\right)^{-b_1-2d_1} \left( \frac{n_1'-k}{n_1'}\right)^{-b_2-2d_2}  + \frac{2}{n^3\varepsilon_n} \sum_{k = 1}^{n_1'-1} \left( \frac{k}{n_1'}\right)^{-b_1-2(1+d_1)} \left( \frac{n_1'-k}{n_1'}\right)^{-b_2-2(1+d_2)} \\[5pt]
    \lesssim & \frac{2\varepsilon_n}{n} \sum_{k = 1}^{n_1'-1} \left[\left( \frac{k}{n_1'}\right)^{-b_1-2d_1} + \left( \frac{n_1'-k}{n_1'}\right)^{-b_2-2d_2} \right] + \frac{2}{n^3\varepsilon_n} \sum_{k = 1}^{n_1'-1} \left[\left( \frac{k}{n_1'}\right)^{-b_1-2(1+d_1)} + \left( \frac{n_1'-k}{n_1'}\right)^{-b_2-2(1+d_2)} \right].
\end{align*}
In the last step, we used the relation $x^{u} (1-x)^v \lesssim x^u + (1-x)^v$ for $u,v \in \R$ and any $x \in (0,1)$. 
Fix $j \in \{1,2\}$. Under Assumption~\ref{ass:smoothness2}\ref{ass:smoothness24}, we have
\begin{align*}
    \frac{2\varepsilon_n}{n} \sum_{k=1}^{n_1'-1} \left(\frac{k}{n_1'}\right)^{-b_j-2d_j} \lesssim  \varepsilon_n = o(1),
\end{align*}
and 
\begin{align*}
    \frac{2}{n^3\varepsilon_n} \sum_{k=1}^{n_1'-1} \left( \frac{k}{n_1'}\right)^{-b_j-2(1+d_j)} \lesssim \frac{n^{b_j+2(1+d_j)}}{n^3\eps_n} =o(1), 
\end{align*}
which implies that 
\begin{align*}
   \frac{1}{n\varepsilon_n} \sum_{k = 1}^{n_1'-1} \left( \frac{k}{n_1'}\right)^{-b_1} \left( \frac{n_1'-k}{n_1'}\right)^{-b_2}  \left(\frac{h_{k/n_1'} + \frac{1}{n_1'}}{\left(\frac{k}{n_1'}\right)^{1+d_1}\left(\frac{n_1'-k}{n_1'}\right)^{1+d_2}}\right)^2 = o(1).
\end{align*}
This completes the proof of Point 1.

As for Point 2, an analogous reasoning yields
\begin{align*}
& E\left[\int_{[0,1]^4} {\rm cov}\big(\widehat f_U^{\,(1)}(s), \widehat f_U^{\,(1)}(s')\big) {\rm cov}\big(\widehat f_U^{\,(1)}(t), \widehat f_U^{\,(1)}(t')\big) d\widehat\mu_{\mathbf Y}(s,t,s',t')\right] \\[5pt]
\leq& E\left[\int_{[0,1]^4} \frac{C^2}{n^2 h_s} \frac{s\bar s}{s^{b_1} \bar s^{\, b_2}}\mathds 1\{|s-s'| \leq 7h_s\}\frac{1}{h_{t'}} \frac{t'\bar{t'}}{{t'}^{b_1} \bar{t'}^{\, b_2}}\mathds 1\{|t'-t| \leq 7h_{t'}\} 
 d\widehat\mu_{\mathbf Y}(s,t,s',t')\right]\\[5pt]
 =& C^2 E\Bigg[\left(\frac{1}{n}\int_{[0,1]} s^{1 - b_1} \bar s^{1 - b_2} \frac{Y_{(\lceil n_1's + 7n_1'h_{s}\rceil)}^{(1)} - Y_{(\lceil n_1's - 7n_1'h_{s}\rceil)}^{(1)}}{h_s} d(\widehat F^{(1)})^{-1}(s)\right)\\[5pt]
 &\qquad \times \left(\frac{1}{n}\int_{[0,1]}{t'}^{1-b_1}\bar{t'}^{1-b_2}\frac{Y_{(\lceil n_1't' + 7n_1'h_{t'}\rceil)}^{(2)} - Y_{(\lceil n_1't' - 7n_1'h_{t'}\rceil)}^{(2)}}{h_{t'}} d(\widehat F^{(2)})^{-1}(t')\right)  \Bigg] \\
 = & C^2 E\left[\frac{1}{n}\int_{[0,1]} s^{1 - b_1} \bar s^{1 - b_2} \frac{Y_{(\lceil n_1's + 7n_1'h_{s}\rceil)}^{(1)} - Y_{(\lceil n_1's - 7n_1'h_{s}\rceil)}^{(1)}}{h_s} d(\widehat F^{(1)})^{-1}(s)\right]\\[5pt]
 &\qquad \times  C^2 E\left[\frac{1}{n}\int_{[0,1]}{t'}^{1-b_1}\bar{t'}^{1-b_2}\frac{Y_{(\lceil n_1't' + 7n_1'h_{t'}\rceil)}^{(2)} - Y_{(\lceil n_1't' - 7n_1'h_{t'}\rceil)}^{(2)}}{h_{t'}} d(\widehat F^{(2)})^{-1}(t')\right],
\end{align*}
where the last equality follows from independence between the samples $(Y_i^{(1)})_i$ and $(Y_i^{(2)})_i$. That each expectation in the upper bound is $o(1)$ follows from the Cauchy--Schwarz inequality together with the same arguments used in the proof of Point 1. 

As for Point 3, we have
\begin{align*}
& E\left[\int_{[0,1]^4} E\big(\widehat f_U^{\,(1)}(s)\big)E\big(\widehat f_U^{\,(1)}(s')\big) {\rm cov}\big(\widehat f_U^{\,(1)}(t), \widehat f_U^{\,(1)}(t')\big) d\widehat\mu_{\mathbf Y}(s,t,s',t')\right] \\[5pt]
\leq & C\widetilde C_U^2E\Bigg[\left(\int_{[0,1]}s^{1-b_1}\bar s^{1-b_2}d(\widehat F^{(1)})^{-1}(s)\right)\left(\int_{[0,1]}{s'}^{1/2-b_1}\bar{s'}^{1/2-b_2}d(\widehat F^{(1)})^{-1}(s')\right) \\[5pt]
& \qquad \times \left(\frac{1}{n}\int_{[0,1]}{t'}^{1/2}\bar{t'}^{1/2}\frac{Y_{(\lceil n_1't' + 7n_1'h_{t'}\rceil)}^{(2)} - Y_{(\lceil n_1't' - 7n_1'h_{t'}\rceil)}^{(2)}}{h_{t'}} d(\widehat F^{(2)})^{-1}(t')\right)\Bigg]\\[5pt]
\leq &  C\widetilde C_U^2E^{1/2}\left[\left(\int_{[0,1]}s^{1-b_1}\bar s^{1-b_2}d(\widehat F^{(1)})^{-1}(s)\right)^2\right] E^{1/2}\left[\left(\int_{[0,1]}{s'}^{1/2-b_1}\bar{s'}^{1/2-b_2}d(\widehat F^{(1)})^{-1}(s')\right)^2\right] \\[5pt]
& \qquad \times E\left[\left(\frac{1}{n}\int_{[0,1]}{t'}^{1/2}\bar{t'}^{1/2}\frac{Y_{(\lceil n_1't' + 7n_1'h_{t'}\rceil)}^{(2)} - Y_{(\lceil n_1't' - 7n_1'h_{t'}\rceil)}^{(2)}}{h_{t'}} d(\widehat F^{(2)})^{-1}(t')\right)\right],
\end{align*}
where the last inequality follows from the Cauchy--Schwarz inequality and independence between the samples $(Y_i^{(1)})_i$ and $(Y_i^{(2)})_i$. That the  upper bound is $o(1)$ follows from the first two expectations being bounded (by Lemma~\ref{lem:int_FY_FY'}) and the last expectation being $o(1)$ by the Cauchy--Schwarz inequality and the same arguments used in the proof of Point 1.

Finally, as for Point 4, we have
\begin{align*}
    & E\left[\int_{[0,1]^4}E\big(\widehat f_U^{\,(1)}(s)\big)E\big(\widehat f_U^{\,(1)}(t)\big) {\rm cov}\big(\widehat f_U^{\,(1)}(s'), \widehat f_U^{\,(1)}(t')\big) d\widehat\mu_{\mathbf Y}(s,t,s',t')\right] \\[5pt]
    \leq & C\widetilde C_U^2E\Bigg[\left(\int_{[0,1]}s^{1/2-b_1}\bar s^{1/2-b_2}d(\widehat F^{(1)})^{-1}(s)\right)\left(\int_{[0,1]}t^{1/2-b_1}\bar t^{1/2-b_2}d(\widehat F^{(2)})^{-1}(t)\right) \\[5pt]
    &\qquad \times \left(\frac1n \int_{[0,1]}{s'}^{1-b_1}\bar{s'}^{1-b_2}\frac{Y_{(\lceil n_1's' + 7n_1'h_{s'}\rceil)}^{(2)} - Y_{(\lceil n_1'{s'} - 7n_1'h_{s'}\rceil)}^{(2)}}{h_{s'}} d(\widehat F^{(1)})^{-1}(s') \right)\Bigg]\\[5pt]
\leq &  C\widetilde C_U^2E^{1/2}\left[\left(\int_{[0,1]}s^{1/2-b_1}\bar s^{1/2-b_2}d(\widehat F^{(1)})^{-1}(s)\right)^2\right] E^{1/2}\left[\left(\int_{[0,1]}{t}^{1/2-b_1}\bar{t}^{1/2-b_2}d(\widehat F^{(2)})^{-1}(t)\right)^2\right] \\[5pt]
& \qquad \times E^{1/2}\left[\left(\frac{1}{n}\int_{[0,1]}{s'}^{1/2}\bar{s'}^{1/2}\frac{Y_{(\lceil n_1's' + 7n_1'h_{s'}\rceil)}^{(2)} - Y_{(\lceil n_1's' - 7n_1'h_{s'}\rceil)}^{(2)}}{h_{s'}} d(\widehat F^{(1)})^{-1}(s')\right)\right],
\end{align*}
where the last inequality follows from the Cauchy--Schwarz inequality and independence between the samples $(Y_i^{(1)})_i$ and $(Y_i^{(2)})_i$. That the  upper bound is $o(1)$ follows from the first two expectations being bounded (by Lemma~\ref{lem:int_FY_FY'}) and the last expectation being $o(1)$ by the Cauchy--Schwarz inequality and the same arguments used in the proof of Point 1.

\hfill$\square$

\begin{lemma}\label{lem_fU_tilde_s}
    There exists some absolute constant $\widetilde C_U>0$ such that, for all $s,\tilde s \in (0,1)$ such that $\tilde s \in [s-h_s, s+h_s]$,
    \begin{align*}
        f_U(\tilde s) \leq \frac{\widetilde C_U}{s^{\, b_1} \bar s^{\, b_2}}.
    \end{align*}
\end{lemma}
\textbf{Proof:} Without loss of generality, assume that $s \in (0,1/2]$ (the symmetric case $s \in [1/2,1)$ can be treated analogously). By Assumption~\ref{ass:smoothness1}\ref{ass:smoothness13}, we have
\begin{align*}
    f_U(\tilde s) &\leq \dfrac{C_U}{\widetilde s^{b_1} \bar{\widetilde s}^{b_2}}.
\end{align*}
If $s \in [1/4,1/2]$ then $\tilde s \in [1/8,3/4]$ since we have $h_s \leq \frac{1}{2} s \bar s$ by definition. 
We then obtain
\begin{align*}
    f_U(\tilde s) &\leq C \leq \dfrac{C'}{s^{b_1} \bar{s}^{b_2}}
\end{align*}
for some absolute constant $C'>0$. 
Otherwise, $s \in (0,1/4]$ so that $\tilde s \leq 1/2$ and $\tilde s \geq s/2$. 
We therefore obtain
\begin{align*}
    f_U(\tilde s) &\leq \dfrac{C_U}{\widetilde s^{b_1} \bar{\widetilde s}^{b_2}} \leq \frac{C_U}{(s /2)^{b_1}(1/2)^{b_2}} \leq \frac{2^{b_1+b_2}C_U}{s^{\,b_1} \bar s^{\,b_2}} \leq \frac{2C_U}{s^{\,b_1} \bar s^{\,b_2}}.
\end{align*}
\hfill $\square$

\begin{lemma}\label{lem:25}
~
\begin{align}
      E \left[\displaystyle \int_{\R^2}\frac{w(\widehat F_Y^{(1)}(y),\widehat F_Y^{(2)}(y'))}{h_{\widehat F_Y^{(1)}(y)}h_{\widehat F_Y^{(2)}(y')}} I(y,y')dydy'\right]&=o(1), \label{eq:1} \\
      E \left[\displaystyle \int_{\R^2}\frac{w(\widehat F_Y^{(1)}(y),\widehat F_Y^{(2)}(y'))}{h_{\widehat F_Y^{(1)}(y)}h_{\widehat F_Y^{(2)}(y')}} II(y,y')dydy'\right] &= o(1), \label{eq:2} \\
      E\left[\displaystyle \int_{\R^2}\frac{w(\widehat F_Y^{(1)}(y),\widehat F_Y^{(2)}(y'))}{h_{\widehat F_Y^{(1)}(y)}h_{\widehat F_Y^{(2)}(y')}} III(y,y')dydy'\right] &= o(1). \label{eq:3}
  \end{align}  
\end{lemma}
{\bf Proof:} First, we show \eqref{eq:1}. By independence of $U_2$ from $(Z_1,\ldots,Z_{n_3},Y_1,\ldots,Y_{n_1},X_1)$, we have
\begin{align*}
    &\quad \int_{\R^2}\frac{w(\widehat F_Y^{(1)}(y),\widehat F_Y^{(2)}(y'))}{h_{\widehat F_Y^{(1)}(y)}h_{\widehat F_Y^{(2)}(y')}} I(y,y')dydy' \\
    &=\int_{\R^2}\frac{w(\widehat F_Y^{(1)}(y),\widehat F_Y^{(2)}(y'))}{h_{\widehat F_Y^{(1)}(y)}h_{\widehat F_Y^{(2)}(y')}} \\[5pt]
    &\qquad \times E_{\mathbf Y}\Big[\Big\vert \mathds 1\left\{\abs{\widehat U^{(1)}-\widehat F_Y^{(1)}(y)}\leq h_{\widehat F_Y^{(1)}(y)}\right\}-\mathds 1\left\{\abs{ U_1-\widehat F_Y^{(1)}(y)}\leq h_{\widehat F_Y^{(1)}(y)}\right\}\Big\vert \times \mathds 1\{\mathcal A_n^{(1)}\}\Big] \\[5pt]
    & \qquad \times{\Pr}_{\mathbf Y}\left(\abs{U_2-\widehat F_Y^{(2)}(y')}\leq h_{\widehat F_Y^{(2)}(y')}\right)dydy' \\[5pt]
     &=\int_{\R^2}\frac{w(\widehat F_Y^{(1)}(y),\widehat F_Y^{(2)}(y'))}{h_{\widehat F_Y^{(1)}(y)}} \\[5pt]
     &\qquad \times E_{\mathbf Y}\Big[\Big\vert \mathds 1\left\{\abs{\widehat U^{(1)}-\widehat F_Y^{(1)}(y)}\leq h_{\widehat F_Y^{(1)}(y)}\right\}-\mathds 1\left\{\abs{ U_1-\widehat F_Y^{(1)}(y)}\leq h_{\widehat F_Y^{(1)}(y)}\right\}\Big\vert\times \mathds 1\{\mathcal A_n^{(1)}\}\Big] \\[5pt]
     & \qquad \times f_U(\widetilde F_Y(y'))dydy',
\end{align*}
by the mean-value theorem for some $\widetilde F_Y(y')\in(\widehat F_Y^{(2)}(y')\pm h_{\widehat F_Y^{(2)}(y')})$. Fix $\delta>0$ such that $\delta<1-2[(b_1+d_1)\lor (b_2+d_2)]$ which exists by Assumption~\ref{ass:smoothness1}\ref{ass:smoothness14}. We have for all $y,y'\in\R$
\[
w(\widehat F_Y^{(1)}(y),\widehat F_Y^{(2)}(y'))\leq\widehat F_Y^{(1)}(y)^{1-b_1-d_1-\delta}\widehat F_Y^{(2)}(y')^{b_1+d_1+\delta}\bar{\widehat F}_Y^{(1)}(y)^{1-b_2-d_2-\delta}\bar{\widehat F}_Y^{(2)}(y')^{b_2+d_2+\delta}.
\]
Hence,
\begin{align*}
    &\quad \int_{\R^2}\frac{w(\widehat F_Y^{(1)}(y),\widehat F_Y^{(2)}(y'))}{h_{\widehat F_Y^{(1)}(y)}h_{\widehat F_Y^{(2)}(y')}} I(y,y')dydy' \\[5pt]
     &\lesssim\int_{\R^2}\frac{\widehat F_Y^{(1)}(y)^{1-b_1-d_1-\delta}\widehat F_Y^{(2)}(y')^{b_1+d_1+\delta}\bar{\widehat F}_Y^{(1)}(y)^{1-b_2-d_2-\delta}\bar{\widehat F}_Y^{(2)}(y')^{b_2+d_2+\delta}}{h_{\widehat F_Y^{(1)}(y)}} \\[5pt]
     & \qquad \times E_{\mathbf Y}\Big[\Big\vert \mathds 1\left\{\abs{\widehat U_1-\widehat F_Y^{(1)}(y)}\leq h_{\widehat F_Y^{(1)}(y)}\right\}-\mathds 1\left\{\abs{ U_1-\widehat F_Y^{(1)}(y)}\leq h_{\widehat F_Y^{(1)}(y)}\right\}\Big\vert\times \mathds 1\{\mathcal A_n^{(1)}\}\Big]\\[5pt]
     & \qquad \times (\widetilde F_Y^{(2)}(y'))^{-b_1}(\bar{\widetilde F}_Y^{(2)}(y'))^{-b_2}dydy' \\[5pt]
     &\lesssim\frac{1}{\eps_n}\int_{\R^2}\widehat F_Y^{(1)}(y)^{-b_1-d_1-\delta}\widehat F_Y^{(2)}(y')^{d_1+\delta}\bar{\widehat F}_Y^{(1)}(y)^{-b_2-d_2-\delta}\bar{\widehat F}_Y^{(2)}(y')^{d_2+\delta} \\[5pt]
     & \qquad \times E_{\mathbf Y}\Big[\Big\vert \mathds 1\left\{\abs{\widehat U_1-\widehat F_Y^{(1)}(y)}\leq h_{\widehat F_Y^{(1)}(y)}\right\}-\mathds 1\left\{\abs{ U_1-\widehat F_Y^{(1)}(y)}\leq h_{\widehat F_Y^{(1)}(y)}\right\}\Big\vert\times\mathds 1\{\mathcal A_n^{(1)}\} \Big]dydy', 
\end{align*}
by Lemma~\ref{lem:useful_lem}. Next, we take the expectation with respect to $(Y_i)_{i=1}^{n_1}$ and split the integration domain $\R^2$ as follows
\begin{align}
    \R^2 = \big\{\mathcal{D}_1 \times \R \big\} \bigcup \big\{\mathcal{D}_2 \times \R \big\} \bigcup \big\{\mathcal{D}_3 \times \R \big\}\label{eq_domain_splitting}
\end{align}
where 
\begin{align*}
    & \mathcal{D}_1 = \Big(-\infty,(\widehat F_Y^{(1)})^{-1}\Big(\frac{2a_n}{\sqrt{n}}\Big)\Big],\\[5pt]
    & \mathcal{D}_2 = \Big((\widehat {F}_Y^{(1)})^{-1}\Big(\frac{2a_n}{\sqrt{n}}\Big),(\widehat F_Y^{(1)})^{-1}\Big(1-\frac{2a_n}{\sqrt{n}}\Big)\Big),\\[5pt]
    & \mathcal{D}_3 = \Big[(\widehat F_Y^{(1)})^{-1}\Big(1-\frac{2a_n}{\sqrt{n}}\Big),\infty\Big).
\end{align*}
Consider the second domain $\mathcal{D}_2$. By Lemma~\ref{lem:useful_lem_indicators},
\begin{align*}
    &\frac{1}{\eps_n}\int_{\mathcal{D}_2}\widehat F_Y^{(1)}(y)^{-b_1-d_1-\delta}\widehat F_Y^{(2)}(y')^{d_1+\delta}\bar{\widehat F}_Y^{(1)}(y)^{-b_2-d_2-\delta}\bar{\widehat F}_Y^{(2)}(y')^{d_2+\delta} \\[5pt]
    &\quad \times \left(\widehat F_Y^{(1)}(y)\right)^{-b_1} \left(\bar{\widehat F}_Y^{(1)}(y)\right)^{-b_2}\frac{a_n}{\sqrt{n}}dydy'\\[5pt]
    & = \frac{2a_n}{\eps_n\sqrt{n}}\int_{(\widehat F_Y^{(1)})^{-1}(\frac{2a_n}{\sqrt{n}})}^{(\widehat F_Y^{(1)})^{-1}(1-\frac{2a_n}{\sqrt{n}})}\widehat F_Y^{(1)}(y)^{-2b_1-d_1-\delta}\bar{\widehat F}_Y^{(1)}(y)^{-2b_2-d_2-\delta} dy\int_\R \widehat F_Y^{(2)}(y')^{d_1+\delta} \bar{\widehat F}_Y^{(2)}(y')^{d_2+\delta} dy'.
\end{align*}
Taking the expectation with respect to $(Y_i)_i$ and using the Cauchy--Schwarz inequality, the latter quantity is upper bounded as
\begin{align*}
    &\quad \frac{2a_n}{\eps_n\sqrt{n}} E^{1/2}\left[\left(\int_{(\widehat F_Y^{(1)})^{-1}(\frac{2a_n}{\sqrt{n}})}^{(\widehat F_Y^{(1)})^{-1}(1-\frac{2a_n}{\sqrt{n}})}\widehat F_Y^{(1)}(y)^{-2b_1-d_1-\delta}\bar{\widehat F}_Y^{(1)}(y)^{-2b_2-d_2-\delta} dy\right)^2\right] \\[5pt]
    &\qquad \times E^{1/2}\left[\left(\int_\R \widehat F_Y^{(2)}(y')^{d_1+\delta} \bar{\widehat F}_Y^{(2)}(y')^{d_2+\delta} dy'\right)^2\right]\\[5pt]
    & \lesssim \frac{a_n}{\eps_n \sqrt{n}} E^{1/2}\left[\left(\int_{(\widehat F_Y^{(1)})^{-1}(\frac{2a_n}{\sqrt{n}})}^{(\widehat F_Y^{(1)})^{-1}(1-\frac{2a_n}{\sqrt{n}})}\widehat F_Y^{(1)}(y)^{-2b_1-d_1-\delta}\bar{\widehat F}_Y^{(1)}(y)^{-2b_2-d_2-\delta} dy\right)^2\right]  \text{ by Lemma~\ref{lem:int_FY_FY'}}\\[5pt]
    & = o(1)
\end{align*}
by Lemma~\ref{lem:int_D2_vanishes} applied with $\gamma_1 = -2b_1 -d_1 - \delta\geq \frac{1}{2} - b_1$ and $\gamma_2 = -2b_2 -d_2 - \delta\geq \frac{1}{2} - b_2$.

We now consider the first domain $\mathcal{D}_1$ (the third one $\mathcal{D}_3$ can be analyzed in a similar way). 
We have
\begin{align*}
    &\frac{1}{\eps_n}\int_{-\infty}^{(\widehat F_Y^{(1)})^{-1}(\frac{2a_n}{\sqrt{n}})}\int_\R\widehat F_Y^{(1)}(y)^{-b_1-d_1-\delta}\widehat F_Y^{(2)}(y')^{d_1+\delta}\bar{\widehat F}_Y^{(1)}(y)^{-b_2-d_2-\delta}\bar{\widehat F}_Y^{(2)}(y')^{d_2+\delta}(\widehat F_Y^{(1)}(y))^{1-b_1}dydy'\\[5pt]
    &\lesssim \frac{1}{\eps_n}\int_{-\infty}^{(\widehat F_Y^{(1)})^{-1}(\frac{2a_n}{\sqrt{n}})}\widehat F_Y^{(1)}(y)^{1-2b_1-d_1-\delta} dy\int_\R \widehat F_Y^{(2)}(y')^{d_1+\delta} \bar{\widehat F}_Y^{(2)}(y')^{d_2+\delta} dy'.
\end{align*}
Moreover, by the Cauchy--Schwarz inequality, we have
\begin{align*}
    & \quad E\left[\frac{1}{\eps_n}\int_{-\infty}^{(\widehat F_Y^{(1)})^{-1}(\frac{2a_n}{\sqrt{n}})}\widehat F_Y^{(1)}(y)^{1-2b_1-d_1-\delta} dy\int_\R \widehat F_Y^{(2)}(y')^{d_1+\delta} \bar{\widehat F}_Y^{(2)}(y')^{d_2+\delta} dy'\right] \\[5pt]
    &\leq E^{1/2}\left[\left(\frac{1}{\eps_n}\int_{-\infty}^{(\widehat F_Y^{(1)})^{-1}(\frac{2a_n}{\sqrt{n}})}\widehat F_Y^{(1)}(y)^{1-2b_1-d_1-\delta} dy\right)^2 \right]E^{1/2}\left[\left(\int_\R \widehat F_Y^{(2)}(y')^{d_1+\delta} \bar{\widehat F}_Y^{(2)}(y')^{d_2+\delta} dy'\right)^2\right] \\[5pt]
    & = o(1)
\end{align*}
by Lemmas~\ref{lem:int_FY_FY'} and~\ref{lem:int_D1_vanishes}.

\medskip
Second, we show \eqref{eq:2}. This follows from the observation that $I(y,y')=II_N(y',y)$ and the proof is analogous to the steps above.

\medskip
Third, we show \eqref{eq:3}. By independence of $\widehat U^{(1)}$ and $\widehat U^{(2)}$, we have
\begin{align*}
    &\quad \int_{\R^2}\frac{w(\widehat F_Y^{(1)}(y),\widehat F_Y^{(2)}(y'))}{h_{\widehat F_Y^{(1)}(y)}h_{\widehat F_Y^{(2)}(y')}} III(y,y')dydy' \\[5pt]
    & = \int_{\R^2}\frac{w\big(\widehat F_Y^{(1)}(y),\widehat F_Y^{(2)}(y')\big)}{h_{\widehat F_Y^{(1)}(y)}h_{\widehat F_Y^{(2)}(y')}} E_{\mathbf Y}\bigg[\Big\vert \mathds 1\Big\{\big|\widehat U^{(1)}-\widehat F_Y^{(1)}(y)\,\big|\leq \, h_{\widehat F_Y^{(1)}(y)}\Big\}-\mathds 1\Big\{\big| U_1-\widehat F_Y^{(1)}(y)\big|\leq h_{\widehat F_Y^{(1)}(y)}\Big\}\Big\vert \\[5pt]
    &\qquad \times\Big\vert\mathds 1\left\{\abs{\widehat U^{(2)}-\widehat F_Y^{(2)}(y')}\leq h_{\widehat F_Y^{(2)}(y')}\right\}-\mathds 1\left\{\abs{ U_2-\widehat F_Y^{(2)}(y')}\leq h_{\widehat F_Y^{(2)}(y')}\right\}\Big\vert\times\mathds 1\{\mathcal A_n\}\bigg]dydy'\\[5pt]
    & \leq \int_{\R^2}\frac{(\widehat F_Y^{(1)}(y)\widehat F_Y^{(2)}(y'))^{1/2}(\bar{\widehat F}_Y^{(1)}(y),\bar{\widehat F}_Y^{(2)}(y'))^{1/2}}{h_{\widehat F_Y^{(1)}(y)}h_{\widehat F_Y^{(2)}(y')}} \\
    & \qquad \times E_{\mathbf Y}\Big[\Big\vert \mathds 1\left\{\abs{\widehat U^{(1)}-\widehat F_Y^{(1)}(y)\,}\leq \, h_{\widehat F_Y^{(1)}(y)}\right\}-\mathds 1\left\{\abs{ U_1-\widehat F_Y^{(1)}(y)}\leq h_{\widehat F_Y^{(1)}(y)}\right\}\Big\vert \times \mathds{1}\{\mathcal{A}_n^{(1)}\}\Big] \\[5pt]
    &\qquad \times E_{\mathbf Y}\Big[\Big\vert\mathds 1\left\{\abs{\widehat U^{(2)}-\widehat F_Y^{(2)}(y')}\leq h_{\widehat F_Y^{(2)}(y')}\right\}-\mathds 1\left\{\abs{ U_2-\widehat F_Y^{(2)}(y')}\leq h_{\widehat F_Y^{(2)}(y')}\right\}\Big\vert\times\mathds 1\{\mathcal A_n^{(2)}\}\Big]dydy'\\[5pt]
    & = \bigg(\int_{\R}\frac{(\widehat F_Y^{(1)}(y) \bar{\widehat F}_Y^{(1)}(y))^{1/2}}{h_{\widehat F_Y^{(1)}(y)}} E_{\mathbf Y}\Big[\Big\vert \mathds 1\left\{\abs{\widehat U^{(1)}-\widehat F_Y^{(1)}(y)\,}\leq \, h_{\widehat F_Y^{(1)}(y)}\right\} \\[5pt]
    & \hspace{5cm}  -\mathds 1\left\{\abs{ U_1-\widehat F_Y^{(1)}(y)}\leq h_{\widehat F_Y^{(1)}(y)}\right\}\Big\vert \times \mathds{1}\{\mathcal{A}_n^{(1)}\}\Big] dy\bigg)^2\\[5pt]
    & = \bigg(\int_{\R}\frac{(\widehat F_Y^{(1)}(y)\bar{\widehat F}_Y^{(1)}(y))^{-1/2}}{\eps_n} E_{\mathbf Y}\Big[\Big\vert \mathds 1\left\{\abs{\widehat U^{(1)}-\widehat F_Y^{(1)}(y)\,}\leq \, h_{\widehat F_Y^{(1)}(y)}\right\} \\
    & \hspace{5cm} -\mathds 1\left\{\abs{ U_1-\widehat F_Y^{(1)}(y)}\leq h_{\widehat F_Y^{(1)}(y)}\right\}\Big\vert  \times \mathds{1}\{\mathcal{A}_n^{(1)}\}\Big] dy\bigg)^2\\[5pt]
    & = \bigg(\int_{\mathcal{D}_1 \cup \mathcal{D}_2 \cup  \mathcal{D}_3}\frac{1}{\eps_n} (\widehat F_Y^{(1)}(y) \bar{\widehat F}_Y^{(1)}(y))^{-1/2} \\[5pt]
    & \hspace{2cm} \times E_{\mathbf Y}\Big[\Big\vert \mathds 1\left\{\abs{\widehat U^{(1)}-\widehat F_Y^{(1)}(y)\,}\leq \, h_{\widehat F_Y^{(1)}(y)}\right\}-\mathds 1\left\{\abs{ U_1-\widehat F_Y^{(1)}(y)}\leq h_{\widehat F_Y^{(1)}(y)}\right\}\Big\vert  \times \mathds{1}\{\mathcal{A}_n^{(1)}\}\Big] dy\bigg)^2
\end{align*}
where the domains $\mathcal{D}_i$ are defined in \eqref{eq_domain_splitting}. 
 We now show that the latter quantity tends to zero as $N \to \infty$ by successively placing ourselves on the domains $\mathcal{D}_1, \mathcal{D}_2$ and $\mathcal{D}_3$.  
Consider first the domain $\mathcal{D}_2$. By Lemma~\ref{lem:useful_lem_indicators} and the Cauchy-Schwarz inequality, we have
\begin{align*}
    &E\Big[\int_{\mathcal{D}_2}\frac{1}{\eps_n}\Big(\widehat F_Y^{(1)}(y)\bar{\widehat F}_Y^{(1)}(y)\Big)^{-1/2}\\[5pt]
    & \qquad \times E_{\mathbf Y}\Big[\Big\vert \mathds 1\left\{\abs{\widehat U^{(1)}-\widehat F_Y^{(1)}(y)\,}\leq \, h_{\widehat F_Y^{(1)}(y)}\right\}-\mathds 1\left\{\abs{ U_1-\widehat F_Y^{(1)}(y)}\leq h_{\widehat F_Y^{(1)}(y)}\right\}\Big\vert  \times \mathds{1}\{\mathcal{A}_n^{(1)}\}\Big] dy\Big]\\[5pt]
    & \lesssim  E^{1/2}\left[\left(\int_{\mathcal{D}_2}\frac{1}{\eps_n} (\widehat F_Y^{(1)}(y)\bar{\widehat F}_Y^{(1)}(y))^{-1/2}\dfrac{a_n}{\sqrt{n}}\left(\widehat F_Y^{(1)}(y)\right)^{-b_1} \left(\bar{\widehat F}_Y^{(1)}(y)\right)^{-b_2} dy\right)^2\right]\\[5pt]
    &= o(1)
\end{align*}
by Lemma~\ref{lem:int_D2_vanishes}. 
We now consider the domain $\mathcal{D}_1$ (the integral over the domain $\mathcal{D}_3$ can be controlled in a similar way). 
We have, by Lemma~\ref{lem:useful_lem_indicators}
\begin{align*}
    &E\Bigg[\int_{\mathcal{D}_1}\frac{1}{\eps_n} \Big(\widehat F_Y^{(1)}(y) \bar{\widehat F}_Y^{(1)}(y)\Big)^{-1/2}E_{\mathbf Y}\Big[\Big\vert \mathds 1\left\{\abs{\widehat U^{(1)}-\widehat F_Y^{(1)}(y)\,}\leq \, h_{\widehat F_Y^{(1)}(y)}\right\}\\[5pt]
    &\hspace{6cm} -\mathds 1\left\{\abs{ U_1-\widehat F_Y^{(1)}(y)}\leq h_{\widehat F_Y^{(1)}(y)}\right\}\Big\vert\Big] dy\Bigg]\\[5pt]
    & \lesssim E\left[\int_{\mathcal{D}_1}\frac{1}{\eps_n} \Big(\widehat F_Y^{(1)}(y)\bar{\widehat F}_Y^{(1)}(y)\Big)^{-1/2}\widehat F_Y^{(1)}(y)^{1-b_1} dy\right]\\[5pt]
    & \lesssim E\left[\int_{\mathcal{D}_1}\frac{1}{\eps_n} \widehat F_Y^{(1)}(y)^{\frac{1}{2}-b_1} dy\right]\\[5pt]
    & = o(1)
\end{align*}
by Lemma~\ref{lem:int_D1_vanishes}. This concludes the proof.

\begin{lemma}\label{lem:useful_lem_indicators}
    For any $y \in \R$, we have
\begin{align*}
    & \quad E_{\mathbf Y}\Big[\Big\vert \mathds 1\left\{\abs{\widehat U_1^{(1)}-\widehat F_Y^{(1)}(y)}\leq h_{\widehat F_Y^{(1)}(y)}\right\}-\mathds 1\left\{\abs{ U_1-\widehat F_Y^{(1)}(y)}\leq h_{\widehat F_Y^{(1)}(y)}\right\}\Big\vert\times\mathds 1\{\mathcal A_n^{(1)}\}\Big] \\[5pt]
    & \lesssim    \begin{cases}
        \dfrac{a_n}{\sqrt{n}}\left(\widehat F_Y^{(1)}(y)\right)^{-b_1} \left(\bar{\widehat F}_Y^{(1)}(y)\right)^{-b_2} & \text{ if } \widehat F_Y^{(1)}(y) \land \bar{\widehat F}_Y^{(1)}(y) \geq \frac{2a_n}{\sqrt{n}}\\[5pt]       
        \Big(\widehat F_Y^{(1)}(y)\Big)^{1-b_1} & \text{ if } \widehat F_Y^{(1)}(y) < \frac{2a_n}{\sqrt{n}}\\[10pt]
        \Big(\bar{\widehat F}_Y^{(1)}(y)\Big)^{1-b_2} & \text{ if } \bar{\widehat F}_Y^{(1)}(y) < \frac{2a_n}{\sqrt{n}}.
    \end{cases}
\end{align*}
\end{lemma}
{\bf Proof:}
We have
\begin{align*}
   &  E_{\mathbf Y}\Big[\Big\vert \mathds 1\left\{\abs{\widehat U_1^{(1)}-\widehat F_Y^{(1)}(y)}\leq h_{\widehat F_Y^{(1)}(y)}\right\}-\mathds 1\left\{\abs{ U_1-\widehat F_Y^{(1)}(y)}\leq h_{\widehat F_Y^{(1)}(y)}\right\}\Big\vert\times \mathds 1\{\mathcal A_n^{(1)}\}\Big]\\[5pt]
   &\quad = 
       {\Pr}_{\,\mathbf{Y}}\left(\abs{\widehat U_1^{(1)}-\widehat F_Y^{(1)}(y)}\leq h_{\widehat F_Y^{(1)}(y)} \;\text{ and } \; \abs{U_1-\widehat F_Y^{(1)}(y)}> h_{\widehat F_Y^{(1)}(y)} \;\text{ and } \; \mathcal A_n^{(1)}\right) \\[5pt]
       & \qquad + 
       {\Pr}_{\mathbf{Y}}\left(\abs{\widehat U_1^{(1)}-\widehat F_Y^{(1)}(y)}> h_{\widehat F_Y^{(1)}(y)} \;\text{ and } \; \abs{U_1-\widehat F_Y^{(1)}(y)}\leq h_{\widehat F_Y^{(1)}(y)} \;\text{ and } \; \mathcal A_n^{(1)}\right).
\end{align*}
Assume first that $\widehat F_Y^{(1)}(y) \land \bar{\hat F}_{Y}(y)\geq 2a_n/\sqrt{n}$.
We have
\begin{align*}
    &\quad {\Pr}_{\mathbf{Y}}\left(\abs{\widehat U_1^{(1)}-\widehat F_Y^{(1)}(y)}\leq h_{\widehat F_Y^{(1)}(y)} \;\text{ and } \; \abs{U_1-\widehat F_Y^{(1)}(y)}> h_{\widehat F_Y^{(1)}(y)} \;\text{ and } \; \mathcal A_n^{(1)}\right) \\[5pt]
    & \leq {\Pr}_{\mathbf{Y}}\left(\abs{U_1-\widehat F_Y^{(1)}(y)}\leq h_{\widehat F_Y^{(1)}(y)} + |\widehat U_1^{(1)} - U_1|\;, \; \abs{U_1-\widehat F_Y^{(1)}(y)}> h_{\widehat F_Y^{(1)}(y)} \;,\; |\widehat U_1^{(1)} - U_1| \leq \frac{a_n}{\sqrt{n}}\right) \\[5pt]
    & \leq {\Pr}_{\mathbf{Y}}\left(\abs{U_1-\widehat F_Y^{(1)}(y)} \in \Big[ h_{\widehat F_Y^{(1)}(y)}, h_{\widehat F_Y^{(1)}(y)} + \frac{a_n}{\sqrt{n}} \Big]\right) \\[5pt]
    & \leq {\Pr}_{\mathbf{Y}}\left(\abs{U_1-\widehat F_Y^{(1)}(y)} \in \Big[ h_{\widehat F_Y^{(1)}(y)} - \frac{a_n}{\sqrt{n}} , h_{\widehat F_Y^{(1)}(y)} + \frac{a_n}{\sqrt{n}} \Big]\right).
    \end{align*}
    Similarly, we have
    \begin{align*}
    &\quad {\Pr}_{\mathbf{Y}}\left(\abs{\widehat U_1^{(1)}-\widehat F_Y^{(1)}(y)}> h_{\widehat F_Y^{(1)}(y)} \;\text{ and } \; \abs{U_1-\widehat F_Y^{(1)}(y)} \leq h_{\widehat F_Y^{(1)}(y)} \;\text{ and } \; \mathcal A_n^{(1)}\right) \\
    & \leq {\Pr}_{\mathbf{Y}}\left(\abs{U_1-\widehat F_Y^{(1)}(y)}\leq h_{\widehat F_Y^{(1)}(y)} \;, \; \abs{U_1-\widehat F_Y^{(1)}(y)}> h_{\widehat F_Y^{(1)}(y)} - |\widehat U_1^{(1)} - U_1| \;, \; |\widehat U_1^{(1)} - U_1| \leq \frac{a_n}{\sqrt{n}}\right) \\[5pt]
    & \leq {\Pr}_{\mathbf{Y}}\left(\abs{U_1-\widehat F_Y^{(1)}(y)} \in \Big[ h_{\widehat F_Y^{(1)}(y)} -  \frac{a_n}{\sqrt{n}} , h_{\widehat F_Y^{(1)}(y)} \Big]\right) \\[5pt]
    & \leq {\Pr}_{\mathbf{Y}}\left(\abs{U_1-\widehat F_Y^{(1)}(y)} \in \Big[ h_{\widehat F_Y^{(1)}(y)} - \frac{a_n}{\sqrt{n}} , h_{\widehat F_Y^{(1)}(y)} + \frac{a_n}{\sqrt{n}} \Big]\right). 
    \end{align*}
    Now, we have
    \begin{align*}
    &{\Pr}_{\mathbf{Y}}\left(\abs{U_1-\widehat F_Y^{(1)}(y)} \in \Big[ h_{\widehat F_Y^{(1)}(y)} - \frac{a_n}{\sqrt{n}} , h_{\widehat F_Y^{(1)}(y)} + \frac{a_n}{\sqrt{n}} \Big]\right) \\[5pt]
    & = F_U\left(\big[\widehat F_Y^{(1)}(y) - h_{\widehat F_Y^{(1)}(y)} - \frac{a_n}{\sqrt{n}}, \widehat F_Y^{(1)}(y) - h_{\widehat F_Y^{(1)}(y) } + \frac{a_n}{\sqrt{n}}\big] \right) \\[5pt]
    & \qquad + F_U\left(\big[\widehat F_Y^{(1)}(y) + h_{\widehat F_Y^{(1)}(y)} - \frac{a_n}{\sqrt{n}}, \widehat F_Y^{(1)}(y) + h_{\widehat F_Y^{(1)}(y) } + \frac{a_n}{\sqrt{n}}\big]\right)\\[5pt]
    & = \Big[f_U\left(\widetilde F_Y^{(1)}(y)\right) + f_U\left(\widetilde F_Y^{(2)}(y)\right)\Big] \frac{2a_n}{\sqrt{n}},
\end{align*}
where we used the mean-value theorem for some
\begin{align*}
    &\widetilde F_Y^{(1)}(y) \in \big[\widehat F_Y^{(1)}(y) - h_{\widehat F_Y^{(1)}(y)} - \frac{a_n}{\sqrt{n}}, \widehat F_Y^{(1)}(y) - h_{\widehat F_Y^{(1)}(y)}+\frac{a_n}{\sqrt{n}}\big] \subseteq
    \left[\frac{1}{6} \widehat F_{Y}(y), 1 - \frac{1}{6} \bar{\widehat F}_Y^{(1)}(y)\right]\\[5pt]
    & \widetilde F_Y^{(2)}(y) \in \big[\widehat F_Y^{(1)}(y) + h_{\widehat F_Y^{(1)}(y)} , \widehat F_Y^{(1)}(y) + h_{\widehat F_Y^{(1)}(y) } + \frac{a_n}{\sqrt{n}}\big] \subseteq \left[{\widehat F}_Y(y), 1 - \frac{1}{6} \bar{\widehat F}_Y^{(1)}(y)\right],
\end{align*}
 for $N$ sufficiently large. By Assumption~\ref{ass:smoothness1}\ref{ass:smoothness13}, the latter quantity is further upper-bounded as
\begin{align*}
    & \quad 2\Big[f_U\left(\widetilde F_Y^{(1)}(y)\right) + f_U\left(\widetilde F_Y^{(2)}(y)\right)\Big] \frac{a_n}{\sqrt{n}} \\[5pt]
    & \leq 2C_U \left\{\left(\widetilde F_Y^{(1)}(y)\right)^{-b_1} \left(1 -\widetilde F_Y^{(1)}(y)\right)^{-b_2} + \left(\widetilde F_Y^{(2)}(y)\right)^{-b_1} \left(1 -\widetilde F_Y^{(2)}(y)\right)^{-b_2} \right\} \frac{a_n}{\sqrt{n}}\\[5pt]
    & \leq 4 C_U \left(\frac{1}{6}\widehat F_Y^{(1)}(y)\right)^{-b_1} \left(\frac{1}{6}\bar{\widehat F}_Y^{(1)}(y)\right)^{-b_2}  \frac{a_n}{\sqrt{n}}\\[5pt]
    & \leq C \left(\widehat F_Y^{(1)}(y)\right)^{-b_1} \left(\bar{\widehat F}_Y^{(1)}(y)\right)^{-b_2}\frac{a_n}{\sqrt{n}}.
\end{align*}
This proves the desired result in the case where $\widehat F_Y^{(1)}(y) \land \bar{\widehat F}_Y^{(1)}(y) \geq \frac{2a_n}{\sqrt{n}}$. 

Assume now that $\widehat F_Y^{(1)}(y) \lor \bar{\widehat F}_Y^{(1)}(y) < 2 \frac{a_n}{\sqrt{n}}$. 
By symmetry, assume that $\widehat F_Y^{(1)}(y) < 2 a_n/\sqrt{n}$. 
We have
\begin{align}
    &\quad {\Pr}_{\mathbf{Y}}\left(\abs{\widehat U_1^{(1)}-\widehat F_Y^{(1)}(y)}> h_{\widehat F_Y^{(1)}(y)} \;\text{ and } \; \abs{U_1-\widehat F_Y^{(1)}(y)} \leq h_{\widehat F_Y^{(1)}(y)} \;\text{ and } \; \mathcal A_n^{(1)}\right)\nonumber\\[5pt]
    & \leq {\Pr}_{\mathbf{Y}}\left( \abs{U_1-\widehat F_Y^{(1)}(y)} \leq h_{\widehat F_Y^{(1)}(y)}\right)\nonumber\\[5pt]
    & \leq {\Pr}_{\mathbf{Y}}\left( U_1\leq 2\widehat F_Y^{(1)}(y)\right) \nonumber\\[5pt]
    & \leq \int_{0}^{2\widehat F_Y^{(1)}(y)} C_U t^{-b_1} (1-t)^{-b_2} dt\nonumber\\[5pt]
    & \lesssim \Big(\widehat F_Y^{(1)}(y)\Big)^{1-b_1}.\label{eq_first_proba_edges}
\end{align}
Moreover,
\begin{align}
    & \quad{\Pr}_{\mathbf{Y}}\left(\abs{\widehat U_1^{(1)}-\widehat F_Y^{(1)}(y)}\leq h_{\widehat F_Y^{(1)}(y)} \;\text{ and } \; \abs{U_1-\widehat F_Y^{(1)}(y)}> h_{\widehat F_Y^{(1)}(y)} \;\text{ and } \; \mathcal A_n^{(1)}\right) \nonumber\\[5pt]
    & \leq {\Pr}_{\mathbf{Y}}\left(\abs{\widehat U_1^{(1)}-\widehat F_Y^{(1)}(y)}\leq h_{\widehat F_Y^{(1)}(y)} \right) \nonumber\\[5pt]
    & \leq {\Pr}_{\mathbf{Y}}\left(\widehat U_1^{(1)} \leq 2\widehat F_Y^{(1)}(y)\right).\label{eq_second_proba_edges}
\end{align}
Defining $n_3' = n_3/2$, we have, for all $t\in[0,1]$, 
\begin{align*}
    {\Pr}_{\mathbf Y}(\widehat U_1^{(1)} \leq t)& = E_{\mathbf Y}\left[{\Pr}_{\mathbf Y}\left(\widehat F_Z^{(1)}(X_1)\leq t \, \big| \, X_1\right)\right]\\[5pt]
    &=E_{\mathbf Y}\Big[{\Pr}_{\mathbf Y}\big({\rm Binomial}(n_3',F_Z(X_1))\leq n_3't \, \big| \,X_1\big)\Big].
\end{align*}
We place ourselves conditionally on $X_1$. Suppose first that $t>F_Z(X_1)$. Then we have
\begin{align*}
    {\Pr}_{\mathbf Y}\big({\rm Binomial}(n_3',F_Z(X_1))\leq n_3' t \, \big| \,X_1\big)&\leq 1.
\end{align*}
Otherwise, by Bennett's inequality \citep[e.g., Theorem~2.9 in][]{boucheron2013concentration}, for all $t\leq F_Z(X_1)$,
\begin{align*}
    &{\Pr}_{\mathbf Y}\big({\rm Binomial}(n_3',F_Z(X_1))\leq n_3't \, \big| \,X_1\big)\leq \exp\left(-n_3'F_Z(X_1) \,h\!\left(\frac{n_3'(t-F_Z(X_1))}{n_3' F_Z(X_1)}\right)\right) \\[5pt]
    \text{i.e. } &{\Pr}_{\mathbf Y}\big({\rm Binomial}(n_3',U_1)\leq n_3't \, \big| \,X_1\big)= \exp\left(-n_3'U_1 \,h\!\left(\frac{t-U_1}{ U_1}\right)\right),
\end{align*}
where $h(x)=(1+x)\log(1+x)-x$. Taking expectation on both sides, Assumption~\ref{ass:smoothness1}\ref{ass:smoothness13} yields for $t<1/4$
\begin{align*}
    {\Pr}_{\mathbf{Y}}(\widehat U_1 \leq t) &\leq \int_0^{2t} f_U(u) du + \int_{2t}^1 \exp\left(-n_3' u \,h\!\left(\frac{t-u}{u}\right)\right) f_U(u) du\\[5pt]
    & \lesssim \int_{0}^{2t}  u^{-b_1}  du + \int_{2t}^{1/2} \exp\left(-n_3'u \,h\!\left(\frac{t-u}{u}\right)\right) u^{-b_1} du \\[5pt]
    & \hspace{2.5cm} +  \int_{1/2}^{1} \exp\left(-n_3'u \,h\!\left(\frac{t-u}{u}\right)\right) (1-u)^{-b_2} du\\[5pt]
    & \lesssim t^{1-b _1} + \int_{2t}^{1/2} \exp\left(-cn_3'u \right) u^{-b_1} du \\[5pt]
    & \hspace{1.5cm} +  \int_{1/2}^{1} \exp\left(-cn_3'u \right) (1-u)^{-b_2} du \\[5pt]
    &\lesssim t^{1-b_1} + n_3'^{b_1-1} \int_{2tcn_3'}^{\infty}\exp(-z)z^{-b_1}dz+ n_3'^{b_2-1}\exp(-cn_3')\int_{cn_3'/2}^\infty\exp(-z)z^{-b_2}dz\\[5pt]
    &\lesssim t^{1-b_1}+ {n_3'}^{b_1-1}\exp(-2tcn_3')+{n_3'}^{b_2-1}\exp\left(-\frac{3}{2}cn_3'\right).
\end{align*}
For all $t\geq1/n_3'$,
$${n_3'}^{b_1-1}\exp(-2tcn_3')\leq t^{1-b_1}.$$
Moreover, for $N$ larger than a constant depending on $b_1$ and $b_2$,
$$
  {n_3'}^{b_2-1}\exp\left(-\frac{3}{2}cn_3'\right)\leq t^{1-b_1}.  
$$
It follows that, for all $t\geq 1/n_3'$, 
\[
 {\Pr}_{\mathbf{Y}}(\widehat U_1 \leq t) \lesssim t^{1-b_1}.
\]
By taking $t=\widehat F_Y^{(1)}(y)$ and combining \eqref{eq_first_proba_edges} and~\eqref{eq_second_proba_edges}, the result follows.
\hfill$\square$

\begin{lemma}\label{lem:int_D2_vanishes}
    Let $\gamma_1 \geq -\frac{1}{2} - b_1$ and $\gamma_2 \geq -\frac{1}{2} - b_2$. Then it holds that
    \begin{align*}
        \frac{a_n}{\eps_n \sqrt{n}} E^{1/2}\left[\left(\int_{(\widehat F_Y^{(1)})^{-1}(\frac{2a_n}{\sqrt{n}})}^{(\widehat F_Y^{(1)})^{-1}(1-\frac{2a_n}{\sqrt{n}})}\widehat F_Y^{(1)}(y)^{\gamma_1}\bar{\widehat F}_Y^{(1)}(y)^{\gamma_2} dy\right)^2\right]=o(1).
    \end{align*}
\end{lemma}
{\bf Proof:}

We have

\begin{align*}
    & \qquad \frac{a_n}{\eps_n \sqrt{n}} E^{1/2}\left[\left(\int_{(\widehat F_Y^{(1)})^{-1}(\frac{2a_n}{\sqrt{n}})}^{(\widehat F_Y^{(1)})^{-1}(1-\frac{2a_n}{\sqrt{n}})}\widehat F_Y^{(1)}(y)^{\gamma_1}\bar{\widehat F}_Y^{(1)}(y)^{\gamma_2} dy\right)^2\right]  \\[5pt]
    & \lesssim \frac{a_n}{\eps_n \sqrt{n}} \Big\{E^{1/2}\left[\left(\int_{(\widehat F_Y^{(1)})^{-1}(\frac{2a_n}{\sqrt{n}})}^{(\widehat F_Y^{(1)})^{-1}(1/2)}\widehat F_Y^{(1)}(y)^{\gamma_1}\bar{\widehat F}_Y^{(1)}(y)^{\gamma_2} dy\right)^2\right] \\[5pt]
    & \qquad + E^{1/2}\left[\left(\int_{(\widehat F_Y^{(1)})^{-1}(1/2)}^{(\widehat F_Y^{(1)})^{-1}(1-\frac{2a_n}{\sqrt{n}})}\widehat F_Y^{(1)}(y)^{\gamma_1}\bar{\widehat F}_Y^{(1)}(y)^{\gamma_2} dy\right)^2\right]  \Big\}\\[5pt]
    & \lesssim \frac{a_n}{\eps_n \sqrt{n}} \Big\{E^{1/2}\left[\left(\int_{(\widehat F_Y^{(1)})^{-1}(\frac{2a_n}{\sqrt{n}})}^{(\widehat F_Y^{(1)})^{-1}(1/2)}\widehat F_Y^{(1)}(y)^{\gamma_1} dy\right)^2\right] + E^{1/2}\left[\left(\int_{(\widehat F_Y^{(1)})^{-1}(1/2)}^{(\widehat F_Y^{(1)})^{-1}(1-\frac{2a_n}{\sqrt{n}})}\bar{\widehat F}_Y^{(1)}(y)^{\gamma_2} dy\right)^2\right]  \Big\}\\[5pt]
    & \leq \frac{a_n}{\eps_n \sqrt{n}} \Big\{E^{1/2}\left[\left(\left(\frac{2a_n}{\sqrt{n}}\right)^{\gamma_1} \times \bigg((\widehat F_Y^{(1)})^{-1}\Big(\frac{1}{2}\Big) - (\widehat F_Y^{(1)})^{-1}\Big(\frac{2a_n}{\sqrt{n}}\Big)\bigg)\right)^2\right]\\[5pt]
    & \hspace{1.5cm}+ E^{1/2}\left[\left(\left(\frac{2a_n}{\sqrt{n}}\right)^{\gamma_2} \times \bigg((\widehat F_Y^{(1)})^{-1}\Big(1 - \frac{2a_n}{\sqrt{n}}\Big)\bigg) - (\widehat F_Y^{(1)})^{-1}\Big(\frac{1}{2}\Big)\right)^2\right]  \Big\}.
\end{align*}
We handle the first expectation (the second one can be analyzed in a similar way). 
We have
\begin{equation}\label{eq:middleterm}
    \frac{1}{\eps_n} \left(\frac{a_n}{\sqrt{n}}\right)^{1+ \gamma_1} E^{1/2}\left[\bigg((\widehat F_Y^{(1)})^{-1}\Big(\frac{1}{2}\Big) - (\widehat F_Y^{(1)})^{-1}\Big(\frac{2a_n}{\sqrt{n}}\Big)\bigg)^2\right].
\end{equation}
Letting $n_1' = n_1/2$, we recall that the density of $\xi_{(r)}$ is \citep[see, e.g., Equation~(2.1.6) in][]{david2004order}:
\[
f_{\xi_{(r)}}(x)=\frac{1}{{\rm B}(r,n_1'-r+1)}x^r\bar{x}^{n_1'-r}.
\]
Combining this with Assumption~\ref{ass:smoothness1}\ref{ass:smoothness12} and letting $r=\lceil 2a_n\sqrt{n}\rceil$, we have
\begin{align*}
    E\left[\Big((\widehat F_Y^{(1)})^{-1}\Big(\frac{2a_n}{\sqrt{n}}\Big)\Big)^2\right] & = E\left[( F_Y^{-1}(\xi_{(r)}))^2\right] \\[5pt]
    &\lesssim\int_0^1x^{-2d_1}\bar{x}^{-2d_2}\frac{1}{{\rm B}(r,n_1-r+1)}x^{r-1}\bar{x}^{n_1'-r}dx\\[5pt]
    &=\frac{{\rm B}(\lceil 2a_n\sqrt{n}\rceil-2d_1,n_1'-\lceil 2a_n\sqrt{n}\rceil-2d_2+1)}{{\rm B}(\lceil 2a_n\sqrt{n}\rceil,n_1'-\lceil 2a_n\sqrt{n}\rceil+1)}\\[5pt]
    &=\frac{\Gamma(\lceil 2a_n\sqrt{n}\rceil-2d_1)\Gamma(n_1'-\lceil 2a_n\sqrt{n}\rceil-2d_2+1)}{\Gamma(n_1'+1-2(d_1+d_2))} \\[5pt]
    &\quad \times \frac{\Gamma(n_1'+1)}{\Gamma(\lceil 2a_n\sqrt{n}\rceil)\Gamma(n_1'-\lceil 2a_n\sqrt{n}\rceil+1)} \\[5pt]
    &\lesssim (2a_n\sqrt{n})^{-2d_1}{n_1'}^{-2d_2}{n_1'}^{2(d_1+d_2)} \quad \text{ by Lemma~\ref{lem_comparison_Gamma}}\\[5pt]
    &\lesssim \left(\frac{a_n}{\sqrt{n}}\right)^{-2d_1}.
\end{align*}
By a similar reasoning,  there exists a constant $C>0$ independent of $N$ such that 
\[
E\left[\Big((\widehat F_Y^{(1)})^{-1}\Big(\frac{1}{2}\Big)\Big)^2\right]\leq C.
\]
Plugging this result into \eqref{eq:middleterm} yields
\begin{align*}
   & \quad \frac{1}{\eps_n} \left(\frac{a_n}{\sqrt{n}}\right)^{1+\gamma_1} E^{1/2}\left[\bigg((\widehat F_Y^{(1)})^{-1}\Big(\frac{1}{2}\Big) - (\widehat F_Y^{(1)})^{-1}\Big(\frac{2a_n}{\sqrt{n}}\Big)\bigg)^2\right]\\[5pt]
   & \lesssim  \frac{1}{\eps_n} \left(\frac{a_n}{\sqrt{n}}\right)^{1+\gamma_1} \left(\frac{a_n}{\sqrt{n}}\right)^{-d_1} \\[5pt]
   &= \frac{1}{\eps_n} \left(\frac{a_n}{\sqrt{n}}\right)^{1+\gamma_1-d_1} \\[5pt]
   &=o(1)
\end{align*}
since $1+\gamma_1 - d_1 \geq \frac{1}{2} - b_1 - d_1 >0$ and $a_n = O(\log(n))$.
Similarly,
\begin{equation*}
     \frac{1}{\eps_n} \left(\frac{a_n}{\sqrt{n}}\right)^{1+\gamma_2} E^{1/2}\left[\left((\widehat F_Y^{(1)})^{-1}\Big(1 - \frac{2a_n}{\sqrt{n}}\Big) - (\widehat F_Y^{(1)})^{-1}\Big(\frac{1}{2}\Big)\right)^2\right] =o(1).
\end{equation*}

\hfill$\square$

\begin{lemma}\label{lem:int_D1_vanishes}
    For any $\alpha \geq \frac{1}{2} - b_1$, it holds that 
    $$E\left[\left(\int_{-\infty}^{(\widehat F_Y^{(1)})^{-1}(\frac{2a_n}{\sqrt{n}})}\widehat F_Y^{(1)}(y)^{\alpha} dy\right)^2\right]=o(\eps_n^2).$$
\end{lemma}
{\bf Proof:} Let $n_1'=n_1/2$. We have
\begin{align*}
    & E\left[\left(\int_{-\infty}^{(\widehat F_Y^{(1)})^{-1}(\frac{2a_n}{\sqrt{n}})}\widehat F_Y^{(1)}(y)^{\alpha} dy\right)^2\right] \\[5pt]
    &=E\left[\left(\sum_{i=1}^{\lceil 2\sqrt{n}a_n\rceil-1}\left(\frac{i}{n_1'}\right)^{\alpha}\left(Y_{(i+1)}^{(1)}-Y_{(i)}^{(1)}\right)\right)^2\right] \\[5pt]
    &=\sum_{i,j=1}^{\lceil 2\sqrt{n}a_n\rceil-1}\left(\frac{ij}{(n_1')^2}\right)^{\alpha}E\left[\left(Y_{(i+1)}^{(1)}-Y_{(i)}^{(1)}\right)\left(Y_{(j+1)}^{(1)}-Y_{(j)}^{(1)}\right)\right]\\[5pt]
    &\leq \sum_{i,j=1}^{\lceil 2\sqrt{n}a_n\rceil-1}\left(\frac{ij}{(n_1')^2}\right)^{\alpha}E^{1/2}\left[\left(Y_{(i+1)}^{(1)}-Y_{(i)}^{(1)}\right)^2\right]E^{1/2}\left[\left(Y_{(j+1)}^{(1)}-Y_{(j)}^{(1)}\right)^2\right]\\[5pt]
    &= \left(\sum_{i=1}^{\lceil 2\sqrt{n}a_n\rceil-1}\left(\frac{i}{n_1'}\right)^{\alpha}E^{1/2}\left[\left(Y_{(i+1)}^{(1)}-Y_{(i)}^{(1)}\right)^2 \right]\right)^2 \\[5pt]
    &\lesssim \left(\frac{1}{n_1'}\sum_{i=1}^{\lceil 2\sqrt{n}a_n\rceil-1}\left(\frac{i}{n_1'}\right)^{\alpha-(1+d_1)}\right)^2 \\[5pt]
    &=\left(\frac{1}{(n_1')^{\alpha-d_1}}\sum_{i=1}^{\lceil 2\sqrt{n}a_n\rceil-1}i^{(\alpha-d_1)-1}\right)^2 \\[5pt]
    &\asymp \left(\frac{a_n\sqrt{n}}{n}\right)^{2(\alpha-d_1)}\\[5pt]
    &=o(\eps_n^2),
\end{align*}
where the first inequality follows from the Cauchy--Schwarz inequality, the asymptotic inequality follows from Lemma~\ref{lem:spacing_4_mom}, and the asymptotic equivalence from $\alpha-d_1>0$.

\hfill$\square$

\section{Other technical lemmas}

\begin{lemma}\label{lem_product_spacings}
Let $n_1' = n_1/2$. For all $k\in\{1,\ldots,n_1'-1\}$, we have
$$
\begin{aligned}
    E\left[\left(Y_{(\lceil k + 7n_1'h_{k/{n_1'}}\rceil)}^{(1)} - Y_{(\lceil k - 7n_1'h_{{k}/{n_1'}}\rceil)}^{(1)} \right) ^2\right] \lesssim \left(\frac{h_{k/n_1'} + \frac{1}{n_1'}}{\left(\frac{k}{n_1'}\right)^{1+d_1}\left(\frac{n_1'-k}{n_1'}\right)^{1+d_2}}\right)^2.
\end{aligned}
$$
\end{lemma}
\textbf{Proof:}  Let $a_k = \lceil k - 7n_1'h_{k/n_1'}\rceil$ and $b_k = \lceil k + 7n_1'h_{k/n_1'}\rceil$. We note that we always have $b_k \geq a_k +1$ for any $k \in \{1,\dots,n_1'-1\}$.  Suppose that $a_k \geq 10$ and $b_k \leq n_1'-10$. 
Then, we have $a_k\asymp b_k\asymp k$ and $b_k-a_k\lesssim n_1'h_{k/n_1'}$. The result follows from Lemma~\ref{lem:spacing_4_mom}.1.

Now, suppose that $a_k \leq 10$ or $n_1'-b_k \leq 10$. By symmetry, suppose $a_k \leq 10$. By the definition of $h_{k/n_1'}$, we have $7n_1'h_{k/n_1'} < k/12$ for $n_1'$ large enough, so that
\begin{align*}
    \lceil k - 7n_1'h_{k/n_1'}\rceil \leq 10 \implies \lceil 11k/12 \rceil \leq 10 \implies k \leq 10.
\end{align*}
In particular, we have $7n_1'h_{k/n_1'} < k/12 <1$, which implies $b_{k} = a_k + 1$. The result follows from Lemma~\ref{lem:spacing_4_mom}.2.
\hfill$\square$

\begin{lemma}\label{lem_control_|F-hat_F|}
    For any $y \in \R$ and $0\leq b_1,b_2\leq1/2$, we have for all $j\in\{1,2\}$:
    \begin{align*}
        &E^{1/2}\left[\big|\widehat F_Y^{(j)}(y) - F_Y(y)\big|^{2(1-2b_1)}\mathds 1_{\big\{\widehat F^{(j)}_Y(y) \bar{\widehat F}^{(j)
    }_Y(y) \neq 0\big\}}\right] \leq \frac{2^{1/2-b_1}\big(F_Y(y)\bar F_Y(y)\big)^{1/2}}{n_1^{1/2-2b_1}}\land \left(\frac{2F_Y(y)\bar F_Y(y)}{n_1}\right)^{1/2-b_1}\\
        & E^{1/2}\left[\big|\bar{\widehat F}^{(j)}_Y(y) - \bar F_Y(y)\big|^{2(1-2b_2)}\mathds 1_{\big\{\widehat F^{(j)}_Y(y) \bar{\widehat F}^{(j)}_Y(y) \neq 0\big\}}\right] \leq \frac{2^{1/2-b_2}\big(F_Y(y) \bar F_Y(y)\big)^{1/2}}{n_1^{1/2-2b_2}}\land \left(\frac{2F_Y(y)\bar F_Y(y)}{n_1}\right)^{1/2-b_2}.
    \end{align*}
\end{lemma}
\textbf{Proof:} To ease notation, we drop the superscript $(j)$ and simply write $\widehat F_Y$. We define the weights $u = 1/(1-2b_1)$ and $v = 1/(2b_1)$ that satisfy $\frac{1}{u} + \frac{1}{v} = 1$ and $u,v \geq 1$, using the convention that $1/0 = \infty$ if $b_1 = 1/2$. 
By the Hölder inequality, we have
$$\begin{aligned}
    & \hspace{5.5mm} E^{1/2}\left[\big|\widehat F_Y(y) - F_Y(y)\big|^{2(1-2b_1)} \mathds 1_{\big\{\widehat F_Y(y) \bar{\widehat F}_Y(y) \neq 0\big\}}\right] \\[5pt]
    & \leq E^{1/(2u)}\left[\big|\widehat F_Y(y) - F_Y(y)\big|^{2}\right] E^{1/(2v)}\left[\Big( \mathds 1_{\big\{\widehat F_Y(y) \bar{\widehat F}_Y(y) \neq 0\big\}}\Big)^v\right] \\[5pt]
    & \leq \left(\frac{2F_Y(y)\bar F_Y(y)}{n_1}\right)^{\frac{1-2b_1}{2}} \left(\Pr\left(\widehat F_Y(y) \neq 0\right) \land \Pr\big(\bar{\widehat F}_Y(y) \neq 0\big)\right)^{b_1}\\[5pt]
    & \leq \left(\frac{2F_Y(y)\bar F_Y(y)}{n_1}\right)^{\frac{1-2b_1}{2}}\left\{\left(1 \land \left(\frac{n_1  F_Y(y)}{2}\right)\right)^{b_1} \land \left(1 \land \left(\frac{n_1 \bar F_Y(y)}{2}\right)\right)^{b_1}\right\}\\[5pt]
    & = \left(\frac{2F_Y(y)\bar F_Y(y)}{n_1}\right)^{\frac{1-2b_1}{2}}n_1^{b_1}\left\{\frac{F_Y(y) \land \bar F_Y(y)}{2} \land \frac{1}{n_1}\right\}^{b_1} \\[5pt]
    & \leq \left(\frac{2F_Y(y)\bar F_Y(y)}{n_1}\right)^{\frac{1-2b_1}{2}}n_1^{b_1}\left\{F_Y(y) \bar F_Y(y) \land \frac{1}{n_1}\right\}^{b_1} \qquad \text{ using  $u \land \bar u \leq 2u\bar u, \, \forall u \in [0,1]$}\\[5pt]
    & = \frac{2^{1/2-b_1}\big(F_Y(y) \bar F_Y(y)\big)^{1/2}}{n_1^{1/2-2b_1}} \land \left(\frac{2F_Y(y)\bar F_Y(y)}{n_1}\right)^{1/2-b_1}, 
\end{aligned}$$
which proves the first claim.
The second claim can be proved in a similar fashion,  following the same steps as outlined above.
\hfill $\square$

\begin{lemma}\label{lem_decomposition_Omega-hat_Omega}
    For any $x,y \in \R$ and $\alpha \in (0,1]$, it holds that
    \begin{align*}
        \Big|\big(\widehat F_Y^{(1)}(x)\land \widehat F_Y^{(2)}(y)\big)^\alpha - \big(F_Y(x) \land F_Y(y)\big)^\alpha\Big| & \leq \hspace{3.5mm} \Big|\widehat F_Y^{(1)}(x) - F_Y(x) \Big|^{\alpha} \land F_Y(y)^{\alpha} \\
        & \quad + \Big|\widehat F_Y^{(2)}(y) - F_Y(y) \Big|^{\alpha} \land F_Y(x)^{\alpha}\\[5pt]
        & \quad + \Big|\widehat F_Y^{(1)}(x) - F_Y(x) \Big|^{\alpha} \land \Big|\widehat F_Y^{(2)}(y) - F_Y(y)\Big|^{\alpha},
    \end{align*}
    and the same inequality holds true by replacing $\widehat F_Y^{(1)}$, $\widehat F_Y^{(2)}$, and $F_Y$ with $\bar{\widehat F}_Y^{(1)}$, $\bar{\widehat F}_Y^{(2)}$, and $\bar F_Y$, respectively.
\end{lemma}
\textbf{Proof:}
 Below, we will use the classical inequalities that hold true for any $a,b,c\geq 0$ and $\alpha \in (0,1]$ 
\begin{align*}
    &|a \land b - a \land c| \leq a \land |b-c|\\
    &|a^\alpha - b^\alpha| \leq |a-b|^\alpha.
\end{align*}
By the triangle inequality, we have
\begin{align*}
    \big|\big(\widehat a\land \widehat b\big)^\alpha - \big(a \land b\big)^\alpha\big| & \leq \left(\big|\big(\widehat a \land \widehat b\big)^\alpha - \big(a \land \widehat b\big)^\alpha\big|\right) + \left(\big|\big(a \land \widehat b\big)^\alpha - \big(a \land b\big)^\alpha\big|\right)\\[5pt]
    & \leq \left(\big|\widehat a^\alpha - a^\alpha \big|\land \widehat b^{\alpha}\right) + \left( a^{\alpha} \land \big|\widehat b^\alpha - b^\alpha \big|\right)\\[5pt]
    & \leq \left(\big|\widehat a - a \big|^{\alpha}\land \widehat b^{\alpha}\right) +  \left(a^{\alpha} \land \big|\widehat b - b \big|^{\alpha}\right)\\[5pt]
    & \leq \left(\big|\widehat a - a \big|^{\alpha}\land b^{\alpha}\right) + \left(a^{\alpha} \land \big|\widehat b - b \big|^{\alpha}\right)\\[5pt]
    & \quad + \big|\widehat a - a \big|^{\alpha}\land \big|\widehat b - b\big|^{\alpha}.
\end{align*}
Applying this with $a = F_Y(x)$, $\hat a = \widehat F_Y^{(1)}(x)$, $b = F_Y(y)$, and $\hat b = \widehat F_Y^{(2)}(y)$ yields the result.
\hfill $\square$

\begin{lemma}\label{lem_assumptions_DCT_J2N}
    For any $x,y \in \R$ and $\alpha \in (0,1]$, it holds that 
    \begin{align*}
        E \left[\Big|\big(\widehat F_Y^{(1)}(x)\land \widehat F_Y^{(2)}(y)\big)^\alpha - \big(F_Y(x) \land F_Y(y)\big)^\alpha\Big|\right] \leq 6 \big(F_Y(x) \land F_Y(y)\big)^{\alpha}
    \end{align*}
    and $E^{1/2} \left[\Big|\big(\widehat F_Y^{(1)}(x)\land \widehat F_Y^{(2)}(y)\big)^\alpha - \big(F_Y(x) \land F_Y(y)\big)^\alpha\Big|\right] \to 0$ as  $N \to \infty $.
\end{lemma}
\textbf{Proof:}
By Lemma~\ref{lem_decomposition_Omega-hat_Omega}, we have
\begin{align*}
    E\left[\Big|\big(\widehat F_Y^{(1)}(x)\land \widehat F_Y^{(2)}(y)\big)^\alpha - \big(F_Y(x) \land F_Y(y)\big)^\alpha\Big|\right] & \leq ~~E\left[\Big|\widehat F_Y^{(1)}(x) - F_Y(x) \Big|^{\alpha}\right] \land E\left[\Big|\widehat F_Y^{(2)}(y) - F_Y(y)\Big|^{\alpha}\right]\\[5pt]
        & \quad + E\left[ \Big|\widehat F_Y^{(1)}(x) - F_Y(x) \Big|^{\alpha}\right] \land F_Y(y)^{\alpha} \\[5pt]
        & \quad + E \left[ \Big|\widehat F_Y^{(2)}(y) - F_Y(y) \Big|^{\alpha}\right] \land F_Y(x)^{\alpha}.\numberthis\label{eq_lemma_m_alpha}
\end{align*}
We now control the quantity $E\left[\big|\widehat F_Y^{(1)}(x) - F_Y(x) \big|^{\alpha}\right]$ for any fixed $x \in \R$. 
We define the weights $v = 2/\alpha \geq 1$ and $u = \frac{v}{v-1}$. 
These quantities satisfy $u \leq 2/\alpha$ since $v-1 \geq 1$, and $\frac1u + \frac1v = 1$.
By H\"older's inequality, we therefore have
$$\begin{aligned}
    E \left[\Big|\widehat F_Y^{(1)}(x) - F_Y(x) \Big|^{\alpha}\right] &= E \left[\big|\widehat F_Y^{(1)}(x) - F_Y(x) \big|^{\alpha} \mathds{1}_{\widehat F_Y^{(1)}(x) \neq 0} + F_Y(x)^{\alpha} \, \mathds 1_{\widehat F_Y^{(1)}(x) = 0}\right] \\[5pt]
    & \leq E^{1/u}\!\left[\Big|\widehat F_Y^{(1)}(x) - F_Y(x) \Big|^{\alpha u}\right] \Pr\!\left(\mathds{1}_{\widehat F_Y^{(1)}(x) \neq 0}\right)^{1/v} + F_Y(x)^{\alpha}\Pr\!\left(\mathds{1}_{\widehat F_Y^{(1)}(x) = 0}\right)\\[5pt]
    & \leq E^{\alpha/2}\!\left[\big|\widehat F_Y^{(1)}(x) - F_Y(x) \big|^{2}\right] \left(1 - \bar F_Y(x)^{n_1/2}\right)^{1/v} + F_Y(x)^{\alpha} \bar F_Y(x)^{n_1/2}\\
    & \leq \left(\frac{2F_Y(x)\bar F_Y(x)}{n_1}\right)^{\alpha/2}\left(1 \land \big(\frac{n_1}{2}F_Y(x)\big)\right)^{1/v} + F_Y(x)^{\alpha} \bar F_Y(x)^{n_1/2}\\[5pt]
    & \leq \left(\frac{2F_Y(x)}{n_1}\right)^{\alpha/2}\land F_Y(x)^{\alpha}+ F_Y(x)^{\alpha} \bar F_Y(x)^{n_1/2},
\end{aligned}$$
where we used Jensen's inequality in the third line (which can be applied here since $\alpha u/2\leq 1$) and the inequality $1-(1-a)^n \leq 1 \land (na)$ for any $a \in [0,1]$ and $n \in \N$ in the fourth line. 
It follows from the display above that $E \left[\Big|\widehat F_Y^{(1)}(x) - F_Y(x) \Big|^{\alpha}\right] \leq 2F_Y(x)^\alpha$ and that $E \left[\Big|\widehat F_Y^{(1)}(x) - F_Y(x) \Big|^{\alpha}\right] \to 0$ as $n_1 \to \infty$ for any fixed $x$. 
By swapping the roles of $x$ and $y$, we similarly obtain that $E \left[\Big|\widehat F_Y^{(2)}(y) - F_Y(y) \Big|^{\alpha}\right] \leq 2F_Y(y)^\alpha$ and that $E \left[\Big|\widehat F_Y^{(2)}(y) - F_Y(y) \Big|^{\alpha}\right] \to 0$ as $n_1 \to \infty$ for any fixed $y$.
Combining with~\eqref{eq_lemma_m_alpha} yields that
\begin{align*}
    E \left[\Big|\big(\widehat F_Y^{(1)}(x) \land \widehat F_Y^{(2)}(y)\big)^\alpha - \big(F_Y(x) \land F_Y(y)\big)^\alpha \Big|\right] \leq 6 \big(F_Y(x) \land F_Y(y)\big)^\alpha
\end{align*}
and that $E \left[\big|\big(\widehat F_Y^{(1)}(x)\land\widehat F_Y^{(2)}(y)\big)^\alpha - \big(F_Y(x) \land F_Y(y)\big)^\alpha\big|\right] \to 0$ as  $N \to \infty$ for any fixed $x,y \in \R$, which concludes the proof.
\hfill $\square$

\begin{lemma}\label{lem:spacing_4_mom}
~
\begin{enumerate}
    \item For all $a,b\in\{10,\ldots,n_1/2-10\}$ such that $a<b$,
    \begin{align*}
    E\left[\left(Y_{(b)}^{(1)}-Y_{(a)}^{(1)}\right)^2\right] \lesssim \left[\left(\frac{a}{n_1}\right)^{-(1+d_1)} \left(\frac{n_1-b}{n_1}\right)^{-(1+d_2)}\left(\frac{b-a}{n_1}\right)\right]^2.
    \end{align*}
    \item For all $k\in\{1,\ldots,10\}\cup\{n_1/2-11,\ldots,n_1/2-1\}$:
    \[
    E\left[\left(Y_{(k+1)}^{(1)}-Y_{(k)}^{(1)}\right)^2\right] \lesssim n_1^{2(\mathds1\{k\leq 10\}d_1+\mathds1\{k\geq n_1/2-11\}d_2)}.
    \]
    \end{enumerate}
\end{lemma}
\textbf{Proof:} For ease of notation we drop the superscript $(1)$ and write $Y_{(a)},Y_{(b)}$ instead of $Y_{(a)}^{(1)}, Y_{(b)}^{(1)}$.

\medskip
\noindent 1. Suppose that $a \geq 10$ and $b \leq n_1/2-10$ such that $a<b$. 
    \begin{align*}
       E\left[\left(Y_{(b)} - Y_{(a)}\right) ^2\right] &= E\left[\Big(F_Y^{-1}\left(\xi_{(b)}\right) - F_Y^{-1}\left(\xi_{(a)}\right)\Big)^2\right]\\[5pt]
       &=  E\left[\left((F_Y^{-1})'(\tilde \xi)\right)^2  \left(\xi_{(b)} - \xi_{(a)}\right)^2\right]\\[5pt]
       & \leq E\left[\tilde \xi^{-2(1+d_1)}(1-\tilde \xi)^{-2(1+d_2)}  \left(\xi_{(b)} - \xi_{(a)}\right)^2 \left(\ind{\tilde \xi \geq 1/2} + \ind{\tilde \xi \leq 1/2}\right)\right],
   \end{align*}
   for some $\tilde\xi\in(\xi_{(a)},\xi_{(b)})$.
   We bound $E\left[\tilde \xi^{-2(1+d_1)}(1-\tilde \xi)^{-2(1+d_2)}   \left(\xi_{(b)} - \xi_{(a)}\right)^2 \ind{\tilde \xi \leq 1/2}\right]$, and the second term in the sum can be analyzed in a similar way. 
We have
\begin{align}
    &E\left[\tilde \xi^{-2(1+d_1)}(1-\tilde \xi)^{-2(1+d_2)}  \left(\xi_{(b)} - \xi_{(a)}\right)^2 \ind{\tilde \xi \leq 1/2}\right] \nonumber\\[5pt]
    & \leq 2^{2(1+d_2)} E\left[\left( \xi_{(a)}\right)^{-2(1+d_1)}\big(\xi_{(b)} - \xi_{(a)}\big)^2 \right] \nonumber\\[5pt]
    & \leq 2^{2d}E^{1/2}\left[( \xi_{(a)})^{-4(1+d_1)}\right]E^{1/2}\left[\big(\xi_{(b)} - \xi_{(a)}\big)^4\right]. \label{eq_Cauchy_Schwarz_order_stats}
\end{align}
We compute the two expectations in the right-hand side separately. 
We recall that the density of $\xi_{(a)}$ at point $u\in(0,1)$ is given by
\begin{align*}
    f_{\xi_{(a)}}(u) = \frac{n_1!}{(a-1)! (n_1-a)!}u^{a-1}(1-u)^{n_1-a}.
\end{align*}
Therefore, since $a \geq 10>4(1+d_1)$, we obtain
$$
\begin{aligned}
    E\left[\xi_{(a)}^{-4(1+d_1)}\right] &= \int_{(0,1)} \frac{n_1!}{(a-1)! (n_1-a)!}u^{a-1-4(1+d_1)}(1-u)^{n_1-a}du\\[5pt]
    & = \frac{n_1!}{(a-1)! (n_1-a)!} {\rm B}(a-4(1+d_1), n_1-a+1)\\
    & = \frac{\Gamma(n_1+1)}{\Gamma(a) \Gamma(n_1-a + 1)} \frac{\Gamma(a-4d) \Gamma(n_1-a+1)}{\Gamma(n_1-4(1+d_1)+1)}\\[5pt]
    & = \frac{\Gamma(n_1+1)}{\Gamma(n_1-4(1+d_1)+1)} \frac{\Gamma(a-4(1+d_1))}{\Gamma(a)}\\[5pt]
    & \lesssim n_1^{4(1+d_1)}a^{-4(1+d_1)} \\[5pt]
    &= \left(\frac{a}{n_1}\right)^{-4(1+d_1)}, 
\end{aligned}
$$
where we used Lemma~\ref{lem_comparison_Gamma} to obtain the inequality. 
Moreover, 
\begin{align*}
    E\left[\big(\xi_{(b)} - \xi_{(a)}\big)^4\right]
    & =E\left[{\rm Beta}\big(b-a,n_1-b+a+1\big)^4\right]\\[5pt]
    & =\prod_{r=0}^3\frac{b-a+r}{n_1+d+r}\\[5pt]
    &\asymp \left(\frac{b-a}{n_1}\right)^4.
\end{align*}
Combining with~\eqref{eq_Cauchy_Schwarz_order_stats}, we obtain
\begin{align*}
    E\left[\tilde \xi^{-2(1+d_1)}(1-\tilde \xi)^{-2(1+d_2)}  \left(\xi_{(b)} - \xi_{(a)}\right)^2 \ind{\tilde \xi \leq 1/2}\right] \lesssim \left(\frac{a}{n_1}\right)^{-2(1+d_1)} \left(\frac{b - a}{n_1}\right)^2,
\end{align*}
and similarly,
\begin{align*}
   E\left[\tilde \xi^{-2(1+d_1)}(1-\tilde \xi)^{-2(1+d_2)}  \left(\xi_{(b)} - \xi_{(a)}\right)^2 \ind{\tilde \xi \geq 1/2}\right] \lesssim \left(\frac{n_1-b}{n_1}\right)^{-2(1+d_2)} \left(\frac{b - a}{n_1}\right)^2.
\end{align*}
Therefore, we obtain
\begin{align*}
    E\left[\left(Y_{(b)} - Y_{(a)}\right)^2\right] &\lesssim \left[\left(\frac{a}{n_1}\right)^{-2(1+d_1)} + \left(\frac{n_1-b}{n_1}\right)^{-2(1+d_2)}\right]\left(\frac{b-a}{n_1}\right)^2\\
    & \asymp \left[\left(\frac{a}{n_1}\right)^{-(1+d_1)}\left(\frac{n_1-b}{n_1}\right)^{-(1+d_2)}\left(\frac{b-a}{n_1}\right)\right]^2.
\end{align*}
This completes the case $a \geq 10$ and $b \leq n_1/2-10$. 

\medskip
\noindent 2. Fix $k\leq10$. We recall that, for any $u,v \in \R$, the joint density of  $(Y_{(k)}, Y_{(k+1)})$ is 
\begin{align*}
    f_{Y_{(k)}, Y_{(k+1)}}(u,v) = f_Y(u) f_Y(v) F_Y(u)^{k-1} \bar F_Y(v)^{n_1-k-1} \frac{n_1!}{(k-1)!(n_1-k-1)!} \ind{u<v}
\end{align*}
\citep[see, e.g.][Theorem 5.4.6]{CaseBerg:01}. 
We use the change of variables $s = F_Y(u)$ and $t = F_Y(v)$ and the notation $d \vec x = dx_1 dx_2$. We obtain
\begin{align*}
   &\quad  E\left[\big(Y_{(k+1)} - Y_{(k)}\big)^2\right] \\[5pt]
   &= \int_{\R^2} (v-u)^2 f_Y(u) f_Y(v) F_Y(u)^{k-1} \bar F_Y(v)^{n_1-k-1} \frac{n_1! \, \ind{u<v}}{(k-1)!(n_1-k-1)!}  dudv\\
    & = \int_{(0,1)^2} \left(F_Y^{-1}(t) - F_Y^{-1}(s)\right)^2 s^{k-1} (1-t)^{n_1-k-1} \frac{n_1! \, \ind{s<t}}{(k-1)!(n_1-k-1)!} dsdt\\[5pt]
    & = \int_{(0,1)^2} \left(\int_{s}^t (F_Y^{-1})'(x) dx\right)^2 s^{k-1} (1-t)^{n_1-k-1} \frac{n_1! \, \ind{s<t}}{(k-1)!(n_1-k-1)!} dsdt\\[5pt]
    & \leq \int_{(0,1)^2} \left(\int_{s}^t \big(x^{-(1+d_1)}(1-x)^{-(1+d_2)}\big) dx\right)^2 s^{k-1} (1-t)^{n_1-k-1} \frac{n_1! \, \ind{s<t}}{(k-1)!(n_1-k-1)!} dsdt\\[5pt]
    & = \int_{(0,1)^2} \left(\prod_{i=1}^2\int_{s}^t  \big(x_i^{-(1+d_1)}(1-x_i)^{-(1+d_2)}\big) dx_i \right) s^{k-1} (1-t)^{n_1-k-1} \frac{n_1! \, \ind{s<t}}{(k-1)!(n_1-k-1)!} dsdt\\
    & = \int_{(0,1)^4} s^{k-1} \ind{s<\min_{j} x_j} (1-t)^{n_1-k-1} \ind{t<\max_j x_j}\\[5pt]
    &\hspace{2cm} \times \prod_{i=1}^2 \big(x_i^{-(1+d_1)}(1-x_i)^{-(1+d_2)}\big)
     \frac{n_1!}{(k-1)!(n_1-k-1)!} ds \, dt \, d \vec x\\[5pt]
    & = \int_{(0,1)^2} \frac{\min_j x_j^k}{k} \cdot \frac{(1-\max_j x_j)^{n_1-k}}{n_1-k}\prod_{i=1}^2 \big(x_i^{-(1+d_1)}(1-x_i)^{-(1+d_2)}\big) 
     \frac{n_1!}{(k-1)!(n_1-k-1)!}  d\vec{x}\\[5pt]
   & \leq \int_{(0,1)^2}\prod_{i=1}^2 x_i^{k/2 - (1+d_1)} (1-x_i)^{(n_1-k)/2-(1+d_2)}   {n_1 \choose k} d\vec{x}\\[5pt]
     & =  {n_1 \choose k} \left(\int_{0}^1 x^{k/2 - (1+d_1)} (1-x)^{(n_1-k)/2-(1+d_2)}  dx\right)^2\\[5pt]
     & = {n_1 \choose k}\left[{\rm B}\left(\frac{k}{2}-d_1, \frac{n_1-k}{2}-d_2\right)\right]^2\\[5pt]
     & = \frac{\Gamma(n_1+1)}{\Gamma(k+1) \Gamma(n_1-k+1)} \left[\frac{\Gamma(\frac{k}{2}-d_1) \Gamma(\frac{n_1-k}{2} -d_2)}{ \Gamma(\frac{n_1}{2} - d_1-d_2)}\right]^2.
\end{align*}
Now, note that $\Gamma(\frac{k}{2}-d_1)$ is a constant, and 
\begin{align*}
    \frac{\Gamma(n_1+1)}{\Gamma(k+1) \Gamma(n_1-k+1)} \left[\frac{\Gamma(\frac{k}{2}-d_1) \Gamma(\frac{n_1-k}{2}-d_2)}{ \Gamma(\frac{n_1}{2}-d_1-d_2)}\right]^2 &\asymp n_1^k\left(n_1^{d_1 - k/2}\right)^2 = n_1^{2d_1}.
\end{align*} 
This concludes the proof.
\hfill$\square$

\begin{lemma}\label{lem_comparison_Gamma}
    Let $\eps>0$. For all $a\in(-10,10)$ and $x\in(0,\infty)$ such that $x+a>\eps$, it holds that
    \begin{align*}
        &\frac{\Gamma(x+a)}{\Gamma(x)} \leq (x+\lfloor a \rfloor)^a \quad \text{ if } a>0,\\
        & \frac{\Gamma(x+a)}{\Gamma(x)} \leq \frac{(x-\lfloor\abs{a}\rfloor)^a}{\left(1+\frac{1}{1+\eps}\right)^{\lfloor\abs{a}\rfloor+1}} \quad \text{ if } a<0.
    \end{align*}
\end{lemma}
\textbf{Proof:}
First, suppose that $0<a<1$. Then, for all $x>0$,
\[
   \frac{\Gamma(x+a)}{\Gamma(x)} \leq x^a\leq  (x+\lfloor a \rfloor)^a,
\]
where the first inequality follows from the right part of \cite{Wendel1948}'s double inequality.  \\
Second, suppose that $a\geq 1$. Let $u\in(0,1)$ and define $b=a/u$. We have
\begin{align*}
     \frac{\Gamma(x+a)}{\Gamma(x)}&=\frac{\Gamma(x+u)}{\Gamma(x)}\cdot \frac{\Gamma(x+2u)}{\Gamma(x+u)}\cdots\frac{\Gamma(x+\lfloor b\rfloor u)}{\Gamma(x+(\lfloor b\rfloor-1)u)}\cdot\frac{\Gamma(x+bu)}{\Gamma(x+\lfloor b\rfloor u)}.
\end{align*}
By applying \cite{Wendel1948}'s inequality to each term in the product, we obtain
\begin{align*}
     \frac{\Gamma(x+a)}{\Gamma(x)}&\leq x^u(x+u)^u\cdots(x+(\lfloor b\rfloor-1)u)^u\cdot(x+\lfloor b\rfloor u)^{u(b-\lfloor b\rfloor)} \\[5pt]
     &\leq (x+(\lfloor b\rfloor-1)u)^{u\lfloor b\rfloor}(x+\lfloor b\rfloor u)^{u(b-\lfloor b\rfloor)} \\[5pt]
    &\leq (x+\lfloor b\rfloor u)^{u\lfloor b\rfloor}(x+\lfloor b\rfloor u)^{u(b-\lfloor b\rfloor)} \\[5pt]
     &= (x+\lfloor b\rfloor u)^{a}.
\end{align*}
Since $b\mapsto \lfloor b\rfloor$ is right-continuous, by taking the limit as $u\to 1$ we obtain
\begin{align*}
     \frac{\Gamma(x+a)}{\Gamma(x)}&\leq (x+\lfloor a \rfloor)^{a}.
\end{align*}
Third, suppose that $-1<a<0$ and let $\tilde x=x+a$. Then,
\begin{align}
    \frac{\Gamma(x+a)}{\Gamma(x)} = \frac{\Gamma(\tilde x)}{\Gamma(\tilde x+(-a))}&\leq\tilde x^{a}\left(\frac{\tilde x+(-a)}{\tilde x}\right)^{1-(-a)} \notag \\[5pt]
    &=x^{a}\cdot\frac{x}{x+a} \notag \\[5pt]
    &\leq \frac{x^a}{1+\frac{a}{\eps-a}} \notag \\[5pt]
    &\leq \frac{x^a}{1-\frac{1}{1+\eps}}, \label{eq:neg_ineq}
\end{align}
where the first inequality follows from the left part of \cite{Wendel1948}'s double inequality since $0<(-a)<1$ and $\tilde x>0$, the second inequality follows from $x+a>\eps$, and the third inequality follows from $a>-1$.\\
Fourth, suppose that $a\leq-1$ and let $u\in(-1,0)$ such that $b=a/u$. We have
\begin{align*}
     \frac{\Gamma(x+a)}{\Gamma(x)}&=\frac{\Gamma(x+u)}{\Gamma(x)}\cdot \frac{\Gamma(x+2u)}{\Gamma(x+u)}\cdots\frac{\Gamma(x+\lfloor b\rfloor u)}{\Gamma(x+(\lfloor b\rfloor-1)u)}\cdot\frac{\Gamma(x+bu)}{\Gamma(x+\lfloor b\rfloor u)}.
\end{align*}
By applying inequality~\eqref{eq:neg_ineq} to each term in the product, we obtain
\begin{align*}
     \frac{\Gamma(x+a)}{\Gamma(x)}&\leq \frac{1}{\left(1-\frac{1}{1+\eps}\right)^{\lfloor b\rfloor+1}}x^u(x+u)^u\cdots(x+(\lfloor b\rfloor -1)u)^u(x+\lfloor b\rfloor u)^{u(b-\lfloor b\rfloor)} \\[5pt]
     &\leq \frac{1}{\left(1-\frac{1}{1+\eps}\right)^{\lfloor b\rfloor+1}}(x+\lfloor b\rfloor u)^{a}.
\end{align*}
Since $b\mapsto \lfloor b\rfloor$ is right-continuous, by taking the limit as $u\to -1$ we obtain
\begin{align*}
     \frac{\Gamma(x+a)}{\Gamma(x)}&\leq\frac{1}{\left(1-\frac{1}{\eps+1}\right)^{\lfloor\abs{a}\rfloor+1}}(x-\lfloor \abs{a}\rfloor)^{a}.
\end{align*}
\hfill $\square$

\begin{lemma}\label{lem_control_s_bar_s_t_bar_t}
    Let $s,t \in [0,1]$ and assume that $|s-t| \leq \varepsilon (s\bar s + t\bar t)$ for some $\varepsilon \leq 1/2$. 
    Then $t \hspace{.4mm} \bar t \leq 6 \, s\bar s$, and $s\bar s 
 \leq 6 \,  t \hspace{.4mm} \bar t$.
\end{lemma}
\textbf{Proof:} The second claim follows from the first one by symmetry of the roles of $s$ and $t$. 
We therefore focus on proving the first claim. 

Without loss of generality, assume $s \leq 1/2$ (otherwise, it suffices to repeat the following argument with $\bar s$ and $\bar t$ instead of $s$ and $t$, respectively). 
If $t \leq s$, then $t \hspace{.5mm} \bar t\leq s \bar s$ and the result follows.
Assume from now on that $t>s$. 
By assumption, we have
$$
\begin{aligned}
    t - s \leq \varepsilon (s\bar s + t \hspace{.5mm} \bar t) \leq \varepsilon (s+t)  \qquad  \text{ i.e. } \qquad  t \leq \frac{1+\varepsilon}{1-\varepsilon}s,
\end{aligned}
$$
by rearranging the terms.
Hence 
$$\begin{aligned}
    t \hspace{.5mm} \bar t \leq t \leq \frac{1+\varepsilon}{1-\varepsilon}s \leq \frac{1+\varepsilon}{1-\varepsilon}s \cdot 2 \bar s \leq 6 s \bar s.
\end{aligned}$$
This completes the proof.
\hfill$\square$

\newpage
\bibliography{biblio}